\documentclass{llncs}
\usepackage{latexsym}

\newcommand{\keywords}[1]{\noindent {\bf Keywords:}\ #1.}
\begin{document}

\title{Conceptua: \\ Institutions in a Topos}
\author{Robert E. Kent\thanks{creator of the Information Flow Framework ({\ttfamily\footnotesize http://suo.ieee.org/IFF/})}}
\institute{Ontologos \\ 550 SW Staley Dr, Pullman WA 99163, USA\\
\email{rekent@ontologos.org}}

\maketitle

\begin{abstract}
Tarski's semantic definition of truth is 
the composition of its extensional and intensional aspects.
Abstract satisfaction,
the core of the semantic definition of truth,
is the basis for the theory of institutions 
\cite{goguen:burstall:92}.
The satisfaction relation for first order languages (the truth classification),
and the preservation of truth by first order interpretations (the truth infomorphism),
form a key motivating example in the theory of Information Flow (IF) \cite{barwise:seligman:97}.
The concept lattice notion, 
which is the central structure studied by the theory of Formal Concept Analysis (FCA) \cite{ganter:wille:99},
is constructed by the polar factorization of derivation. 
The study of classification structures (IF) 
and the study of conceptual structures (FCA)
provide a principled foundation for 
the logical theory of knowledge representation and organization.
In an effort to unify these two areas,
the paper ``Distributed Conceptual Structures'' \cite{kent:02}
abstracted the basic theorem of FCA 
in order to established three levels of categorical equivalence 
between classification structures and conceptual structures.
In this paper, 
we refine this approach by resolving the equivalence 
as the category-theoretic factorization of the Galois connection of derivation.
The equivalence between classification and conceptual structures 
is mediated by the opposite motions of factorization and composition.
Abstract truth factors through the concept lattice of theories
in terms of its extensional and intensional aspects.
\end{abstract}

\keywords{factorization system, order adjunction, classification, concept lattice, institution}

\section{Introduction\label{sec:introduction}}
%%%%%%%%%%%%%%%%%%%%%%%%%%%%%%%%%%%%%%%%%%%%%%%%%%%%%%%%%%%%%%%%%%%%%%%%%%%%%%%%
%%%%%%%%%%%%%%%%%%%%%%%%%%%%%%%%%%%%%%%%%%%%%%%%%%%%%%%%%%%%%%%%%%%%%%%%%%%%%%%%

Human knowledge is made up of the conceptual structures of many communities of interest.
In order to establish coherence in human knowledge representation,
it is important to enable communcation between the conceptual structures of different communities
The conceptual structures of any particular community is representable in an ontology.
Such a ontology provides a formal linguistic standard for that community.
However,
a standard community ontology is established for various purposes,
and makes choices that force a given interpretation, 
while excluding others that may be equally valid for other purposes.
Hence, a given representation is relative to the purpose for that representation.
Due to this relativity of represntation,
in the larger scope of all human knowledge
it is more important to standardize methods and frameworks for relating ontologies
than to standardize any particular choice of ontology.
The standardization of methods and frameworks is called the semantic integration of ontologies.

The minimum framework in which a lattice of theories construction can be defined 
is called a \emph{conceptua}.
A conceptua is a framework for conceptual structures. 
Such a framework is built on a topos-theoretic base.
In the succeeding paper ``The Lattice of Theories Construction'',
the following four related axiomatic structures were developed.
\begin{center}
\begin{tabular}{ccc}
& concrete & abstract \\
\begin{tabular}{c} conceptual structures \\ category \end{tabular}
& {\bfseries cCS} & {\bfseries aCS} \\
&& \\
\begin{tabular}{c} lattice of theories \\ category \end{tabular}
& {\bfseries cLoT} & {\bfseries aLoT} \\ 
\end{tabular}
\end{center}
In particular,
an abstract conceptual structures category,
one satisfying {\bfseries aCS},
is a finitely complete order-enriched category.
In this paper we show that various parts of a conceptua satisfy these axiomatizations.
For example,
$\mathsf{Ord}(\mathcal{B})$,
the category of preorders in the topos $\mathcal{B}$,
is an abstract conceptual structures category.

The entailment ``lattice of theories'' construction (LOT) has been touted in the knowledge representation and ontology communities as a fundamental structuring mechanism for the representation and management of ontologies.
This paper provides a detailed category-theoretic analysis of the concept ``lattice of theories'' construction (cLOT),
a principled first step approximation to LOT.
As shown here,
classification structures and conceptual structures can alternately be defined in three isomorphic versions:
a relation version, a function version and an adjunction version.
In the past,
the relation version has been the default for the definition of classification structures (\cite{barwise:seligman:97}, \cite{ganter:wille:99}).
Also in the past,
the function version has been the default for the definition of conceptual structures (\cite{ganter:wille:99}, \cite{kent:02}).
This paper advocates and develops the adjunction version for both,
since it is more category-theoretic, it is conceptually simpler,
and it needs no extra assumptions\footnotemark[12].
The equivalence between classification and conceptual structures is mediated by two motions.
In any of the three versions,
the motion towards conceptual structures is defined by factorization,
and
the motion towards classification structures is defined by composition.
The adjunction version allows for a very simple and highly category-theoretic explication of these motions 
--- 
the motion of factorization is defined via polar factorization of adjunctions,
and the motion of composition is defined via composition of adjunctions
\footnote{There is always a philosophical and practical question of identity in category-theoretic studies in computer science.
The question is ``Are things identical when they are equal, when they are isomorphic, 
or when they are equivalent?''.
In the category-theory of conceptual structures,
this occurs at both the object and category level.
In this paper, 
we regard objects as being identical when they are isomorphic,
but not when they are only equivalent.
We regard categories as being identical when they are equivalent (or isomorphic).}.

Section~\ref{sec:factorization:systems} proves a general equivalence result (Thm.~\ref{general-equivalence}) between the arrow category and the factorization category of an arbitrary category having a factorization system with choice.
Section~\ref{sec:topos:theory} reviews the elements of topos theory.
(Thm.~\ref{special-equivalence}) to the polar factorization system on the category of preorders and adjunctions in a topos.
Section~\ref{sec:classification:structures} discusses the representation of classification structures in a topos as a derivation Galois connection.
Section~\ref{sec:conceptual:structures} discusses the representation of conceptual structures in a topos as the polar factorization of extension and intension,
developing a restricted equivalence (Thm.~\ref{restricted-equivalence}) between classification structures and conceptual structures.
Section~\ref{sec:institutions} discusses the application of this equivalence to the theory of institutions in a topos,
in particular to the abstract cLOT construction in a topos.

%%%%%%%%%%%%%%%%%%%%%%%%%%%%%%%%%%%%%%%%%%%%%%%%%%%%%%%%%%%%%%%%%%%%%%%%%%%%%%%%
%%%%%%%%%%%%%%%%%%%%%%%%%%%%%%%%%%%%%%%%%%%%%%%%%%%%%%%%%%%%%%%%%%%%%%%%%%%%%%%%

%%%%%%%%%%%%%%%%%%%%%%%%%%%%%%%%%%%%%%%%%%%%%%%%%%%%%%%%%%%%
%%%%%%%%%%%%%%%%%%%%%%%%%%%%%%%%%%%%%%%%%%%%%%%%%%%%%%%%%%%%
\section{Factorization Systems}\label{sec:factorization:systems}
%%%%%%%%%%%%%%%%%%%%%%%%%%%%%%%%%%%%%%%%%%%%%%%%%%%%%%%%%%%%
%%%%%%%%%%%%%%%%%%%%%%%%%%%%%%%%%%%%%%%%%%%%%%%%%%%%%%%%%%%%

\begin{sloppypar}
Let $\mathcal{C}$ be an arbitrary category.
A \emph{factorization system} in $\mathcal{C}$ is a pair 
$\langle \mathcal{E}, \mathcal{M} \rangle$ of classes of $\mathcal{C}$-morphisms satisfying the following conditions. 
{\bfseries Subcategories:}
All $\mathcal{C}$-isomorphisms are in $\mathcal{E} \,\cap\, \mathcal{M}$. 
Both $\mathcal{E}$ and $\mathcal{M}$ are closed under $\mathcal{C}$-composition.
Hence,
$\mathcal{E}$ and $\mathcal{M}$ are $\mathcal{C}$-subcategories with the same objects as $\mathcal{C}$.
{\bfseries Existence:} 
Every $\mathcal{C}$-morphism $f : A \rightarrow B$ has an $\langle \mathcal{E}, \mathcal{M} \rangle$-factorization\footnote{An $\langle \mathcal{E}, \mathcal{M} \rangle$-factorization is a quadruple $(A, e, C, m, B)$ where 
$e : A \rightarrow C$ and $m : C \rightarrow B$ 
is a composable pair of $\mathcal{C}$-morphisms 
with $e \in \mathcal{E}$ and $m \in \mathcal{M}$.};
that is,
there is an $\langle \mathcal{E}, \mathcal{M} \rangle$-factorization $(A, e, C, m, B)$ 
and $f$ is its composition\footnote{In this paper, all compositions are written in diagrammatic form.} $f = e \cdot m$.
{\bfseries Diagonalization:}
For every commutative square $e \cdot s = r \cdot m$ of $\mathcal{C}$-morphisms, with $e \in \mathcal{E}$ and $m \in \mathcal{M}$, 
there is a unique $\mathcal{C}$-morphism $d$ with $e \cdot d = r$ and $d \cdot m = s$\footnote{This diagonalization condition implies the following condition.
{\bfseries Uniqueness:}
Any two $\langle \mathcal{E}, \mathcal{M} \rangle$-factorizations of a $\mathcal{C}$-morphism are isomorphic; that is, if $(A, e, C, m, B)$ and $(A, e^\prime, C^\prime, m^\prime, B)$ are two $\langle \mathcal{E}, \mathcal{M} \rangle$-factorizations of $f : A \rightarrow B$, then there is a unique $\mathcal{C}$-isomorphism $h : C \cong C^\prime$ with $e \cdot h = e^\prime$ and $h \cdot m^\prime = m$.}.
%An \emph{epi-mono factorization system} is one 
%where $\mathcal{E}$ is contained in the class of $\mathcal{C}$-epimorphisms
%and $\mathcal{M}$ is contained in the class of $\mathcal{C}$-monomorphisms.
%For an epi-mono factorization system,
%the diagonalization condition is equivalent to the uniqueness condition.
\end{sloppypar}

Let $\mathcal{C}^{\mathsf{2}}$ denote the arrow category\footnote{Recall that $\mathsf{2}$ is the two-object category, pictured as $\bullet \rightarrow \bullet$, with one non-trivial morphism. The arrow category $\mathcal{C}^{\mathsf{2}}$ is (isomorphic to) the functor category $[\mathsf{2}, \mathcal{C}]$.} of $\mathcal{C}$.
An object of $\mathcal{C}^{\mathsf{2}}$ is a triple $(A, f, B)$,
where $f : A \rightarrow B$ is a $\mathcal{C}$-morphism.
A morphism of $\mathcal{C}^{\mathsf{2}}$,
$(a, b) : (A_1, f_1, B_1) \rightarrow (A_2, f_2, B_2)$,
is a pair of $\mathcal{C}$-morphisms $a : A_1 \rightarrow A_2$ and $b : B_1 \rightarrow B_2$ that form a commuting square $a \cdot f_2 = f_1 \cdot b$.
There are source and target projection functors
$\partial_0^{\mathcal{C}}, \partial_1^{\mathcal{C}} : \mathcal{C}^{\mathsf{2}} \rightarrow \mathcal{C}$
and an arrow natural transformation
$\alpha_{\mathcal{C}} : \partial_0^{\mathcal{C}} \Rightarrow \partial_1^{\mathcal{C}} : \mathcal{C}^{\mathsf{2}} \rightarrow \mathcal{C}$
with component
$\alpha_{\mathcal{C}}(A, f, B) = f : A \rightarrow B$
(background of Fig.~\ref{factorization-equivalence}).
Let $\mathcal{E}^{\mathsf{2}}$ denote the full subcategory of $\mathcal{C}^{\mathsf{2}}$
whose objects are the morphisms in $\mathcal{E}$.
Make the same definitions for $\mathcal{M}^{\mathsf{2}}$.
Just as for $\mathcal{C}^{\mathsf{2}}$,
the category $\mathcal{E}^{\mathsf{2}}$ has source and target projection functors
$\partial_0^\mathcal{E}, \partial_1^\mathcal{E} : \mathcal{E}^{\mathsf{2}} \rightarrow \mathcal{C}$
and an arrow natural transformation
$\alpha_{\mathcal{E}} : \partial_0^\mathcal{E} \Rightarrow \partial_1^\mathcal{E} : \mathcal{E}^{\mathsf{2}} \rightarrow \mathcal{C}$
(foreground of Fig.~\ref{factorization-equivalence}).
The same is true for $\mathcal{M}^{\mathsf{2}}$.
Let $\mathcal{E} \odot \mathcal{M}$ denote the category
of $\langle \mathcal{E}, \mathcal{M} \rangle$-factorizations
(top foreground of Fig.~\ref{factorization-equivalence}),
whose objects are $\langle \mathcal{E}, \mathcal{M} \rangle$-factorizations $(A, e, C, m, B)$,
and whose morphisms $(a, c, b) : (A_1, e_1, C_1, m_1, B_1) \rightarrow (A_2, e_2, C_2, m_2, B_2)$
are $\mathcal{C}$-morphism triples
where $(a, c) : (A_1, e_1, C_1) \rightarrow (A_2, e_2, C_2)$ 
is an $\mathcal{E}^{\mathsf{2}}$-morphism
and $(c, b) : (C_1, m_1, B_1) \rightarrow (C_2, m_2, B_2)$ 
is an $\mathcal{M}^{\mathsf{2}}$-morphism.
$\mathcal{E} \odot \mathcal{M}
= \mathcal{E}^{\mathsf{2}} \times_{\mathcal{C}} \mathcal{M}^{\mathsf{2}}$
is the pullback (in the category of categories)
of the $1^{\mathrm{st}}$-projection of $\mathcal{E}^{\mathsf{2}}$
and the $0^{\mathrm{th}}$-projection of $\mathcal{M}^{\mathsf{2}}$.
There is a composition functor
$\circ_{\mathcal{C}} : \mathcal{E} \odot \mathcal{M} \rightarrow \mathcal{C}^{\mathsf{2}}$
that commutes with projections:
on objects $\circ_{\mathcal{C}}(A, e, C, m, B) = (A, e \circ_{\mathcal{C}} m, B)$,
and on morphisms $\circ_{\mathcal{C}}(a, c, b) = (a, b)$.

An $\langle \mathcal{E}, \mathcal{M} \rangle$-factorization system with choice has a specified factorization for each $\mathcal{C}$-morphism;
that is, there is a choice function from the class of $\mathcal{C}$-morphisms to the class of $\langle \mathcal{E}, \mathcal{M} \rangle$-factorizations mapping each $\mathcal{C}$-morphism to one of its factorizations.
With this choice,
diagonalization is uniquely determined.
When choice is specified,
there is a factorization functor
$\div_{\mathcal{C}} : \mathcal{C}^{\mathsf{2}} \rightarrow \mathcal{E} \odot \mathcal{M}$,
which is defined on objects as the chosen $\langle \mathcal{E}, \mathcal{M} \rangle$-factorization $\div_{\mathcal{C}}(A, f, B) = (A, e, C, m, B)$
and on morphisms as
$\div_{\mathcal{C}}(a, b) = (a, c, b)$
where $c$ is defined by diagonalization
($\div_{\mathcal{C}}$ is functorial by uniqueness of diagonalization).
Clearly,
factorization followed by composition is the identity
$\div_{\mathcal{C}} \circ\, \circ_{\mathcal{C}} 
= \mathsf{id}_{\mathcal{C}}$.
By uniqueness of factorization (up to isomorphism)
composition followed by factorization is an isomorphism
$\circ_{\mathcal{C}} \,\circ \div_{\mathcal{C}} 
\cong \mathsf{id}_{\mathcal{E} \odot \mathcal{M}}$.

\begin{theorem} [General Equivalence] \label{general-equivalence}
When a category $\mathcal{C}$ has an $\langle \mathcal{E}, \mathcal{M} \rangle$-factorization system with choice,
the $\mathcal{C}$-arrow category
is equivalent (Fig.~\ref{factorization-equivalence}) to
the $\langle \mathcal{E}, \mathcal{M} \rangle$-factorization category
\[\mathcal{C}^{\mathsf{2}} \equiv \mathcal{E} \odot \mathcal{M}.\]
\end{theorem}
This equivalence is mediated by factorization and composition.

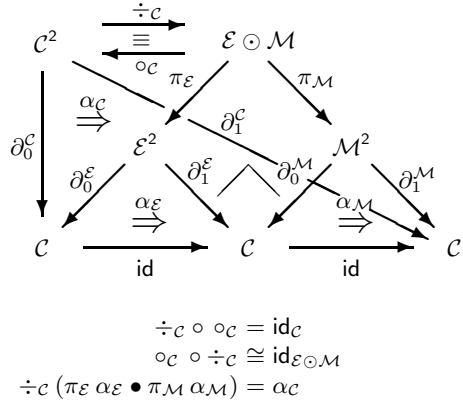
\begin{figure}
\begin{center}
\setlength{\unitlength}{0.78pt}
\begin{tabular}{c}
\begin{picture}(200,110)(-83,0)
\put(0,0){\begin{picture}(100,100)(0,0)
\put(5,75){\makebox(100,50){$\mathcal{E} \odot \mathcal{M}$}}
\put(-50,25){\makebox(100,50){$\mathcal{E}^{\mathsf{2}}$}}
\put(50,25){\makebox(100,50){$\mathcal{M}^{\mathsf{2}}$}}
\put(-100,-25){\makebox(100,50){$\mathcal{C}$}}
\put(0,-25){\makebox(100,50){$\mathcal{C}$}}
\put(100,-25){\makebox(100,50){$\mathcal{C}$}}
\put(-32,56){\makebox(100,50){\footnotesize{$\pi_\mathcal{E}$}}}
\put(33,56){\makebox(100,50){\footnotesize{$\pi_\mathcal{M}$}}}
\put(-80,8){\makebox(100,50){\footnotesize{$\partial_0^\mathcal{E}$}}}
\put(-23,12){\makebox(100,50){\footnotesize{$\partial_1^\mathcal{E}$}}}
\put(23,12){\makebox(100,50){\footnotesize{$\partial_0^\mathcal{M}$}}}
\put(82,8){\makebox(100,50){\footnotesize{$\partial_1^\mathcal{M}$}}}
\put(-50,-35){\makebox(100,50){\footnotesize{$\mathsf{id}$}}}
\put(50,-35){\makebox(100,50){\footnotesize{$\mathsf{id}$}}}
\put(-48,-4){\makebox(100,50){\footnotesize{$\alpha_{\mathcal{E}}$}}}
\put(-48,-14){\makebox(100,50){\Large{$\Rightarrow$}}}
\put(52,-4){\makebox(100,50){\footnotesize{$\alpha_{\mathcal{M}}$}}}
\put(52,-14){\makebox(100,50){\Large{$\Rightarrow$}}}
\put(50,25){\begin{picture}(30,15)(0,-15)
\put(0,0){\line(-1,-1){15}}
\put(0,0){\line(1,-1){15}}
\end{picture}}
\thicklines
\put(40,90){\vector(-1,-1){30}}
\put(60,90){\vector(1,-1){30}}
\put(-10,40){\vector(-1,-1){30}}
\put(10,40){\vector(1,-1){30}}
\put(-30,0){\vector(1,0){60}}
\put(90,40){\vector(-1,-1){30}}
\put(110,40){\vector(1,-1){30}}
\put(70,0){\vector(1,0){60}}
\end{picture}}
\thicklines
\put(-21,108){\vector(1,0){40}}
\put(-49,90){\makebox(100,50){\footnotesize{$\div_{\mathcal{C}}$}}}
\put(-52,75.5){\makebox(100,50){\footnotesize{$\equiv$}}}
\put(-49,62){\makebox(100,50){\footnotesize{$\circ_{\mathcal{C}}$}}}
\put(19,94){\vector(-1,0){40}}
\put(-35,90){\line(2,-1){44}}
\put(21,62){\line(2,-1){43}}
\put(79,33){\line(2,-1){15}}
\put(109,18){\vector(2,-1){28}}
\put(-50,85){\vector(0,-1){70}}
\put(-98,75){\makebox(100,50){$\mathcal{C}^{\mathsf{2}}$}}
\put(-109,24){\makebox(100,50){\footnotesize{$\partial_0^{\mathcal{C}}$}}}
\put(-7,36){\makebox(100,50){\footnotesize{$\partial_1^{\mathcal{C}}$}}}
\put(-75,45){\makebox(100,50){\footnotesize{$\alpha_{\mathcal{C}}$}}}
\put(-75,35){\makebox(100,50){\Large{$\Rightarrow$}}}
\end{picture}
\\ \\ \\
$\begin{array}{rcl}
\div_{\mathcal{C}} \circ\, \circ_{\mathcal{C}} & = & \mathsf{id}_{\mathcal{C}} 
\\
\circ_{\mathcal{C}} \,\circ \div_{\mathcal{C}} & \cong & \mathsf{id}_{\mathcal{E} \odot \mathcal{M}} 
\\
\div_{\mathcal{C}} 
\left(
\pi_\mathcal{E} \, \alpha_{\mathcal{E}}
\bullet
\pi_\mathcal{M} \, \alpha_{\mathcal{M}}
\right)
& = &
\alpha_{\mathcal{C}}
\end{array}$
\end{tabular}
\end{center}
\caption{Factorization Equivalence}
\label{factorization-equivalence}
\end{figure}

%\include{fact}
%%%%%%%%%%%%%%%%%%%%%%%%%%%%%%%%%%%%%%%%%%%%%%%%%%%%%%%%%%%%%%%%%%%%%%%%%%%%%%%%
%%%%%%%%%%%%%%%%%%%%%%%%%%%%%%%%%%%%%%%%%%%%%%%%%%%%%%%%%%%%%%%%%%%%%%%%%%%%%%%%

%%%%%%%%%%%%%%%%%%%%%%%%%%%%%%%%%%%%%%%%%%%%%%%%%%%%%%%%%%%%%%%%%%%%%%%%%%%%%%%%
%%%%%%%%%%%%%%%%%%%%%%%%%%%%%%%%%%%%%%%%%%%%%%%%%%%%%%%%%%%%%%%%%%%%%%%%%%%%%%%%

%%%%%%%%%%%%%%%%%%%%%%%%%%%%%%%%%%%%%%%%%%%%%%%%%%%%%%%%%%%%%%%%%%%%%%%%%%%%%%%%
%%%%%%%%%%%%%%%%%%%%%%%%%%%%%%%%%%%%%%%%%%%%%%%%%%%%%%%%%%%%%%%%%%%%%%%%%%%%%%%%
\section{Topos Theory\label{sec:topos:theory}}
%%%%%%%%%%%%%%%%%%%%%%%%%%%%%%%%%%%%%%%%%%%%%%%%%%%%%%%%%%%%%%%%%%%%%%%%%%%%%%%%
%%%%%%%%%%%%%%%%%%%%%%%%%%%%%%%%%%%%%%%%%%%%%%%%%%%%%%%%%%%%%%%%%%%%%%%%%%%%%%%%

%%%%%%%%%%%%%%%%%%%%%%%%%%%%%%%%%%%%%%%%%%%%%%%%%%%%%%%%%%%%%%%%%%%%%%%%%%%%%%%%
\subsection{Topos Fundamentals\label{subsec:topos:fundamentals}}
%%%%%%%%%%%%%%%%%%%%%%%%%%%%%%%%%%%%%%%%%%%%%%%%%%%%%%%%%%%%%%%%%%%%%%%%%%%%%%%%

Let $\mathcal{C}$ be any category.
Monics into an $\mathcal{C}$-object $A$ are ordered.
Two monic $m_1 : B_1 \rightarrow A$ and $m_2 : B_1 \rightarrow A$ are ordered as $m_1 \subseteq m_2$ when there is a $\mathcal{C}$-morphism $b : B_1 \rightarrow B_2$ such that $m_1 = b \cdot m_2$.
Then $b$ is unique and also monic.
The associated equivalence relation induces equivalence classes on monics called subobjects of $A$.
The order on monics lifts to an order on subobjects.
Subobject order is transitive, reflexive and antisymmetric.
The poset of subobejcts on $A$ is denoted $\mathsf{sub}(A)$.
Assume that $\mathcal{C}$ is finitely complete.
The pullback of a monic along a $\mathcal{C}$-morphism gives another monic.
This defines a contravariant functor
\[ \mathsf{sub} : \mathcal{C}^{\mathrm{op}} \rightarrow \mathsf{Set}. \]

A \emph{topos} $\mathcal{B}$ is a finitely-complete, cartesian closed category with a representable subobject functor.

\paragraph{Finite Limits.}
There is a terminal object $1$ in $\mathcal{B}$.
For each pair of $\mathcal{B}$-objects $(A,B)$, 
there is a specified binary product $A{\times}B$
together with its projections 
$\pi_A : A{\times}B \rightarrow A$ and $\pi_B : A{\times}B \rightarrow B$.
In particular,
for each pair of $\mathcal{B}$-objects $(A,B)$, 
there is a twist or symmetrizing 
$\mathcal{B}$-isomorphism $\tau_{A,B} : B{\times}A \rightarrow A{\times}B$, 
which is the unique mediating morphism for the product cone 
$A \stackrel{\pi_A}{\leftarrow} B{\times}A \stackrel{\pi_B}{\rightarrow} B$.

\paragraph{Cartesian Closed.}
For each $\mathcal{B}$-object $A$,
and hence functor 
$A{\times}(\mbox{-}) : \mathcal{B} \rightarrow \mathcal{B} : B \mapsto A{\times}B$, 
there is an \emph{exponent} functor
${(\mbox{-})}^A : \mathcal{B} \rightarrow \mathcal{B} : C \mapsto C^A$, 
with natural bijection\footnote{We use the notation
$\mathcal{B}(A,B) = B^A$ for the exponent or (internal) hom-object, an object of  $\mathcal{B}$. We use the notation $\mathsf{hom}(A,B)$ for the (external) hom-set, an object of $\mathsf{Set}$. There is an underlying functor 
$|\mbox{-}| : \mathcal{B} \rightarrow \mathsf{Set}$, 
defined by $|\mbox{-}| \doteq \mathcal{B}(1,\mbox{-}) = {(\mbox{-})}^1$.
Then,
$\mathsf{hom}(A,B) \cong |B^A|$.}
\[ \mathsf{hom}(A{\times}B,C) \cong \mathsf{hom}(B,C^A) \]
mediated by the \emph{constant augmentation} $\mathcal{B}$-morphism (unit)
$\gamma_{A,B} : B \rightarrow {(A{\times}B)}^A$ 
and the \emph{evaluation} $\mathcal{B}$-morphism (counit)
$\varepsilon_{A,C} : A{\times}{C^A} \rightarrow C$.
A $\mathcal{B}$-morphism $f : A{\times}B \rightarrow C$ 
is in bijective correspondence with 
the $\mathcal{B}$-morphism $g : B \rightarrow C^A$,
where
$f = (A{\times}g) \cdot \varepsilon_{A,C}$ and 
$g = \gamma_{A,B} \cdot f^A: B \rightarrow C^A$.

\paragraph{Subobject Classifier.}
The subobject functor being representable means
that there is a $\mathcal{B}$-object of truth values $\Omega$ that serves as a dualizing object.
This comes equipped with a subobject $\top : 1 \hookrightarrow \Omega$ called the \emph{truth} element
such that
for any $\mathcal{B}$-object $A$ and subobject $A_0 \hookrightarrow A$,
there is a unique $\mathcal{B}$-morphism $\chi_{A_0} : A \rightarrow \Omega$
such that
\begin{center}
\begin{tabular}{c}
\setlength{\unitlength}{0.75pt}
\begin{picture}(80,60)(0,0)
%\put(0,0){\framebox(80,60){}}
\put(-30,-15){\makebox(60,30){$1$}}
\put(-30,45){\makebox(60,30){$A_0$}}
\put(50,-15){\makebox(60,30){$\Omega$}}
\put(50,45){\makebox(60,30){$A$}}
\put(25,-23){\makebox(30,30){\footnotesize{$\top$}}}
\put(-25,15){\makebox(30,30){\footnotesize{$!_{A_0}$}}}
\put(82,15){\makebox(30,30){\footnotesize{$\chi_{A_0}$}}}
\put(20,0){\vector(1,0){40}}
\put(20,60){\vector(1,0){40}}
\put(0,45){\vector(0,-1){30}}
\put(80,45){\vector(0,-1){30}}
\put(60,20){\line(1,0){10}}
\put(60,20){\line(0,-1){10}}
\end{picture}
\end{tabular}
\end{center}
is a pullback.
The fact that the subobject functor is represented by $\Omega$ is equivalent to the fact that there is a natural isomorphism
\[ \mathsf{sub}(\mbox{-}) \rightarrow \mathsf{hom}(\mbox{-},\Omega), \]
which takes a subobject to its character\footnote{Define the power operator ${\wp}A \doteq \mathcal{B}(A,\Omega)$. Then, $|{\wp}A| \cong \mathsf{sub}(A)$.}.

\paragraph{Boolean Connectives.}
Let \emph{false} $\bot : 1 \rightarrow \Omega$ be the character of the subobject $!_1 : 0 \hookrightarrow 1$.
In turn,
define \emph{negation} $\neg : \Omega \rightarrow \Omega$ be the character of the subobject $\bot : 1 \hookrightarrow \Omega$.
Define \emph{conjunction} $\wedge : \Omega{\times}\Omega \rightarrow \Omega$ to be the charcter of truth paired with itself $(t,t) : 1 \rightarrow \Omega{\times}\Omega$.
Let $e : E \rightarrow \Omega{\times}\Omega$ denote the equalizer of the parallel pair
${\wedge},\pi_0 : \Omega{\times}\Omega \rightarrow \Omega$ consisting of conjunction and first projection.
Define \emph{implication} $\Rightarrow : \Omega{\times}\Omega \rightarrow \Omega$ to be the character of $e$. 
Prove: $\langle \Omega, \Rightarrow \rangle$ is a preorder.
Define the \emph{bottom} $\mathcal{B}$-morphism
$\bot_A =\; !_A \cdot \bot : A \rightarrow 1 \rightarrow \Omega$
and the \emph{top} $\mathcal{B}$-morphism
$\top_A =\; !_A \cdot \top : A \rightarrow 1 \rightarrow \Omega$,
where top $\top_A$ is the character of the monic $1_A : A \rightarrow A$.

\paragraph{Subobjects.}
Hence,
in a topos $\mathcal{B}$, 
a subobject $m$ of a $\mathcal{B}$-object $A$
can be represented: 
either (1) as an equivalence class of monics $\iota_m : \Box_m \rightarrow A$ 
or (2) as a $\mathcal{B}$-morphism $\chi_m : A \rightarrow \Omega$.
Hence,
there are two ways to order subobjects of $A$:
$m_1 \subseteq m_2$
when
either (1) there is a (necessarily monic) $\mathcal{B}$-morphism 
$m : \Box_{m_1} \rightarrow \Box_{m_2}$
such that $\iota_{m_1} = m \cdot \iota_{m_2}$
or (2) $\iota_{m_1} \cdot \chi_{m_2} = \top_{B_1} = {!}_{B_1} \cdot \top$. 

\paragraph{Finite Colimits.}
A topos $\mathcal{B}$ also has finite colimits.
Thus, there is an initial object $0$ in $\mathcal{B}$.
For each pair of $\mathcal{B}$-objects $(A,B)$, 
there is a specified binary coproduct $A{+}B$
together with its injections $\iota_A : A \rightarrow A{+}B$ and $\iota_B : B \rightarrow A{+}B$.
For each pair of $\mathcal{B}$-objects $(A,B)$, 
there is a twist or symmetrizing isomorphism $\tau_{B,A} : B{+}A \rightarrow A{+}B$.

\paragraph{Power Objects.}
In any topos $\mathcal{B}$, 
the power object ${\wp}A = \Omega^A = \mathcal{B}(A,\Omega)$ has the property that
\[ \mathsf{hom}(A,{\wp}B) 
\stackrel{01}{\cong} \mathsf{rel}(A,B) 
\stackrel{10}{\cong} \mathsf{hom}(B,{\wp}A). \]
where the set of $\mathcal{B}$-relations $\mathsf{rel}(A,B)$ is in bijective correspondence (character) with the set of $\mathcal{B}$-morphisms $\mathsf{hom}(A{\times}B,\Omega)$.
These exponential adjoints are natural isomorphisms in $A$ and $B$.
For any $\mathcal{B}$-object $A$,
the \emph{delta} morphism $\Delta_A : A \hookrightarrow A{\times}A$ 
is the product pairing of the identity $\mathcal{B}$-morphism on $A$ with itself.
For any $\mathcal{B}$-object $A$,
the \emph{singleton} $\mathcal{B}$-morphism
${\{\mbox{-}\}}_A : A \rightarrow {\wp}A$
corresponds to the identity relation $1_A : A \rightharpoondown A$,
whose subobject is delta
and whose character is
the \emph{diagonal} $\mathcal{B}$-morphism
$\delta_A : A{\times}A \rightarrow \Omega$
which is the exponential adjoint of singleton.
For any $\mathcal{B}$-object $A$,
the \emph{membership} $\mathcal{B}$-relation
${\in}_A : A \rightharpoondown {\wp}A$
has as its character the evaluation morphism
$\varepsilon_{A,\Omega} : A{\times}{\Omega^A} \rightarrow \Omega$;
so that the natural bijection
$\mathsf{rel}(A,B) \cong \mathsf{hom}(B,{\wp}A)$
bijectively maps the character $\chi_r : A{\times}B \rightarrow \Omega$ of a $\mathcal{B}$-relation $r : A \rightharpoondown B$ 
to the $\mathcal{B}$-morphism $g : B \rightarrow {\wp}A$,
where
$\chi_r = (A{\times}g) \cdot \chi_{{\in}_A}$ and 
$g = \gamma_{A,B} \cdot \chi_r^A: B \rightarrow {\wp}A$.

\begin{figure}
\begin{center}
\begin{tabular}[b]{c}
\setlength{\unitlength}{1.0pt}
\begin{picture}(80,60)(0,0)
%\put(0,0){\framebox(80,60){}}
\put(-30,-15){\makebox(60,30){${\wp}A$}}
\put(-30,45){\makebox(60,30){$\mathcal{B}{/}A$}}
\put(50,-15){\makebox(60,30){${\wp}B$}}
\put(50,45){\makebox(60,30){$\mathcal{B}{/}B$}}
\put(25,-1){\makebox(30,30){\footnotesize{$\exists_f$}}}
\put(25,-11){\makebox(30,30){\footnotesize{$f^{{-}1}$}}}
\put(25,-21){\makebox(30,30){\footnotesize{$\forall_f$}}}
\put(25,59){\makebox(30,30){\footnotesize{$\Sigma_f$}}}
\put(25,49){\makebox(30,30){\footnotesize{$f^\ast$}}}
\put(25,39){\makebox(30,30){\footnotesize{$\Pi_f$}}}
\put(-25,15){\makebox(30,30){\footnotesize{$\sigma_A$}}}
\put(-3,15){\makebox(30,30){\footnotesize{$\iota_A$}}}
\put(55,15){\makebox(30,30){\footnotesize{$\sigma_B$}}}
\put(77,15){\makebox(30,30){\footnotesize{$\iota_B$}}}
\put(20,10){\vector(1,0){40}}
\put(60,0){\vector(-1,0){40}}
\put(20,-10){\vector(1,0){40}}
\put(20,70){\vector(1,0){40}}
\put(60,60){\vector(-1,0){40}}
\put(20,50){\vector(1,0){40}}
\put(-5,45){\vector(0,-1){30}}
\put(5,15){\vector(0,1){30}}
\put(75,45){\vector(0,-1){30}}
\put(85,15){\vector(0,1){30}}
\end{picture}
\\ \\ \\
\begin{tabular}{p{10cm}}
{\footnotesize Let $f : A \rightarrow B$ be any $\mathcal{B}$-morphism.
In the diagram above
we have the adjunctions
$\exists_f \dashv f^{{-}1} \dashv \forall_f$,
$\Sigma_f \dashv f^\ast \dashv \Pi_f$,
$\sigma_A \dashv \iota_A$,
$\sigma_B \dashv \iota_B$,
and
the commutative squares
$\iota_B \cdot f^\ast = f^{{-}1} \cdot \iota_A$
and
$\iota_A \cdot \Pi_f = \forall_f \cdot \iota_B$
hold be definition,
and
the commutative squares (up to natural isomorphism)
$\Sigma_f \cdot \sigma_B \cong \sigma_A \cdot \exists_f$
and
$f^\ast \cdot \sigma_A \cong \sigma_B \cdot f^{{-}1}$
hold by uniqueness of adjoints.
We have
$\iota_A \cdot \Sigma_f \cong \exists_f \cdot \iota_B$
iff $f$ is a monomorphism,
and we have
$\Pi_f \cdot \sigma_B \cong \sigma_A \cdot \forall_f$
for all $f$ iff $\mathcal{B}$ satisfies the implicit axiom of choice (IC).}
\end{tabular}
\end{tabular}
\end{center}
\caption{Hyperdoctrinal Diagram}
\label{hyperdoctrinal-diagram}
\end{figure}

\paragraph{Epi-Mono Factorization.}
Any topos $\mathcal{B}$ has a chosen epi-mono factorization system,
where a $\mathcal{B}$-morphism factors in terms of the equalizer of its cokernel pair.
For any $\mathcal{B}$-morphism $f : A \rightarrow B$,
the \emph{cokernel pair}
$B \stackrel{\iota_0}{\rightarrow} \mathsf{cok}(f) \stackrel{\iota_1}{\leftarrow} B$
is the pushout of $f$ with itself.
\begin{center}
\begin{tabular}{c}
\\
\setlength{\unitlength}{0.7pt}
\begin{picture}(60,60)(0,0)
%\put(0,0){\framebox(60,60){}}
\put(-30,45){\makebox(60,30){$A$}}
\put(-30,-15){\makebox(60,30){$B$}}
\put(30,45){\makebox(60,30){$B$}}
\put(30,-15){\makebox(60,30){$\mathsf{cok}(f)$}}
\put(-25,15){\makebox(30,30){\footnotesize{$f$}}}
\put(15,55){\makebox(30,30){\footnotesize{$f$}}}
\put(55,15){\makebox(30,30){\footnotesize{$\iota_0$}}}
\put(15,-25){\makebox(30,30){\footnotesize{$\iota_1$}}}
\put(15,0){\vector(1,0){20}}
\put(15,60){\vector(1,0){30}}
\put(0,45){\vector(0,-1){30}}
\put(60,45){\vector(0,-1){30}}
\end{picture}
\\ \\
\end{tabular}
\end{center}
Let $\mu_f : f(A) \rightarrow \mathsf{cok}(f)$ be the equalizer of the parallel pair
$\iota_0,\iota_1 : \mathsf{cok}(f) \rightarrow B$.
Since $f \cdot \iota_0 = f \cdot \iota_1$,
there is a unique $\mathcal{B}$-morphism $\varepsilon_f : A \rightarrow f(A)$
such that $f = \varepsilon_f \cdot \mu_f$.
Then $(A, \varepsilon_f, f(A), \mu_f, B)$ is an epi-mono factorization of $f$.

\paragraph{Internal Category Theory.}
A \emph{category} $\mathbf{C}$ in (internal to) a topos $\mathcal{B}$ is a sextuple 
$\mathbf{C} 
= \langle \mathsf{obj}(\mathbf{C}), \mathsf{mor}(\mathbf{C}),
\iota_{\mathbf{C}}, {\circ}_{\mathbf{C}}, 
\partial^{0}_{\mathbf{C}}, \partial^{1}_{\mathbf{C}} \rangle$
consisting of 
a $\mathcal{B}$-object of \emph{objects} $\mathsf{obj}(\mathbf{C})$,
a $\mathcal{B}$-object of \emph{morphisms} $\mathsf{mor}(\mathbf{C})$,
an \emph{identity} $\mathcal{B}$-morphism 
$\iota_{\mathbf{C}} : \mathsf{obj}(\mathbf{C}) \rightarrow \mathsf{mor}(\mathbf{C})$,
a \emph{composition} $\mathcal{B}$-morphism 
${\circ}_{\mathbf{C}} : \mathsf{mor}(\mathbf{C}){\times}_{\mathsf{obj}(\mathbf{C})}\mathsf{mor}(\mathbf{C}) \rightarrow \mathsf{mor}(\mathbf{C})$,
and \emph{source} and \emph{target} $\mathcal{B}$-morphisms 
$\partial^{0}_{\mathbf{C}},\partial^{1}_{\mathbf{C}}
: \mathsf{mor}(\mathbf{C}) \rightarrow \mathsf{obj}(\mathbf{C})$.
This data is subject to 
the \emph{associativity} and \emph{unit} laws:
\begin{center}
$\begin{array}{r@{\hspace{10pt}=\hspace{10pt}}l}
(\pi_0^3,{\circ}_{\mathbf{C}}) \cdot {\circ}_{\mathbf{C}}
& ({\circ}_{\mathbf{C}},\pi_2^3) \cdot {\circ}_{\mathbf{C}} \\
(\partial^{0}_{\mathbf{C}} \cdot \iota_{\mathbf{C}}, 1_{\mathsf{mor}(\mathbf{C})}) \cdot {\circ}_{\mathbf{C}}
& 1_{\mathsf{mor}(\mathbf{C})} \\
(1_{\mathsf{mor}(\mathbf{C})}, \partial^{1}_{\mathbf{C}} \cdot \iota_{\mathbf{C}}) \cdot {\circ}_{\mathbf{C}}
& 1_{\mathsf{mor}(\mathbf{C})}.
\end{array}$
\end{center}
A \emph{functor} 
$\mathbf{F} : \mathbf{A} \rightarrow \mathbf{B}$
from $\mathcal{B}$-category $\mathbf{A}$ to $\mathcal{B}$-category $\mathbf{B}$
in (internal to) a topos $\mathcal{B}$
is a pair
$\mathbf{F} 
= \langle 
\mathsf{obj}(\mathbf{F}), 
\mathsf{mor}(\mathbf{F})
\rangle$,
consisting of 
an \emph{object} $\mathcal{B}$-morphism 
$\mathsf{obj}(\mathbf{F})
: \mathsf{obj}(\mathbf{A}) \rightarrow \mathsf{obj}(\mathbf{B})$
and
a \emph{morphism} $\mathcal{B}$-morphism 
$\mathsf{mor}(\mathbf{F})
: \mathsf{mor}(\mathbf{A}) \rightarrow \mathsf{mor}(\mathbf{B})$,
which preserves source, target, composition and identity:
\begin{center}
$\begin{array}{r@{\hspace{10pt}=\hspace{10pt}}l}
\mathsf{mor}(\mathbf{F}) \cdot \partial^{0}_{\mathbf{B}}
& \partial^{0}_{\mathbf{A}} \cdot \mathsf{obj}(\mathbf{F}) \\
\mathsf{mor}(\mathbf{F}) \cdot \partial^{1}_{\mathbf{B}}
& \partial^{1}_{\mathbf{A}} \cdot \mathsf{obj}(\mathbf{F}) \\
\mathsf{mor}(\mathbf{F}){\times}_{\mathsf{obj}(\mathbf{F})}\mathsf{mor}(\mathbf{F})
 \cdot {\circ}_{\mathbf{B}}
& {\circ}_{\mathbf{A}} \cdot \mathsf{mor}(\mathbf{F}) \\
\mathsf{obj}(\mathbf{F}) \cdot \iota_{\mathbf{B}}
& \iota_{\mathbf{A}} \cdot \mathsf{mor}(\mathbf{F}).
\end{array}$
\end{center}
Let $\mathsf{Cat}(\mathcal{B})$ denote the category of $\mathcal{B}$-categories and $\mathcal{B}$-functors,
where the composition and identities of $\mathcal{B}$-functors is defined componentwise.
The underlying functor
${|\mbox{-}|}_{\mathcal{B}}
: \mathsf{Cat}(\mathcal{B}) \rightarrow \mathcal{B}$
maps $\mathcal{B}$-categories to their object $\mathcal{B}$-object
and maps $\mathcal{B}$-functors to their object $\mathcal{B}$-morphism.

%%%%%%%%%%%%%%%%%%%%%%%%%%%%%%%%%%%%%%%%%%%%%%%%%%%%%%%%%%%%%%%%%%%%%%%%%%%%%%%%
\subsection{Relational Structures\label{subsec:relational:structures}}
%%%%%%%%%%%%%%%%%%%%%%%%%%%%%%%%%%%%%%%%%%%%%%%%%%%%%%%%%%%%%%%%%%%%%%%%%%%%%%%%

\paragraph{Relations.}
A \emph{(binary) relation} $r : A \rightharpoondown B$
in (internal to) $\mathcal{B}$ 
from $A$ to $B$ can be regarded as either
(1) a character
$\chi_r : A{\times}B \rightarrow \Omega$,
(2) a subobject of the product
$\iota_r : \Box_r \hookrightarrow A{\times}B$
(the character of the latter is the former),
or 
(3) a projection pair $\pi^r_0 : \Box_r \rightarrow A$ and $\pi^r_1 : \Box_r \rightarrow B$, 
whose pairing $\iota_r = (\pi^r_0,\pi^r_1)$ is monic (and the subobject).
We can use the abbreviation $r$ for the character $\chi_r$.
For any  relation $r : A \rightharpoondown B$ and 
any pair of morphisms $a : A^{\prime} \rightarrow A$ and $b : B^{\prime} \rightarrow B$,
we can use the abbreviation $r(a,b)$ for the relation whose character is 
$(a{\times}b) \cdot \chi_r
: A^{\prime}{\times}B^{\prime} \rightarrow A{\times}B \rightarrow \Omega$.
Using exponential adjoints on its character,
a $\mathcal{B}$-relation $r : A \rightharpoondown B$ can equivalently be regarded as a $\mathcal{B}$-morphism in two ways:
either (1) the \emph{01-fiber} $r^{01} : A \rightarrow {\wp}B$
or (2) the \emph{10-fiber} $r^{10} : B \rightarrow {\wp}A$.
In the other direction using membership,
the relation can be expressed in terms of the fibers as
either $r = {\in}^{\propto}_A(r^{01},1_B)$
or $r = {\in}_A(1_A,r^{10})$.
These facts are equivalent to the statements that
the fibers form infomorphisms
$(r^{01},1_B) : ({\wp}B,B,{\in}^{\propto}_B) \rightleftharpoons (A,B,r)$
and $(1_A,r^{10}) : (A,B,r) \rightleftharpoons (A,{\wp}A,{\in}_A)$.
The \emph{transpose} $r^{\propto} : B \rightharpoondown A$ 
is the relation whose character is the composition $\tau_{A,B} \cdot r : B{\times}A \rightarrow A{\times}B \rightarrow \Omega$ of twist with character.
Then, $(r^{\propto})^{01} = r^{10}$ and $(r^{\propto})^{10} = r^{01}$.
Two relations $r,s : A \rightharpoondown B$ are ordered $r \leq s$
when their subobjects are ordered $\mathsf{sub}(r) \subseteq \mathsf{sub}(s)$;
that is, $\iota_r \cdot \chi_s = \top_{\Box_r}$.

The \emph{identity relation}
$1_A : A \rightharpoondown A$
has character $\delta_A : A{\times}A \rightarrow \Omega$
and subobject $\Delta_A = (1_A,1_A) : A \hookrightarrow A{\times}A$.
We can use the abbreviation $A$ for $1_A$.
A pair of relations $(r,s)$ 
is composable
when the target of the first is the source of the second:
$r : A \rightharpoondown B$ and $s : B \rightharpoondown C$.
For any composable pair of relations $(r,s)$,
there is a \emph{composition relation}
$r \circ s : A \rightharpoondown C$,
whose character is the composition of the exponential adjoint of
$(1_A {\times} \Delta_B {\times} 1_C) \cdot (r {\times} s) \cdot \wedge
: A {\times} B {\times} C \rightarrow A {\times} B {\times} B {\times} C \rightarrow \Omega {\times} \Omega \rightarrow \Omega$
with the existential image
$\exists_B : \Omega^B \stackrel{\exists_{!_B}}{\rightarrow} \Omega^1 \cong \Omega$.
For any composable pair of relations $(r,s)$,
there is also a \emph{composition morphism}
${\circ}_{(r,s)} : \Box_r \times_B \Box_s \rightarrow \Box_{(r \circ s)}$,
whose source is the vertex of the pullback
$\pi_0 : \Box_r \leftarrow \Box_r{{\times}_B}\Box_r \rightarrow \Box_s : \pi_1$
of the opspan 
$\pi^r_1 : \Box_r \rightarrow \Box_B \leftarrow \Box_s : \pi^s_0$,
and
satisfies the commutative diagrams
${\circ}_{(r,s)} \cdot \pi^{(r \circ s)}_0 = \pi_0 \cdot \pi^r_0$
and
${\circ}_{(r,s)} \cdot \pi^{(r \circ s)}_1 = \pi_1 \cdot \pi^s_1$.
Composition satisfies the associative law $r \circ (s \circ t) = (r \circ s) \circ t$ for any pair of composable pairs $(r,s)$ and $(s,t)$,
and satisfies the unit laws $1_A \circ r = r = r \circ 1_B$ for any relation $r : A \rightharpoondown B$.

For any $\mathcal{B}$-object $A$,
a subobject $X \hookrightarrow A$ can be regarded as a relation in two ways:
(1) in the direct sense $X : 1 \rightharpoondown A$ 
with character $1{\times}A \cong A \rightarrow \Omega$,
or 
(2) in the inverse sense $X^{\!\propto} : A \rightharpoondown 1$
with character $A{\times}1 \cong A \rightarrow \Omega$.
Any $\mathcal{B}$-morphism $f : A \rightarrow B$ can be regarded as a relation in two ways:
(1) in the direct sense 
$f^\triangleright = B(f,1_{B}) : A \rightharpoondown B$ 
with character $A{\times}B \rightarrow \Omega$ 
that is one of the two exponential adjoints of
the composite $\mathcal{B}$-morphism $f \cdot {\{\mbox{-}\}}_B : A \rightarrow B \rightarrow {\wp}B$,
or
(2) in the inverse sense 
$f^\triangleleft = B(1_{B},f) : B \rightharpoondown A$ 
with character $B{\times}A \rightarrow \Omega$ 
that is the other exponential adjoint of $f \cdot {\{\mbox{-}\}}_B$.
Then,
$(f^{\triangleleft})^{\propto} = f^\triangleright$ 
and 
$(f^{\triangleright})^{01} = f \cdot {\{\mbox{-}\}}_B = (f^{\triangleleft})^{10}$.

\paragraph{Residuation.}
For any relation $r : A \rightharpoondown B$,
the \emph{left residuation} of a relation $s : A \rightharpoondown C$ along $r$ 
is a relation $r{\setminus}s : B \rightharpoondown C$,
whose character 
$\chi_{r{\setminus}s} : A{\times}B \rightarrow \Omega^A \stackrel{\Rightarrow}{\rightarrow} A$
is defined to be the composition with implication of the exponential adjoint of the composite
\begin{center}
$\begin{array}{l}
(\Delta_A{\times}1_{(B{\times}C)}) \cdot (1_A{\times}\tau_{B,A}{\times}1_C) \cdot (r{\times}s) \cdot {\Rightarrow} 
\\ : A{\times}B{\times}C \rightarrow A{\times}A{\times}B{\times}C
\rightarrow A{\times}B{\times}A{\times}C 
\rightarrow \Omega{\times}\Omega \rightarrow \Omega .
\end{array}$
\end{center}
For any relation $r : A \rightharpoondown B$,
the \emph{right residuation} of a relation $s : C \rightharpoondown B$ along $r$ 
is a relation $r{/}s : C \rightharpoondown A$,
whose character 
$\chi_{r{\setminus}s} : A{\times}B \rightarrow \Omega^A \stackrel{\Rightarrow}{\rightarrow} A$
is defined to be the composition with implication of the exponential adjoint of the composite
\begin{center}
$\begin{array}{l}
(\Delta_A{\times}1_{(B{\times}C)}) \cdot (1_A{\times}\tau_{B,A}{\times}1_C) \cdot (r{\times}s) \cdot {\Rightarrow} 
\\ : A{\times}B{\times}C \rightarrow A{\times}A{\times}B{\times}C
\rightarrow A{\times}B{\times}A{\times}C 
\rightarrow \Omega{\times}\Omega \rightarrow \Omega .
\end{array}$
\end{center}
Along any $\mathcal{B}$-morphism $f : A \rightarrow B$,
the existential image, the inverse image and the universal image
are defined in terms of relational composition and residuation.
\begin{center}
$\begin{array}{r@{\hspace{5pt}=\hspace{5pt}}l@{\hspace{5pt}:\hspace{5pt}}l}
\exists{f} & X \circ f^{\triangleright}                 & {\wp}A \rightarrow {\wp}B \\
{f}^{-1}   & \multicolumn{2}{l}{Y /\, f^{\triangleright}}                             \\
           & Y \circ {f}^{\triangleleft}                & {\wp}B \rightarrow {\wp}A \\
\forall{f} & X /\, f^{\triangleleft}                    & {\wp}A \rightarrow {\wp}B
\end{array}$
\end{center}
Since
\begin{center}
\begin{tabular}{l}
$\exists{f}(X) \leq Y$ \underline{iff} $X \circ f^{\triangleright} \leq Y$
\underline{iff} $X \leq Y / f^{\triangleright}$ \underline{iff} $X \leq {f}^{-1}(Y)$
\\
${f}^{-1}(Y) \leq X$ \underline{iff} $Y \circ {f}^{\triangleleft} \leq X$
\underline{iff} $Y \leq X / f^{\triangleleft}$ \underline{iff} $Y \leq \forall{f}(X)$,
\end{tabular}
\end{center}
we have the two coupled order adjunctions in Figure~\ref{hyperdoctrinal-diagram}
\begin{center}
$\begin{array}{r@{\hspace{5pt}=\hspace{5pt}}l@{\hspace{5pt}:\hspace{5pt}}l}
\mathsf{dir}(f) 
& \langle \exists{f}, {f}^{-1} \rangle
& {\wp}A \rightarrow {\wp}B \\
\mathsf{inv}(f) 
& \langle {f}^{-1}, \forall{f} \rangle
& {\wp}B \rightarrow {\wp}A.
\end{array}$
\end{center}

\paragraph{Derivation.}
Any relation $r : A \rightharpoondown B$ defines derivation monotonic morphisms in two directions:
(1) the \emph{forward derivation} monotonic morphism is the intersection of the existential image of the 01-fiber
$r^{\Rightarrow} = (\exists r^{01})^{\propto} \cdot {\cap}_{B}
: {\wp}A^{\propto} \rightarrow {\wp}{\wp}B^{\propto} \rightarrow {\wp}B$;
(2) the \emph{reverse derivation} monotonic morphism is the intersection of the existential image of the 10-fiber
$r^{\Leftarrow} = (\exists r^{10})^{\propto} \cdot {\cap}_{A}
: {\wp}B^{\propto} \rightarrow {\wp}{\wp}A^{\propto} \rightarrow {\wp}A$.
Clearly,
$r^{\Leftarrow} = (r^{\propto})^{\Rightarrow}$.
{\bfseries Prove:}
Forward derivation is a contravariant monotonic morphism
$r^{\Rightarrow} : {\wp}A^{\propto} \rightarrow {\wp}B$
and
reverse derivation is a contravariant monotonic morphism
$r^{\Leftarrow} : {\wp}B^{\propto} \rightarrow {\wp}A$.

Let $\mathbf{A} = \langle A, \leq_A \rangle$ be any ${\mathcal E}$-preorder.
The \emph{up-segment} monotonic morphism
$\uparrow_A : \mathbf{A}^{\propto} \rightarrow {\wp}\mathbf{A}$
is the exponential adjoint of the character 
$\chi_{\leq_A} : A{\times}A \rightarrow \Omega$
of the order relation $\leq_A : A \rightharpoondown A$. 
The \emph{down-segment} monotonic morphism
$\downarrow_A : \mathbf{A} \rightarrow {\wp}\mathbf{A}$
is the exponential adjoint of the character 
$\tau_{A,A} \cdot \chi_{\leq_A} : A{\times}A \rightarrow \Omega$
of the opposite order relation $\leq^{\propto}_A : A \rightharpoondown A$. 

\paragraph{Example.} 
Let ${\mathcal E}$ be the topos $\mathsf{Set}$. 
For any two relations $R : A \rightharpoondown B$ and $S : A \rightharpoondown C$,
the left residuation relation is defined by
$R{\setminus}S 
= \{ (b,c) \mid b{\in}B, c{\in}C, \forall_{a \in A} (aRb \Rightarrow aSc) \}
: B \rightharpoondown C$.
If $A$ is any set,
the left residuation $(\in_{{\wp}A}){\setminus}(\in_{A}^{\propto})$ 
is define by
$(\in_{{\wp}A}){\setminus}(\in_{A}^{\propto}) 
= \{ ({\mathcal X},a) \mid {\mathcal X}{\in}{\wp}{\wp}A, a{\in}A, \forall_{X \in {\wp}A} (X{\in}{\mathcal X} \Rightarrow a{\in}X) \}
: {\wp}{\wp}A \rightharpoondown A$,
and hence the intersection function
$\cap_A : {\wp}{\wp}A \rightarrow {\wp}A$
maps any collection of subsets ${\mathcal X} \in {\wp}{\wp}A$
to the subset
$\{ a \in A \mid \forall_{X \in {\wp}A} (X{\in}{\mathcal X} \Rightarrow a{\in}X) \}$;
and
the composition $(\in_{A}){\circ}(\in_{{\wp}A})$ 
is define by
$(\in_{A}){\circ}(\in_{{\wp}A}) 
= \{ (a, {\mathcal X}) \mid a{\in}A, {\mathcal X}{\in}{\wp}{\wp}A, \exists_{X \in {\wp}A} (a{\in}X \,\&\, X{\in}{\mathcal X}) \}
: A \rightharpoondown {\wp}{\wp}A$,
and hence the union function
$\cup_A : {\wp}{\wp}A \rightarrow {\wp}A$
maps any collection of subsets ${\mathcal X} \in {\wp}{\wp}A$
to the subset
$\{ a \in A \mid \forall_{X \in {\wp}A} (a{\in}X \,\&\, X{\in}{\mathcal X}) \}$.

\paragraph{Residuation Properties.}
To a large extent the foundation of conceptua is based upon binary relations (or matrices) 
and centered upon the axiom of adjointness between relational composition and residuation. 
This composition/residuation adjointness axiom is similar to the axiom of adjointness between conjunction and implication.
Since composition and residuation are binary, the axiom has two statements:
(1) Left composition is (left) adjoint to left residuation:
$r \circ s \subseteq t$ iff $s \subseteq r \setminus t$, for any compatible binary relations $r$, $s$ and $t$.
(2) Right composition is (left) adjoint to right residuation:
$r \circ s \subseteq t$ iff $r \subseteq t/s$, for any compatible binary relations $r$, $s$ and $t$.
Some derived properties are that residuation preserves composition: 
$(r_{1}{\circ}r_{2}) \setminus t = r_{2}{\setminus}(r_{1}{\setminus}t)$ 
and 
$t/(s_{1}{\circ}s_{2}) = (t/s_{2})/s_{1}$ 
and that residuation preserves identity: 
$Id_{A}{\setminus}t = t$ and $t/Id_{B} = t$. 
The involutions of transpose and negation are of secondary importance. 
The axiom for transpose states that transpose dualizes residuation:  
$(r{\setminus}t)^{\propto} = t^{\propto}/r^{\propto}$
and 
$(t/s)^{\propto} = s^{\propto}{\setminus}t^{\propto}$.

There are two important associative laws 
--- one unconstrained the other constrained. 
There is an unconstrained associative law: 
$(r{\setminus}t)/s = r{\setminus}(t/s)$, 
for all $t \subseteq A{\times}B$, $r \subseteq A{\times}C$ and $s \subseteq D{\times}B$. 
There is also an associative law constrained by closure: 
if $t$ is an endorelation and $r$ and $s$ are closed with respect to $t$, 
$r = t/(r{\setminus}t)$ and $s = (t/s){\setminus}t$, then $(t/s){\setminus}r = s/(r{\setminus}t)$, 
for all $t \subseteq A{\times}A$, $r \subseteq A{\times}B$ and $s \subseteq C{\times}A$. 
${\mathcal E}$-morphisms have a special behavior with respect to derivation. 
If ${\mathcal E}$-morphism $f$ and relation $r$ are composable, 
then $f^{\propto}{\setminus}r = f{\circ}r$. 
If relation $s$ and the opposite of ${\mathcal E}$-morphism $g$ are composable, 
then $s/g = s{\circ}{g}^{\propto}$. 

%\include{top}

%%%%%%%%%%%%%%%%%%%%%%%%%%%%%%%%%%%%%%%%%%%%%%%%%%%%%%%%%%%%%%%%%%%%%%%%%%%%%%%%
%%%%%%%%%%%%%%%%%%%%%%%%%%%%%%%%%%%%%%%%%%%%%%%%%%%%%%%%%%%%%%%%%%%%%%%%%%%%%%%%

%%%%%%%%%%%%%%%%%%%%%%%%%%%%%%%%%%%%%%%%%%%%%%%%%%%%%%%%%%%%%%%%%%%%%%%%%%%%%%%%
%%%%%%%%%%%%%%%%%%%%%%%%%%%%%%%%%%%%%%%%%%%%%%%%%%%%%%%%%%%%%%%%%%%%%%%%%%%%%%%%
%%%%%%%%%%%%%%%%%%%%%%%%%%%%%%%%%%%%%%%%%%%%%%%%%%%%%%%%%%%%%%%%%%%%%%%%%%%%%%%%
%%%%%%%%%%%%%%%%%%%%%%%%%%%%%%%%%%%%%%%%%%%%%%%%%%%%%%%%%%%%%%%%%%%%%%%%%%%%%%%%
\section{Order Structures\label{sec:order:structures}}
%%%%%%%%%%%%%%%%%%%%%%%%%%%%%%%%%%%%%%%%%%%%%%%%%%%%%%%%%%%%%%%%%%%%%%%%%%%%%%%%
%%%%%%%%%%%%%%%%%%%%%%%%%%%%%%%%%%%%%%%%%%%%%%%%%%%%%%%%%%%%%%%%%%%%%%%%%%%%%%%%

%%%%%%%%%%%%%%%%%%%%%%%%%%%%%%%%%%%%%%%%%%%%%%%%%%%%%%%%%%%%%%%%%%%%%%%%%%%%%%%%
%%%%%%%%%%%%%%%%%%%%%%%%%%%%%%%%%%%%%%%%%%%%%%%%%%%%%%%%%%%%%%%%%%%%%%%%%%%%%%%%
\subsection{Orders}
%%%%%%%%%%%%%%%%%%%%%%%%%%%%%%%%%%%%%%%%%%%%%%%%%%%%%%%%%%%%%%%%%%%%%%%%%%%%%%%%
%%%%%%%%%%%%%%%%%%%%%%%%%%%%%%%%%%%%%%%%%%%%%%%%%%%%%%%%%%%%%%%%%%%%%%%%%%%%%%%%

\paragraph{Endorelations.}
A relation is an \emph{endorelation} when source equals target.
An endorelation $r$ is 
\emph{transitive} when $r \circ r \leq r$,
\emph{reflexive} when $1_A \leq r$, 
\emph{symmetric} when $r = r^\propto$, 
and \emph{antisymmetric} when $r \cap r^{\propto} \leq 1_A$.
An \emph{order relation} is a transitive, reflexive endorelation.
An \emph{equivalence relation} is a transitive, reflexive, symmetric endorelation.
For any $\mathcal{B}$-morphism $f : A \rightarrow B$, 
there is an associated \emph{kernel} equivalence relation
$\mathsf{ker}(f) = B(f,f) : A \rightharpoondown A$.
A $\mathcal{B}$-morphism $f : A \rightarrow B$ 
\emph{respects} an equivalence relation $\equiv$ on $A$ 
when ${\equiv} \leq \mathsf{ker}(f)$.
For any equivalence relation $\equiv$ on $A$,
there is an associated quotient $\mathcal{B}$-object $A/{\equiv}$
and canonical $\mathcal{B}$-epimorphism
$[\mbox{-}]_{\equiv} : A \rightarrow A/{\equiv}$,
which is the coequalizer of the parallel pair
of $\mathcal{B}$-morphism associated with $\equiv$.
Respectful morphisms factor through quotients:
for any $\mathcal{B}$-morphism $f : A \rightarrow B$ 
that respects $\equiv$,
there is a unique $\mathcal{B}$-morphism $\hat{f} : A/{\equiv} \rightarrow B$,
where $f = [\mbox{-}] \cdot \hat{f}$.

\paragraph{Preorders.}
A \emph{preorder}\footnote{An alternate definition of a preorder is a parallel pair 
$\partial^A_0,\partial^A_1 : \Box_A \rightarrow A$  of $\mathcal{B}$-morphisms, 
whose pairing $\iota_A = (\partial^A_0,\partial^A_1) : \Box_A \rightarrow A{\times}A$ is monic. 
The character of any such pair is the character of the order relation of a preorder 
as defined in this paper.
There are also two associated $\mathcal{B}$-morphisms:
$\circ_A : \Box_A {\times}_A \Box_A \rightarrow \Box_A$ corresponding to transitivity
and $1_A : A \rightarrow \Box_A$ corresponding to reflexivity,
where $1_A \cdot \iota_A = \Delta_A$.} 
in (internal to) a topos $\mathcal{B}$ is a pair 
$\mathbf{A} = \langle A, {\leq}_A \rangle$,
where $A$ is a $\mathcal{B}$-object 
and ${\leq}_A : A \rightharpoondown A$ is an order relation.
For any preorder $\mathbf{A} = \langle A, \leq_A \rangle$
and
any pair of morphisms $a_{0} : A_{0} \rightarrow A$ and $a_{1} : A_{1} \rightarrow A$,
we can use the abbreviation $\mathbf{A}(a_{0},a_{1})$ for the relation 
${\leq}_A(a_{0},a_{1}) : A_{0} \rightharpoondown A_{1}$.
Every preorder $\mathbf{A}$ has an associated equivalence relation 
$\equiv_A : A \rightharpoondown A$
whose character $A{\times}A \rightarrow \Omega$
is the composite $\mathcal{B}$-morphism 
$(\leq_A,\tau_{A,A} \cdot \leq_A) \cdot \wedge
: A{\times}A \rightarrow \Omega{\times}\Omega \rightarrow \Omega$.
A \emph{partial order} or \emph{posetal object} is a preorder, 
whose order relation is antisymmetric.
The \emph{transpose} or \emph{opposite} preorder is
$\mathbf{A}^{\propto} = \langle A, \leq^{\propto}_A \rangle$.
Any function $f : A \rightarrow B$ maps a preorder $\mathbf{B} = \langle B, \leq_B \rangle$
to the \emph{kernel} preorder
$\mathsf{ker}_f(\mathbf{B}) = \langle A, \leq_f^B \rangle$,
whose order relation is $\leq_f^B = \mathbf{B}(f,f) : A \rightharpoondown A$.
Any preorder $\mathbf{A} = \langle A, \leq_A \rangle$
has an associated \emph{quotient} partial order
$\mathsf{quo}(\mathbf{A}) 
= [\mathbf{A}] 
= \langle A/{\equiv_A}, \leq_{[\mathbf{A}]} \rangle$,
whose order relation $\leq_{[\mathbf{A}]} : A/{\equiv_A} \rightharpoondown A/{\equiv_A}$
has the monic component of the epi-mono factorization of 
$\iota_A \cdot ([\mbox{-}]_{\equiv_A}{\times}\,[\mbox{-}]_{\equiv_A}) 
: \Box_A \rightarrow A{\times}A \rightarrow {(A/{\equiv_A})}{\times}{(A/{\equiv_A})}$
as representative monic of its subobject.

\paragraph{Monotonic Morphisms.}
A \emph{monotonic morphism}
$f : \mathbf{A} \rightarrow \mathbf{B}$
is a $\mathcal{B}$-morphism $f : A \rightarrow B$ that preserves order: 
${\leq}_A \leq {\leq}_f^B$.
An \emph{isotonic morphism}
$f : \mathbf{A} \rightarrow \mathbf{B}$
is a morphism $f : A \rightarrow B$ that preserves and respects order:
${\leq}_A = {\leq}_f^B$.
A parallel pair of monotonic morphisms
$f,g : \mathbf{A} \rightarrow \mathbf{B}$
is ordered $f \leq g$ when
$(f,g) \cdot {\leq}_B : A \rightarrow B{\times}B \rightarrow \Omega$ 
is the top character $\top_A : A \rightarrow \Omega$.
The composition and identities of monotonic morphisms can be defined in terms of the underlying $\mathcal{B}$-objects and $\mathcal{B}$-morphisms.
Let $\mathsf{Ord}(\mathcal{B})$ denote the category of preorders and monotonic morphisms.
There is an underlying functor
$|\mbox{-}|_{\mathcal{B}} : \mathsf{Ord}(\mathcal{B}) \rightarrow \mathcal{B}$,
which gives the underlying $\mathcal{B}$-object of a preorder 
and the underlying $\mathcal{B}$-morphism of a monotonic morphism.

A preorder $\mathbf{A} = \langle A, \leq_A \rangle$
can equivalently be regarded as a monotonic $\mathcal{B}$-morphism in two ways:
either (1) the \emph{up segment} monotonic morphism
${\uparrow}_{\mathbf{A}} = {\leq}_A^{01} : \mathbf{A} \rightarrow {\wp}\mathbf{A}^{\propto}$
or (2) the \emph{down segment} monotonic morphism 
${\downarrow}_{\mathbf{A}} = {\leq}_A^{10} : \mathbf{A} \rightarrow {\wp}\mathbf{A}$.
For any preorder $\mathbf{A} = \langle A, \leq_A \rangle$,
the quotient epimorphism 
$[\mbox{-}] : A \rightarrow A/{\equiv_A}$
is a monotonic morphism
$[\mbox{-}]_A : \mathbf{A} \rightarrow \mathsf{quo}(\mathbf{A})$.
For any monotonic morphism
$f : \mathbf{A} \rightarrow \mathbf{B}$,
since the composite
$f \cdot [\mbox{-}]_B : A \rightarrow B \rightarrow \mathsf{quo}(\mathbf{B})$
respects the equivalence relation of $\mathbf{A}$,
there is a quotient monotonic morphism
$\mathsf{quo}(f) = [f] : \mathsf{quo}(\mathbf{A}) \rightarrow \mathsf{quo}(\mathbf{B})$
that satisfies the naturality condition
$f \cdot [\mbox{-}]_B = [\mbox{-}]_A \cdot [f]$.

\paragraph{Order Bimodules.}
An \emph{order left semimodule}
$\mathbf{r} : \mathbf{A} \rightharpoondown B$
is a relation $r : A \rightharpoondown B$
that is closed on the left (at the source) $({\leq}_A \circ {\leq}_r) \leq {\leq}_r$.
An \emph{order right semimodule}
$\mathbf{r} : A \rightharpoondown \mathbf{B}$
is a relation $r : A \rightharpoondown B$
that is closed on the right (at the target) $({\leq}_r \circ {\leq}_B) \leq {\leq}_r$.
An \emph{order bimodule}
$\mathbf{r} : \mathbf{A} \rightharpoondown \mathbf{B}$
is both a left and right semimodule;
that is,
it is closed on the left and on the right.
Any left semimodule $\mathbf{r} : \mathbf{A} \rightharpoondown B$
is a bimodule $\mathbf{r} : \mathbf{A} \rightharpoondown B$,
where $B = \langle B, \bot_{B{\times}B} \rangle$ is the discrete order.
Any right semimodule $\mathbf{r} : A \rightharpoondown \mathbf{B}$
is a bimodule $\mathbf{r} : A \rightharpoondown \mathbf{B}$.
For any order bimodule
$\mathbf{r} : \mathbf{A} \rightharpoondown \mathbf{B}$,
the 01-fiber is a (contravariant) monotonic function
$\mathbf{r}^{01} : \mathbf{A} \rightarrow {\wp}\mathbf{B}^{\propto}$
and the 10-fiber is a (covariant) monotonic function
$\mathbf{r}^{10} : \mathbf{B} \rightarrow {\wp}\mathbf{A}$.
Any monotonic morphism $f : \mathbf{A} \rightarrow \mathbf{B}$
defines an order bimodule in each direction:
(1) the \emph{forward bimodule}
$f^\triangleright = \mathbf{B}(f,1_{\mathbf{B}})
: \mathbf{A} \rightharpoondown \mathbf{B}$,
which has 01-fiber the composition of itself with up segment
$(\mathbf{f}^\triangleright)^{01}
= \mathbf{f} \cdot {\uparrow}_{\mathbf{B}}
: \mathbf{A} \rightarrow \mathbf{B} \rightarrow {\wp}\mathbf{B}$, 
and has 10-fiber the composition of down segment with inverse image
$(\mathbf{f}^\triangleright)^{10}
= {\downarrow}_{\mathbf{B}} \cdot \mathbf{f}^{-1}
 : \mathbf{B} \rightarrow {\wp}\mathbf{B} \rightarrow {\wp}\mathbf{A}$; and
(2) the \emph{reverse bimodule}
$f^\triangleleft  = \mathbf{B}(1_{\mathbf{B}},f)
: \mathbf{B} \rightharpoondown \mathbf{A}$,
which has 01-fiber the composition of up segment with inverse image
$(\mathbf{f}^\triangleleft)^{01} 
= {\uparrow}_{\mathbf{A}} \cdot \mathbf{f}^{-1}
: \mathbf{B} \rightarrow {\wp}\mathbf{B} \rightarrow {\wp}\mathbf{A}^{\propto}$,
and has 10-fiber the composition of itself with down segment
$(\mathbf{f}^\triangleleft)^{10}
= \mathbf{f} \cdot {\downarrow}_{\mathbf{A}}
: \mathbf{A} \rightarrow \mathbf{B} \rightarrow {\wp}\mathbf{B}$.

Let $\mathsf{Ord}(\mathcal{B})_{=} \subset \mathsf{Ord}(\mathcal{B})$ 
denote the full subcategory of partial orders and monotonic $\mathcal{B}$-morphisms\footnote{Also denoted $\mathsf{Pos}(\mathcal{B})$}.
There is an inclusion functor
$\mathsf{incl}_{\mathcal{B}} 
: \mathsf{Ord}(\mathcal{B})_{=} \rightarrow \mathsf{Ord}(\mathcal{B})$
and a quotient functor
$\mathsf{quo}_{\mathcal{B}} 
: \mathsf{Ord}(\mathcal{B}) \rightarrow \mathsf{Ord}(\mathcal{B})_{=}$.
There is a canon(ical) natural transformation
$\eta_{\mathcal{B}} 
: \mathsf{id}_{\mathsf{Ord}(\mathcal{B})} \Rightarrow \mathsf{quo}_{\mathcal{B}} \circ \mathsf{incl}_{\mathcal{B}}
: \mathsf{Ord}(\mathcal{B}) \rightarrow \mathsf{Ord}(\mathcal{B})$
whose $\mathbf{A}^{\mathrm{th}}$-component is the epimorphic canonical isotone
$[\mbox{-}]_{\mathbf{A}} : \mathsf{A} \rightarrow \mathsf{quo}(\mathbf{A})$.
The quotient functor is left adjoint to the inclusion functor 
$\mathsf{quo}_{\mathcal{B}} \dashv \mathsf{incl}_{\mathcal{B}}$
with counit being an isomorphism and unit being the canon.
This adjunction is a  reflection: 
$\mathsf{Ord}(\mathcal{B})_{=}$ is a reflective subcategory of $\mathsf{Ord}(\mathcal{B})$ 
with the quotient functor being the reflector (Figure~\ref{order-fibration}).

A $\mathcal{B}$-monotonic morphism $f : {\mathbf A} \rightarrow {\mathbf B}$ 
is an \emph{isomorphism} $f : {\mathbf A} \cong {\mathbf B}$
when there is an oppositely-directed $\mathcal{B}$-monotonic morphism 
$f^{-1} : {\mathbf B} \rightarrow {\mathbf A}$
called its inverse
such that $f \cdot f^{-1} = 1_A$ and $f^{-1} \cdot f = 1_B$.
A $\mathcal{B}$-monotonic morphism $f : {\mathbf A} \rightarrow {\mathbf B}$ 
is an \emph{equivalence} $f : {\mathbf A} \equiv {\mathbf B}$
when there is an oppositely-directed $\mathcal{B}$-monotonic morphism 
$f^\prime : {\mathbf B} \rightarrow {\mathbf A}$
called its pseudo-inverse
such that $f \cdot f^\prime \equiv 1_A$ and $f^\prime \cdot f \equiv 1_B$.
A $\mathcal{B}$-monotonic morphism $e : {\mathbf A} \rightarrow {\mathbf B}$ 
is a \emph{pseudo-epimorphism}
when for any parallel pair of $\mathcal{B}$-monotonic morphisms 
$f, g : {\mathbf B} \rightarrow {\mathbf C}$,
if $e \cdot f \equiv e \cdot g$ then $f \equiv g$.
There is a dual definition for a \emph{pseudo-monomorphism}.

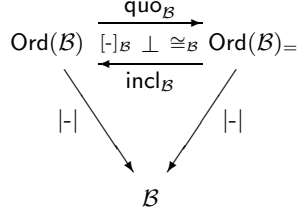
\begin{figure}
\begin{center}
\setlength{\unitlength}{0.65pt}
\begin{picture}(120,120)(-25,10)
\put(-30,75){\makebox(60,30){$\mathsf{Ord}(\mathcal{B})$}}
\put(90,75){\makebox(60,30){$\mathsf{Ord}(\mathcal{B})_{=}$}}
\put(30,-15){\makebox(60,30){$\mathcal{B}$}}
\put(45,100){\makebox(30,20){\footnotesize{$\mathsf{quo}_{\mathcal{B}}$}}}
\put(25,80){\makebox(30,20){\scriptsize{$[\mbox{-}]_{\mathcal{B}}$}}}
\put(45,80){\makebox(30,20){\footnotesize{$\bot$}}}
\put(65,80){\makebox(30,20){\scriptsize{$\cong_{\mathcal{B}}$}}}
\put(45,58){\makebox(30,20){\footnotesize{$\mathsf{incl}_{\mathcal{B}}$}}}
\put(-3,35){\makebox(30,20){\footnotesize{$|\mbox{-}|$}}}
\put(93,35){\makebox(30,20){\footnotesize{$|\mbox{-}|$}}}
\put(30,102){\vector(1,0){60}}
\put(90,78){\vector(-1,0){60}}
\put(10,75){\vector(2,-3){40}}
\put(110,75){\vector(-2,-3){40}}
\end{picture}
\end{center}
\caption{Order Fibration}
\label{order-fibration}
\end{figure}

\paragraph{Finite Limits.}
Given two preorders $\mathbf{A}_1$ and $\mathbf{A}_2$,
the \emph{binary product} is the preorder 
$\mathbf{A}_1 {\times} \mathbf{A}_2
= \langle A_1{\times}A_2, \leq_{A_1{\times}A_2} \rangle$,
whose order relation is defined by 
\begin{center}
$\begin{array}{r@{\hspace{5pt}}c@{\hspace{5pt}}l}
\leq_{A_1{\times}A_2}
& \doteq & \Delta_{A_1{\times}A_2} \cdot (1_{A_1}{\times}\tau_{A_1,A_2}{\times}1_{A_2}) \cdot (\leq_{A_1}{\times}\leq_{A_1}) \cdot \wedge
\\
& : & A_1{\times}A_2 \rightarrow  A_1{\times}A_2{\times}A_1{\times}A_2
\rightarrow A_1{\times}A_1{\times}A_2{\times}A_2 \rightarrow \Omega{\times}\Omega \rightarrow \Omega .
\end{array}$
\end{center}
The underlying component projectio $\mathcal{B}$-morphisms are monotonic: 
$\pi_1 : \mathbf{A}_1 {\times} \mathbf{A}_2 \rightarrow \mathbf{A}_1$
and 
$\pi_2 : \mathbf{A}_1 {\times} \mathbf{A}_2 \rightarrow \mathbf{A}_2$.
This is a categorical product in $\mathsf{Ord}(\mathcal{B})$,
since given any pair of monotonic $\mathcal{B}$-morphisms
$\mathbf{f}_1 : \mathbf{C} \rightarrow \mathbf{A}_1$,
and $\mathbf{f}_2 : \mathbf{C} \rightarrow \mathbf{A}_2$ 
with common source,
the unique mediating $\mathcal{B}$-morphism 
that satisfies
$\mathbf{f} \cdot \pi_1 = \mathbf{f}_1$ and $\mathbf{f} \cdot \pi_2 = \mathbf{f}_2$
is monotonic:
$\mathbf{f} = (\mathbf{f}_1,\mathbf{f}_2) : \mathbf{C} \rightarrow \mathbf{A}_1 {\times} \mathbf{A}_2$.
This definition can be extended to any finite number of preorders.
Also, the finite product of partial orders is a partial order.
The terminal $\mathcal{B}$-object $1$ forms a partial order 
$\mathbf{1} = \langle 1,\top_{1{\times}1} \rangle$ that is the nullary product,
since for any preorder $\mathbf{A}$ the unique $\mathcal{B}$-morphism is monotonic: 
$!_A : \mathbf{A} \rightarrow \mathbf{1}$.
Given any parallel pair of monotonic $\mathcal{B}$-morphisms
$\mathbf{f},\mathbf{g} : \mathbf{A} \rightarrow \mathbf{B}$,
the $\mathcal{B}$-equalizer
$e : E \rightarrow A$
lifts to an \emph{equalizer} 
$e : \mathbf{E} = \langle E, \leq_e \rangle \rightarrow \mathbf{A}$
in $\mathsf{Ord}$,
where the order relation is the kernel of $e$
and hence $e$ is a monic isotone.
Hence,
the categories $\mathsf{Ord}(\mathcal{B})$ and $\mathsf{Ord}(\mathcal{B})_{=}$ are finite complete,
and the underlying functors preserve these limits.

\paragraph{Power.}
Let $A$ be any $\mathcal{B}$-object.
The \emph{power preorder} 
${\wp}\mathbf{A} = \langle {\wp}A, \leq_{{\wp}A} \rangle$
is the power object with the \emph{inclusion} order,
whose subobject is
$\iota_\Omega^P = {(\mbox{-})} \cdot \iota_\Omega 
: \Box_\Omega^P \hookrightarrow {(\Omega{\times}\Omega)}^P \cong \Omega^P\!{\times}\Omega^P$.
The \emph{binary intersection} $\mathcal{B}$-morphism
$\cap_A : {\wp}A{\times}{\wp}A \rightarrow {\wp}A$
is defined, using conjunction $\cap$ on $\Omega$, to be the exponential adjoint of the $\mathcal{B}$-morphism
$(\Delta_A{\times}1_{({\wp}A{\times}{\wp}A)}) 
\cdot (1_{A}{\times}\tau_{A,{\wp}A}{\times}1_{{\wp}A})
\cdot ({\in}_{A}{\times}{\in}_{A}) \cdot {\cap}
: A{\times}{\wp}A{\times}{\wp}A 
\rightarrow A{\times}A{\times}{\wp}A{\times}{\wp}A 
\rightarrow A{\times}{\wp}A{\times}A{\times}{\wp}A 
\rightarrow \Omega{\times}\Omega \rightarrow \Omega$.
The \emph{binary union} and \emph{relative pseudo-complement} 
${\cup}_A, {\Rightarrow}_A : {\wp}A{\times}{\wp}A \rightarrow {\wp}A$
have similar definitions using disjunction $\cup$ and implication $\Rightarrow$ on $\Omega$.
Using the idea that a lattice element is smaller than another element when the meet is the first,
the inclusion relation
${\leq}_A : {\wp}A \rightharpoondown {\wp}A$
can also be defined via the character
$(\Delta_{{\wp}A}{\times}1_{{\wp}A}) \cdot (1_{{\wp}A}{\times}\cap_{A}) \cdot \delta_{{\wp}A}
: {\wp}A{\times}{\wp}A \rightarrow {\wp}A{\times}{\wp}A{\times}{\wp}A 
\rightarrow {\wp}A{\times}{\wp}A \rightarrow \Omega$.
The \emph{intersection} monotonic morphism
$\cap_A : {\wp}{\wp}A^{\propto} \rightarrow {\wp}A$
is the 01-fiber of the left residuation 
${\in}_{{\wp}A} {\setminus}\, {\in}_{A}^{\propto} 
: {\wp}{\wp}A \rightharpoondown A$
of the opposite of the basic membership relation
$\in_{A}^{\propto} : {\wp}A \rightharpoondown A$
along the membership relation on power
$\in_{{\wp}A} : {\wp}A \rightharpoondown {\wp}{\wp}A$.
Of course,
intersection could also be defined with right residuation.
The \emph{union} monotonic morphism
$\cup_A : {\wp}{\wp}A \rightarrow {\wp}A$
is the 10-fiber of the composition 
${\in}_{A} \!\circ {\in}_{{\wp}A}
: A \rightharpoondown {\wp}{\wp}A$
of the basic membership relation
$\in_{A} : A \rightharpoondown {\wp}A$
with the membership relation on power
$\in_{{\wp}A} : {\wp}A \rightharpoondown {\wp}{\wp}A$.
{\bfseries Prove:}
The tuple 
${\wp}A = \langle {\wp}A, {\subseteq}_{A}, {\cup}_{A}, {\cap}_{A}, {\Rightarrow}_{A} \rangle$
forms a complete Heyting algebra in (internal to) $\mathcal{B}$. 

%%%%%%%%%%%%%%%%%%%%%%%%%%%%%%%%%%%%%%%%%%%%%%%%%%%%%%%%%%%%%%%%%%%%%%
%%%%%%%%%%%%%%%%%%%%%%%%%%%%%%%%%%%%%%%%%%%%%%%%%%%%%%%%%%%%%%%%%%%%%%
\subsection{Order Adjunctions.}
%%%%%%%%%%%%%%%%%%%%%%%%%%%%%%%%%%%%%%%%%%%%%%%%%%%%%%%%%%%%%%%%%%%%%%
%%%%%%%%%%%%%%%%%%%%%%%%%%%%%%%%%%%%%%%%%%%%%%%%%%%%%%%%%%%%%%%%%%%%%%

An \emph{(order) adjunction}
$\mathbf{g} 
= \langle \check{\mathbf{g}}, \hat{\mathbf{g}} \rangle
: \mathbf{A}_0 \rightleftharpoons \mathbf{A}_1$
in (internal to) a topos $\mathcal{B}$
consists of 
a left adjoint monotonic morphism in the forward direction
$\check{\mathbf{g}} : \mathbf{A}_{0} \rightarrow \mathbf{A}_{1}$
and a right adjoint monotonic morphism in the reverse direction
$\hat{\mathbf{g}} : \mathbf{A}_{1} \rightarrow \mathbf{A}_{0}$
that satisfy any of the following equivalent conditions:
\begin{description}
\item [fundamental:]
$\mathbf{A}_{1}(\check{\mathbf{g}},1_{\mathbf{A}_{1}})
= \mathbf{A}_{0}(1_{\mathbf{A}_{0}},\hat{\mathbf{g}})$; 
or equivalently,
\item [external:]
$\mathbf{A}_{1}(a \cdot \check{\mathbf{g}}, b)
= \mathbf{A}_{0}(a, b \cdot \hat{\mathbf{g}})$
for every preorder $\mathbf{C}$ and every pair of elements
$a \in^{\mathbf{C}} \mathbf{A}_0$ and $b \in^{\mathbf{C}} \mathbf{A}_1$;
or equivalently,
\item [closure/interior:]
$1_A \leq \check{\mathbf{g}} \cdot \hat{\mathbf{g}}$ 
and $\hat{\mathbf{g}} \cdot \check{\mathbf{g}} \leq 1_B$;
or equivalently,
\item [factor:]
$(1_A, \check{\mathbf{g}} \cdot \hat{\mathbf{g}}) : A \rightarrow A{\times}A$
factors through the subobject of $\leq_A$
and
$(\hat{\mathbf{g}} \cdot \check{\mathbf{g}}, 1_B) : B \rightarrow B{\times}B$
factors through the subobject of $\leq_B$.
\end{description}
By using derivation on the order relation,
any preorder
$\mathbf{A} = \langle A, {\leq}_{\mathbf{A}} \rangle$
defines the \emph{bound} adjunction 
$\mathsf{bnd}_{\mathbf{A}}
= \langle {\Uparrow}_{\mathbf{A}}^{\propto}, {\Downarrow}_{\mathbf{A}} \rangle 
: {\wp}\mathbf{A} \rightarrow {\wp}\mathbf{A}^{\propto}$
where the left adjoint is the \emph{upper bound} monotonic function
${\Uparrow}_{\mathbf{A}} 
= {\leq}_{\mathbf{A}}^{\Rightarrow} 
= \exists {\leq}_{\mathbf{A}}^{01} \cdot {\cap}_{\mathbf{A}}
= \exists {\uparrow}_{\mathbf{A}} \cdot {\cap}_{\mathbf{A}}
: {\wp}\mathbf{A}^{\propto} \rightarrow {\wp}\mathbf{A}$
and the right adjoint is the \emph{lower bound} monotonic function
${\Downarrow}_{\mathbf{A}} 
= {\leq}_{\mathbf{A}}^{\Leftarrow} 
= \exists {\leq}_{\mathbf{A}}^{10} \cdot {\cap}_{A}
= \exists {\downarrow}_{\mathbf{A}} \cdot {\cap}_{A}
: {\wp}\mathbf{A}^{\propto} \rightarrow {\wp}\mathbf{A}$.
Externally,
a closed subobject $X \subseteq^1 A$ of the bound adjunction is the object of lower bounds 
${\Downarrow}_{\mathbf{A}} Y$ for some subobject $Y \subseteq^1 A$,
and an open subobject $Y \subseteq^1 A$ is the object of upper bounds 
${\Uparrow}_{\mathbf{A}} X$ for some subobject $X \subseteq^1 A$.

Composition\footnote{We use the symbol ``$\circ$'' for the composition of adjunctions.} and identities of adjunctions are defined componentwise.
Let $\mathsf{Adj}(\mathcal{B})$ denote the category of preorders and adjunctions.
Partial orders and adjunctions form the full subcategory 
$\mathsf{Adj}(\mathcal{B})_{=} \subset \mathsf{Adj}(\mathcal{B})$. 
Projecting to the left and right gives rise to two component functors.
The left functor 
$\mathsf{left}_{\mathcal{B}} 
: \mathsf{Adj}(\mathcal{B}) \rightarrow \mathsf{Ord}(\mathcal{B})$
is the identity on objects and maps an adjunction
$\mathbf{g} : \mathbf{A}_0 \rightleftharpoons \mathbf{A}_1$
to its left component
$\mathsf{left}_{\mathcal{B}}(\mathbf{g}) = \check{\mathbf{g}} 
: \mathbf{A}_0 \rightarrow \mathbf{A}_1$.
The right functor 
$\mathsf{right}_{\mathcal{B}} 
: \mathsf{Adj}(\mathcal{B}) \rightarrow \mathsf{Ord}(\mathcal{B})$
is the identity on objects and maps an adjunction
$\mathbf{g} : \mathbf{A}_0 \rightleftharpoons \mathbf{A}_1$
to its right component
$\mathsf{right}_{\mathcal{B}}(\mathbf{g}) = \hat{\mathbf{g}}
: \mathbf{A}_1 \rightarrow \mathbf{A}_0$.
We use the same notation for the underlying components
$\mathsf{left}_{\mathcal{B}} = \mathsf{left}_{\mathcal{B}} \circ |\mbox{-}| 
: \mathsf{Adj}(\mathcal{B}) \rightarrow \mathcal{B}$
and
$\mathsf{right}_{\mathcal{B}} = \mathsf{right}_{\mathcal{B}} \circ |\mbox{-}| 
: \mathsf{Adj}(\mathcal{B}) \rightarrow \mathcal{B}$.
The order-enriched \emph{involution} isomorphism
${(\mbox{-})}_{\mathcal{B}}^{\propto} 
: \mathsf{Adj}(\mathcal{B})^{\mathrm{op}} \rightarrow \mathsf{Adj}(\mathcal{B})$
flips source/target and left/right:
${(\mbox{-})}_{\mathcal{B}}^{\propto} \!\circ\, \mathsf{left}_{\mathcal{B}} 
= \mathsf{right}_{\mathcal{B}}$ 
and
${(\mbox{-})}_{\mathcal{B}}^{\propto} \!\circ\, \mathsf{right}_{\mathcal{B}}^{\mathrm{op}} 
= \mathsf{left}_{\mathcal{B}}^{\mathrm{op}}$.

\paragraph{Interior/Closure.}
Let $\mathbf{g} : \mathbf{A}_0 \rightleftharpoons \mathbf{A}_1$ be any $\mathcal{B}$-adjunction between partial orders.
The \emph{closure} of $\mathbf{g}$ is the $\mathcal{B}$-monotonic endomorphism
$(\mbox{-})^{\bullet_{\mathbf{g}}} 
= \check{\mathbf{g}} \cdot \hat{\mathbf{g}}
: \mathbf{A}_0 \rightarrow \mathbf{A}_0$.
Closure is increasing $1_A \leq (\mbox{-})^{\bullet_{\mathbf{g}}}$ and idempotent $(\mbox{-})^{\bullet_{\mathbf{g}}} \cdot (\mbox{-})^{\bullet_{\mathbf{g}}} = (\mbox{-})^{\bullet_{\mathbf{g}}}$.
Idempotency is implied by the fact that
$\check{\mathbf{g}} \cdot \hat{\mathbf{g}} \cdot \check{\mathbf{g}} = \check{\mathbf{g}}$.
The closure equalizer diagram in $\mathsf{Ord}(\mathcal{B})$ is the parallel pair
$1_A, (\mbox{-})^{\bullet_{\mathbf{g}}} : \mathbf{A}_0 \rightarrow \mathbf{A}_0$.
The internal suborder of closed elements of $\mathbf{g}$
is defined to be the equalizer 
$\mathrm{incl}_0^{\mathbf{g}} : \mathsf{clo}(\mathbf{g}) \rightarrow \mathbf{A}_0$
of this diagram.
Being part of a limiting cone, 
$\mathrm{incl}_0^{\mathbf{g}} \cdot (\mbox{-})^{\bullet_{\mathbf{g}}} 
= \mathrm{incl}_0^{\mathbf{g}}$.
The closure of any $A$-element $a : 1 \rightarrow A$
is the $A$-element 
$a^{\bullet_{\mathbf{g}}} = a \cdot (\mbox{-})^{\bullet_{\mathbf{g}}} : 1 \rightarrow A$.
An element $a : 1 \rightarrow A_0$ is a closed element of $\mathbf{g}$
when it factors through $\mathsf{clo}(\mathbf{g})$;
that is,
there is an element $\bar{a} : 1 \rightarrow \mathsf{clo}(\mathbf{g})$ 
such that $a$ is equal to its inclusion $a = \bar{a} \cdot \mathrm{incl}_0^{\mathbf{g}}$;
or equivalently,
when $a$ is equal to its closure $a = a^{\bullet_{\mathbf{g}}}$);
or equivalently,
when $a$ is equal $a = b \cdot \hat{\mathbf{g}}$ 
to the image of some target element $b : 1 \rightarrow A_1$.
Dually,
the \emph{interior} of $\mathbf{g}$ is the $\mathcal{B}$-monotonic endomorphism
$(\mbox{-})^{\circ_{\mathbf{g}}}
= \hat{\mathbf{g}} \cdot \check{\mathbf{g}}
: \mathbf{A}_1 \rightarrow \mathbf{A}_1$.
Interior is decreasing $1_B \geq (\mbox{-})^{\circ_{\mathbf{g}}}$ and idempotent $(\mbox{-})^{\circ_{\mathbf{g}}} \cdot (\mbox{-})^{\circ_{\mathbf{g}}} = (\mbox{-})^{\circ_{\mathbf{g}}}$.
Idempotency is implied by the fact that
$\hat{\mathbf{g}} \cdot \check{\mathbf{g}} \cdot \hat{\mathbf{g}} = \hat{\mathbf{g}}$.
The interior equalizer diagram in $\mathsf{Ord}(\mathcal{B})$ is the parallel pair
$1_B, (\mbox{-})^{\circ_{\mathbf{g}}} : \mathbf{A}_1 \rightarrow \mathbf{A}_1$.
The internal suborder of open elements of $\mathbf{g}$
is defined to be the equalizer 
$\mathrm{incl}_1^{\mathbf{g}} : \mathsf{open}(\mathbf{g}) \rightarrow \mathbf{A}_1$
of this diagram.
Being part of a limiting cone, 
$\mathrm{incl}_1^{\mathbf{g}} \cdot (\mbox{-})^{\circ_{\mathbf{g}}} 
= \mathrm{incl}_1^{\mathbf{g}}$.
The interior of any $B$-element $b : 1 \rightarrow B$
is the $B$-element 
$b^{\circ_{\mathbf{g}}} = b \cdot (\mbox{-})^{\circ_{\mathbf{g}}} : 1 \rightarrow B$.
An element $b : 1 \rightarrow B$ is an open element of $\mathbf{g}$
when it factors through $\mathsf{open}(\mathbf{g})$;
that is,
there is an element $\tilde{b} : 1 \rightarrow \mathsf{open}(\mathbf{g})$ 
such that $b$ is equal to its inclusion $b = \tilde{b} \cdot \mathrm{incl}_1^{\mathbf{g}}$;
or equivalently,
when $b$ is equal to its interior $b = b^{\circ_{\mathbf{g}}}$;
or equivalently,
when $b$ is equal $b = a \cdot \check{\mathbf{g}}$ 
to the image of some source element $a : 1 \rightarrow A_0$.

\paragraph{Reflections/Coreflections.}
A $\mathcal{B}$-\emph{pseudo-reflection} is a $\mathcal{B}$-adjunction 
$\mathbf{g} : \mathbf{A}_0 \rightleftharpoons \mathbf{A}_1$
that satisfies the equivalence
$1_{B} \equiv (\mbox{-})^{\circ_{\mathbf{g}}}$.
The left adjoint of a pseudo-reflection is a pseudo-epimorphism,
and the right adjoint is a pseudo-monomorphism.
A $\mathcal{B}$-\emph{reflection} is a $\mathcal{B}$-pseudo-reflection that is strict: 
it satisfies the identity
$1_{B} = (\mbox{-})^{\circ_{\mathbf{g}}}$. 
The right adjoint of a reflection is an isotonic morphism.
If $\mathbf{g} : \mathbf{A}_0 \rightleftharpoons \mathbf{A}_1$ is a $\mathcal{B}$-reflection
and the source $\mathbf{A}_0$ is a partial order,
then the target $\mathbf{A}_1$ is also a partial order.
Let $\mathsf{Ref}(\mathcal{B})$ denote the morphism subclass of all $\mathcal{B}$-reflections.
A $\mathcal{B}$-\emph{pseudo-coreflection} is a $\mathcal{B}$-adjunction 
$\mathbf{g} : \mathbf{A}_0 \rightleftharpoons \mathbf{A}_1$
that satisfies the equivalence
$1_{A} \equiv (\mbox{-})^{\bullet_{\mathbf{g}}}$. 
The left adjoint of a pseudo-coreflection is a pseudo-monomorphism,
and the right adjoint is a pseudo-epiomorphism.
A $\mathcal{B}$-\emph{coreflection} is a $\mathcal{B}$-pseudo-coreflection that is strict:
it satisfies the identity
$1_{A} = (\mbox{-})^{\bullet_{\mathbf{g}}}$. 
The left adjoint of a coreflection is an isotonic morphism.
If $\mathbf{g} : \mathbf{A}_0 \rightleftharpoons \mathbf{A}_1$ is a $\mathcal{B}$-coreflection
and the target $\mathbf{A}_1$ is a partial order,
then the source $\mathbf{A}_0$ is also a partial order.
Let $\mathsf{Ref}^\propto(\mathcal{B})$ denote the morphism subclass of all $\mathcal{B}$-coreflections.
The involution of a $\mathcal{B}$-pseudo-reflection is a $\mathcal{B}$-pseudo-coreflection, and vice-versa.

Let $\mathbf{g} : \mathbf{A}_0 \rightleftharpoons \mathbf{A}_1$ be a $\mathcal{B}$-adjunction.
A \emph{bipole} (bipolar pair) $(a, b)$ is a pair consisting of
a closed element $a : 1 \rightarrow \mathsf{clo}(\mathbf{g})$ and 
an open element $b : 1 \rightarrow \mathsf{open}(\mathbf{g})$,
where $a = b \cdot \hat{\mathbf{g}}$ 
(equivalently, $a \cdot \check{\mathbf{g}} = b$).
Define the external bipolar order
$(a_1, b_1) \leq (a_2, b_2)$
when $a_1 \leq_{\mathbf{A}_0} a_2$
(equivalently, when $b_1 \leq_{\mathbf{A}_1} b_2$).
There is an internal representaton for this bipolar order.

%%%%%%%%%%%%%%%%%%%%%%%%%%%%%%%%%%%%%%%%%%%%%%%%%%%%%%%%%%%%
\paragraph{The Polar Factorization.}
%%%%%%%%%%%%%%%%%%%%%%%%%%%%%%%%%%%%%%%%%%%%%%%%%%%%%%%%%%%%

Consider the axis diagram in $\mathsf{Ord}(\mathcal{B})$
consisting of the two opspans
$\mathrm{incl}_0^{\mathbf{g}} : \mathsf{clo}(\mathbf{g}) \rightarrow A_0
\leftarrow \mathsf{open}(\mathbf{g}) : \hat{\mathbf{g}}_1$
and
$\check{\mathbf{g}}_0 : \mathsf{clo}(\mathbf{g}) \rightarrow A_1
\leftarrow \mathsf{open}(\mathbf{g}) : \mathrm{incl}_1^{\mathbf{g}}$.
The \emph{axis} preorder $\diamondsuit(\mathbf{g})$ is the pullback of this diagram.
It comes equipped with two pullback projections
$\tilde{\pi}_0^{\mathbf{g}} : \diamondsuit(\mathbf{g}) \rightarrow \mathsf{clo}(\mathbf{g})$
and
$\tilde{\pi}_1^{\mathbf{g}} : \diamondsuit(\mathbf{g}) \rightarrow \mathsf{open}(\mathbf{g})$.
These satisfy
$\tilde{\pi}_0^{\mathbf{g}} \cdot \mathrm{incl}_0^{\mathbf{g}}
= \tilde{\pi}_1^{\mathbf{g}} \cdot \hat{\mathbf{g}}_1$
and
$\tilde{\pi}_0^{\mathbf{g}} \cdot \check{\mathbf{g}}_0
= \tilde{\pi}_1^{\mathbf{g}} \cdot \mathrm{incl}_1^{\mathbf{g}}$.
Define the extended projections
$\pi_0^{\mathbf{g}}
= \tilde{\pi}_0^{\mathbf{g}} \cdot \mathrm{incl}_0^{\mathbf{g}} 
: \diamondsuit(\mathbf{g}) \rightarrow \mathsf{clo}(\mathbf{g}) \rightarrow A_0$
and
$\pi_1^{\mathbf{g}}
= \tilde{\pi}_1^{\mathbf{g}} \cdot \mathrm{incl}_1^{\mathbf{g}} 
: \diamondsuit(\mathbf{g}) \rightarrow \mathsf{open}(\mathbf{g}) \rightarrow A_1$.
Define the left adjoint source restriction
$\check{\mathbf{g}}_0 
\doteq \mathrm{incl}_0^{\mathbf{g}} \cdot \check{\mathbf{g}}
: \mathsf{clo}(\mathbf{g}) \rightarrow A_0 \rightarrow A_1$. 
Define the left adjoint target restriction
$\check{\mathbf{g}}_1 
\doteq \check{\mathbf{g}} \cdot (\mbox{-})_1^{\circ_{\mathbf{g}}}
: A_0 \rightarrow A_1 \rightarrow \mathsf{open}(\mathbf{g})$. 
The pair of monotonic morphisms
$(\mbox{-})_0^{\bullet_{\mathbf{g}}} : A_0 \rightarrow \mathsf{clo}(\mathbf{g})$
and
$\check{\mathbf{g}}_1 : A_0 \rightarrow \mathsf{open}(\mathbf{g})$
forms a cone for the axis diagram,
since
$(\mbox{-})_0^{\bullet_{\mathbf{g}}} \cdot \mathrm{incl}_0^{\mathbf{g}}
= (\mbox{-})^{\bullet_{\mathbf{g}}}
= \check{\mathbf{g}}_1 \cdot \hat{\mathbf{g}}_1$
and
$(\mbox{-})_0^{\bullet_{\mathbf{g}}} \cdot \check{\mathbf{g}}_0
= \check{\mathbf{g}}
= \check{\mathbf{g}}_1 \cdot \mathrm{incl}_1^{\mathbf{g}}$.
Let
$\xi_0^{\mathbf{g}} : A_0 \rightarrow \diamondsuit(\mathbf{g})$
denote the mediating monotonic morphism for this cone;
so that
$\xi_0^{\mathbf{g}}$ is the unique monotonic morphism
such that
$\xi_0^{\mathbf{g}} \cdot \tilde{\pi}_0^{\mathbf{g}}
= (\mbox{-})_0^{\bullet_{\mathbf{g}}}$
and
$\xi_0^{\mathbf{g}} \cdot \tilde{\pi}_1^{\mathbf{g}}
= \check{\mathbf{g}}_1$.
The source embedding/projection pair form a reflection
$\mathsf{ref}(\mathbf{g})
= \langle \xi_0^{\mathbf{g}}, \pi_0^{\mathbf{g}} \rangle
: \mathbf{A}_0 \rightleftharpoons \diamondsuit(\mathbf{g})$
called the \emph{extent reflection} of $g$.
Dually,
the pair of monotonic morphisms
$\hat{g}_0 : A_1 \rightarrow \mathsf{clo}(\mathbf{g})$
and
$(\mbox{-})_1^{\circ_{g}} : A_1 \rightarrow \mathsf{open}(\mathbf{g})$
forms a cone for the axis diagram.
Let
$\xi_1^{g} : A_1 \rightarrow \diamondsuit(\mathbf{g})$
denote the unique mediating monotonic morphism for this cone.
The target projection/embedding pair form a coreflection
$\mathsf{ref}^\propto(\mathbf{g})
= \langle \pi_1^{\mathbf{g}}, \xi_1^{\mathbf{g}} \rangle
: \diamondsuit(\mathbf{g}) \rightleftharpoons \mathbf{A}_1$
called the \emph{intent reflection} of $g$.
The original adjunction factors in terms of 
its extent reflection and intent coreflection
$\mathbf{g} = \mathsf{ref}(\mathbf{g})\circ \mathsf{ref}^\propto(\mathbf{g})$.
The quintuple
$( \mathbf{A}_0, 
\mathsf{ref}(\mathbf{g}), \diamondsuit(\mathbf{g}), \mathsf{ref}^\propto(\mathbf{g}), 
\mathbf{A}_1 )$
is called the \emph{polar factorization} of $\mathbf{g}$.
Since both source $\mathbf{A}_0$ and target $\mathbf{A}_1$ are partial orders, 
the axis $\diamondsuit(\mathbf{g})$ is also a partial order.

\begin{lemma} [Diagonalization]
Assume that we are given a commutative square
$\mathbf{A}_0 \stackrel{\mathbf{e}}{\rightleftharpoons}
\mathbf{C}_1 \stackrel{\mathbf{s}}{\rightleftharpoons} \mathbf{A}_1 
\;=\; 
\mathbf{A}_0 \stackrel{\mathbf{r}}{\rightleftharpoons}
\mathbf{C}_2 \stackrel{\mathbf{m}}{\rightleftharpoons} \mathbf{A}_1$
of adjunctions between partial orders,
with reflection $\mathbf{e}$ and coreflection $\mathbf{m}$.
Then there is a unique adjunction $\mathbf{d} : \mathbf{C}_1 \rightleftharpoons \mathbf{C}_2$
with $\mathbf{e} \circ \mathbf{d} = \mathbf{r}$ and $\mathbf{d} \circ \mathbf{m} = \mathbf{s}$.
\end{lemma}

\begin{proof} 
The necessary conditions give the definitions
$\check{\mathbf{d}} \doteq \check{\mathbf{s}} \cdot \hat{\mathbf{m}} = \hat{\mathbf{e}} \cdot \check{\mathbf{r}}$
and $\hat{\mathbf{d}} \doteq \hat{\mathbf{r}} \cdot \check{\mathbf{e}} = \check{\mathbf{m}} \cdot \hat{\mathbf{s}}$.
Existence follows from these definitions.

In more detail,
the fundamental adjointness property,
the special conditions for (co) reflections
and the above commutative diagram,
resolve into the following identities and inequalities:
$\hat{\mathbf{e}} \cdot \check{\mathbf{e}} = 1_{\mathbf{B}}$,
$1_{\mathbf{A}_0} \leq \check{\mathbf{e}} \cdot \hat{\mathbf{e}}$,
$\hat{\mathbf{s}} \cdot \check{\mathbf{s}} \leq 1_{\mathbf{A}_1}$,
$1_{\mathbf{B}} \leq \check{\mathbf{s}} \cdot \hat{\mathbf{s}}$,
$\hat{\mathbf{r}} \cdot \check{\mathbf{r}} \leq 1_{\mathbf{C}}$,
$1_{\mathbf{A}_0} \leq \check{\mathbf{r}} \cdot \hat{\mathbf{r}}$,
$\hat{\mathbf{m}} \cdot \check{\mathbf{m}} \leq 1_{\mathbf{A}_1}$,
$1_{\mathbf{C}} = \check{\mathbf{m}} \cdot \hat{\mathbf{m}}$,
$\check{\mathbf{e}} \cdot \check{\mathbf{s}} =
\check{\mathbf{r}} \cdot \check{\mathbf{m}}$,
and
$\hat{\mathbf{m}} \cdot \hat{\mathbf{r}} =
\hat{\mathbf{s}} \cdot \hat{\mathbf{e}}$.
By suitable pre- and post-composition we can prove the identities:
$\check{\mathbf{e}} \cdot \check{\mathbf{s}} \cdot \hat{\mathbf{m}} 
= \check{\mathbf{r}}$,
$\check{\mathbf{m}} \cdot \hat{\mathbf{s}} \cdot \hat{\mathbf{e}}
= \hat{\mathbf{r}}$,
$\hat{\mathbf{m}} \cdot \hat{\mathbf{r}} \cdot \check{\mathbf{e}} 
= \hat{\mathbf{s}}$
and
$\hat{\mathbf{e}} \cdot \check{\mathbf{r}} \cdot \check{\mathbf{m}} 
= \check{\mathbf{s}}$,
(and then)
$\check{\mathbf{s}} \cdot \hat{\mathbf{m}} 
= \hat{\mathbf{e}} \cdot \check{\mathbf{r}}$ and
$\hat{\mathbf{r}} \cdot \check{\mathbf{e}} 
= \check{\mathbf{m}} \cdot \hat{\mathbf{s}}$.

[{\bfseries Existence}]
Define the $\mathcal{B}$-morphisms
\underline{$\check{\mathbf{d}}
\doteq 
\check{\mathbf{s}} \cdot \hat{\mathbf{m}} 
= \hat{\mathbf{e}} \cdot \check{\mathbf{r}}$}
and
\underline{$\hat{\mathbf{d}}
\doteq 
\hat{\mathbf{r}} \cdot \check{\mathbf{e}} 
= \check{\mathbf{m}} \cdot \hat{\mathbf{s}}$}.
The properties
$\hat{\mathbf{d}} \cdot \check{\mathbf{d}}
= \check{\mathbf{m}} \cdot \hat{\mathbf{s}} \cdot \hat{\mathbf{e}} \cdot \check{\mathbf{r}}
= \hat{\mathbf{r}} \cdot \check{\mathbf{r}}
\leq 1_{\mathbf{C}}$ and
$\check{\mathbf{d}} \cdot \hat{\mathbf{d}}
= \check{\mathbf{s}} \cdot \hat{\mathbf{m}} \cdot \hat{\mathbf{r}} \cdot \check{\mathbf{e}}
= \check{\mathbf{s}} \cdot \hat{\mathbf{s}}
\geq 1_{\mathbf{B}}$ 
show that
$\mathbf{d} = \langle \check{\mathbf{d}}, \hat{\mathbf{d}} \rangle 
: \mathbf{B} \rightleftharpoons \mathbf{C}$
is a $\mathcal{B}$-adjunction.
The properties
$\check{\mathbf{d}} \cdot \check{\mathbf{m}}
= \hat{\mathbf{e}} \cdot \check{\mathbf{r}} \cdot \check{\mathbf{m}}
= \check{\mathbf{s}}$
and
$\hat{\mathbf{m}} \cdot \hat{\mathbf{d}}
= \hat{\mathbf{m}} \cdot \hat{\mathbf{r}} \cdot \check{\mathbf{e}}
= \hat{\mathbf{s}}$
show that
$\mathbf{d}$
satisfies the required identity
$\mathbf{d} \circ \mathbf{m} = \mathbf{s}$.
The properties
$\check{\mathbf{e}} \cdot \check{\mathbf{d}}
= \check{\mathbf{e}} \cdot \check{\mathbf{s}} \cdot \hat{\mathbf{m}}
= \check{\mathbf{r}}$
and
$\hat{\mathbf{d}} \cdot \hat{\mathbf{e}}
= \check{\mathbf{m}} \cdot \hat{\mathbf{s}} \cdot \hat{\mathbf{e}}
= \hat{\mathbf{r}}$
show that
$\mathbf{d}$
satisfies the required identity
$\mathbf{e} \circ \mathbf{d} = \mathbf{r}$.

[{\bfseries Uniqueness}]
Suppose $\mathbf{b} = \langle \check{\mathbf{b}}, \hat{\mathbf{b}} \rangle
: \mathbf{B} \rightleftharpoons \mathbf{C}$
is another $\mathcal{B}$-adjunction satisfying the require identities
$\mathbf{e} \circ \mathbf{b} = \mathbf{r}$
and $\mathbf{b} \circ \mathbf{m} = \mathbf{s}$.
These identities resolve to the identities
$\check{\mathbf{e}} \cdot \check{\mathbf{b}} = \check{\mathbf{r}}$,
$\hat{\mathbf{b}} \cdot \hat{\mathbf{e}} = \hat{\mathbf{r}}$,
$\check{\mathbf{b}} \cdot \check{\mathbf{m}} = \check{\mathbf{s}}$,
and
$\hat{\mathbf{m}} \cdot \hat{\mathbf{b}} = \hat{\mathbf{s}}$.
Hence,
$\check{\mathbf{b}}
= \check{\mathbf{b}} \cdot \check{\mathbf{m}} \cdot \hat{\mathbf{m}}
= \check{\mathbf{s}} \cdot \hat{\mathbf{m}}
= \check{\mathbf{d}}$,
$\hat{\mathbf{b}}
= \hat{\mathbf{b}} \cdot \hat{\mathbf{e}} \cdot \check{\mathbf{e}}
= \hat{\mathbf{r}} \cdot \check{\mathbf{e}}
= \hat{\mathbf{d}}$
and thus
$\mathbf{b} = \mathbf{d}$.
\qed
\end{proof}

\begin{lemma} [Polar Factorization]
The classes $\mathsf{Ref}(\mathcal{B})$ and $\mathsf{Ref}(\mathcal{B})^\propto$ of reflections and coreflections form a factorization system for $\mathsf{Adj}(\mathcal{B})_{=}$.
The polar factorization makes this a factorization system with choice.
\end{lemma}

\begin{proof} 
The previous discussion and lemma.
$\Box$
\end{proof}

\noindent
With this result,
we can specialize the discussion of section~\ref{sec:factorization:systems}
to the case $\mathcal{C} = \mathsf{Adj}(\mathcal{B})_{=}$.
The arrow category $\mathsf{Adj}(\mathcal{B})_{=}^{\mathsf{2}}$
has adjunctions
$(\mathbf{A}, \mathbf{g}, \mathbf{B})$
as objects
and pairs of adjunctions
$(\mathbf{a}, \mathbf{b}) : (\mathbf{A}_1, \mathbf{g}_1, \mathbf{B}_1) \rightarrow (\mathbf{A}_2, \mathbf{g}_2, \mathbf{B}_2)$ 
forming a commutative diagram $\mathbf{a} \circ \mathbf{g}_2 = \mathbf{g}_1 \circ \mathbf{b}$ as morphisms.
The factorization category 
$\mathsf{Ref}(\mathcal{B}) \odot \mathsf{Ref}(\mathcal{B})^\propto$
has reflection-coreflection factorizations 
$(\mathbf{A}, \mathbf{e}, \mathbf{C}, \mathbf{m}, \mathbf{B})$
as objects
and triples of adjunctions
$(\mathbf{a}, \mathbf{c}, \mathbf{b}) : (\mathbf{A}_1, \mathbf{e}_1, \mathbf{C_1}, \mathbf{m}_1, \mathbf{B}_1) \rightarrow (\mathbf{A}_2, \mathbf{e}_2, \mathbf{C}_2, \mathbf{m}_2, \mathbf{B}_2)$
forming commutative diagrams 
$\mathbf{a} \circ \mathbf{e}_2 = \mathbf{e}_1 \circ \mathbf{c}$ 
and $\mathbf{c} \circ \mathbf{m}_2 = \mathbf{m}_1 \circ \mathbf{b}$ as morphisms.
The polar factorization functor
$\div_{\mathsf{Adj}(\mathcal{B})_{=}} 
: \mathsf{Adj}(\mathcal{B})_{=}^{\mathsf{2}} \rightarrow
\mathsf{Ref}(\mathcal{B}) \odot \mathsf{Ref}(\mathcal{B})^\propto$
maps an adjunction $(\mathbf{A}, \mathbf{g}, \mathbf{B})$ 
to its polar factorization 
$\div_{\mathsf{Adj}(\mathcal{B})_{=}}(\mathbf{A}, \mathbf{g}, \mathbf{B})
= (\mathbf{A}, \mathsf{ref}_{\mathcal{B}}(\mathbf{g}), \diamondsuit(\mathbf{g}), \mathsf{ref}_{\mathcal{B}}^\propto(\mathbf{g}), \mathbf{B})$,
and maps a morphism of adjunctions
$(\mathbf{a}, \mathbf{b}) : (\mathbf{A}_1, \mathbf{g}_1, \mathbf{B}_1) \rightarrow (\mathbf{A}_2, \mathbf{g}_2, \mathbf{B_2})$ 
to a morphism of polar factorizations
$\div_{\mathsf{Adj}(\mathcal{B})_{=}}(\mathbf{a}, \mathbf{b})
= (\mathbf{a}, \diamondsuit_{(\mathbf{a}, \mathbf{b})}, \mathbf{b}) 
: \div_{\mathsf{Adj}(\mathcal{B})_{=}}(\mathbf{A}_1, \mathbf{g}_1, \mathbf{B_1})
\rightarrow 
\div_{\mathsf{Adj}(\mathcal{B})_{=}}(\mathbf{A}_2, \mathbf{g}_2, \mathbf{B}_2)$,
where the axis adjunction
$\diamondsuit_{(\mathbf{a}, \mathbf{b})}
: \diamondsuit(\mathbf{g}_1) \rightleftharpoons \diamondsuit(\mathbf{g}_2)$
is given by diagonalization of the commutative square
$\mathsf{ref}_{\mathcal{B}}(\mathbf{g}_1) 
\circ
\left( \mathsf{ref}_{\mathcal{B}}^\propto(\mathbf{g}_1) \circ \mathbf{b} \right)
= 
\left( \mathbf{a} \circ \mathsf{ref}_{\mathcal{B}}(\mathbf{g}_2) \right)
\circ
\mathsf{ref}_{\mathcal{B}}^\propto(\mathbf{g}_2)$.
The axis 
$\diamondsuit_{(\mathbf{a}, \mathbf{b})}
= \langle \check{\diamondsuit}_{(\mathbf{a}, \mathbf{b})}, \hat{\diamondsuit}_{(\mathbf{a}, \mathbf{b})} \rangle$ 
is defined as follows.
\begin{center}
$\begin{array}{r@{\hspace{5pt}\doteq\hspace{5pt}}c@{\hspace{5pt}=\hspace{5pt}}c}
\check{\diamondsuit}_{(\mathbf{a}, \mathbf{b})}
& \pi_1^{\mathbf{g}_1} \cdot \check{\mathbf{b}} \cdot \xi_1^{\mathbf{g}_2}
& \pi_0^{\mathbf{g}_1} \cdot \check{\mathbf{a}} \cdot \xi_0^{\mathbf{g}_2}
: \diamondsuit(\mathbf{g}_1) \rightarrow \diamondsuit(\mathbf{g}_2)
\\
\hat{\diamondsuit}_{(\mathbf{a}, \mathbf{b})} 
& \pi_0^{\mathbf{g}_2} \cdot \hat{\mathbf{a}} \cdot \xi_0^{\mathbf{g_1}} 
& \pi_1^{\mathbf{g}_2} \cdot \hat{\mathbf{b}} \cdot \xi_1^{\mathbf{g}_1}
: \diamondsuit(\mathbf{g}_2) \rightarrow \diamondsuit(\mathbf{g}_1)
\end{array}$
\end{center}
Hence,
to compute either adjoint, 
first project to either source or target order,
next use the corresponding component adjoint,
and finally embed from the corresponding order.

\begin{theorem} [Special Equivalence] \label{special-equivalence}
The $\mathsf{Adj}(\mathcal{B})_{=}$-arrow category
is equivalent (Fig.~\ref{conceptual-structure}) to
the $\langle \mathsf{Ref}(\mathcal{B}), \mathsf{Ref}(\mathcal{B})^\propto \rangle$-factorization category
\[\mathsf{Adj}(\mathcal{B})_{=}^{\mathsf{2}} \equiv \mathsf{Ref}(\mathcal{B}) \odot \mathsf{Ref}(\mathcal{B})^\propto.\]
\end{theorem}
This equivalence,
mediated by polar factorization and composition,
is a special case for adjunctions of the general equivalence
(Thm.~\ref{general-equivalence}).

%\include{ord}

%%%%%%%%%%%%%%%%%%%%%%%%%%%%%%%%%%%%%%%%%%%%%%%%%%%%%%%%%%%%%%%%%%%%%%%%%%%%%%%%
%%%%%%%%%%%%%%%%%%%%%%%%%%%%%%%%%%%%%%%%%%%%%%%%%%%%%%%%%%%%%%%%%%%%%%%%%%%%%%%%

%%%%%%%%%%%%%%%%%%%%%%%%%%%%%%%%%%%%%%%%%%%%%%%%%%%%%%%%%%%%%%%%%%%%%%%%%%%%%%%%
%%%%%%%%%%%%%%%%%%%%%%%%%%%%%%%%%%%%%%%%%%%%%%%%%%%%%%%%%%%%%%%%%%%%%%%%%%%%%%%%

%%%%%%%%%%%%%%%%%%%%%%%%%%%%%%%%%%%%%%%%%%%%%%%%%%%%%%%%%%%%%%%%%%%%%%
%%%%%%%%%%%%%%%%%%%%%%%%%%%%%%%%%%%%%%%%%%%%%%%%%%%%%%%%%%%%%%%%%%%%%%
\section{Classification Structures}\label{sec:classification:structures}
%%%%%%%%%%%%%%%%%%%%%%%%%%%%%%%%%%%%%%%%%%%%%%%%%%%%%%%%%%%%%%%%%%%%%%
%%%%%%%%%%%%%%%%%%%%%%%%%%%%%%%%%%%%%%%%%%%%%%%%%%%%%%%%%%%%%%%%%%%%%%

%%%%%%%%%%%%%%%%%%%%%%%%%%%%%%%%%%%%%%%%%%%%%%%%%%%%%%%%%%%%%%%%%%%%%%
\subsection{Classifications}\label{subsec:classifications}
%%%%%%%%%%%%%%%%%%%%%%%%%%%%%%%%%%%%%%%%%%%%%%%%%%%%%%%%%%%%%%%%%%%%%%

A \emph{classification structure} $\mathbf{A}$ in (internal to) a topos $\mathcal{B}$ 
has two components,
a $\mathcal{B}$-object of \emph{instances} $\mathsf{inst}(\mathbf{A})$ 
and a $\mathcal{B}$-object of \emph{types} $\mathsf{typ}(\mathbf{A})$.
Classification structures can alternately be defined in three equivalent versions:
a relation version, a morphism version or an adjunction version.
The relation version of classification structure
is constrained by a binary classification relation 
$\models_{\mathbf{A}} : \mathsf{inst}(\mathbf{A}) \rightharpoondown \mathsf{typ}(\mathbf{A})$.
We can use the abbreviation $\mathbf{A}$ for the relation $\models_{\mathbf{A}}$.
This version is known as a formal context in FCA \cite{ganter:wille:99},
where instances are called formal objects,
types are called formal attributes,
and the classification relation is called an incidence relation.
The morphism version of classification structure
is constrained by a pair of dual $\mathcal{B}$-morphisms:
the \emph{intent} morphism
$\mathsf{int}_{\mathbf{A}} = {\mathbf{A}}^{\!01} 
: \mathsf{inst}(\mathbf{A}) \rightarrow {\wp}\,\mathsf{typ}(\mathbf{A})$
and the \emph{extent} morphism
$\mathsf{ext}_{\mathbf{A}} = {\mathbf{A}}^{\!10} 
: \mathsf{typ}(\mathbf{A}) \rightarrow {\wp}\,\mathsf{inst}(\mathbf{A})$.
The adjunction version of classification structure
is constrained by a pair of dual \emph{derivation} monotonic morphisms:
forward derivation
${\mathbf{A}}^{\Rightarrow}
= (\exists\mathsf{int}_{\mathbf{A}})^{\propto} \cdot {\cap}_{\mathsf{typ}(\mathbf{A})}
: {\wp}\,\mathsf{inst}(\mathbf{A})^{\propto} \rightarrow {{\wp}\,\mathsf{typ}(\mathbf{B})}$
mapping instance subobjects (extents) to type subobjects (intents), 
and reverse derivation
${\mathbf{A}}^{\Leftarrow}
= (\exists\mathsf{ext}_{\mathbf{A}})^{\propto} \cdot {\cap}_{\mathsf{inst}(\mathbf{A})}
: {\wp}\,\mathsf{inst}(\mathbf{B})^{\propto} \rightarrow {{\wp}\,\mathsf{typ}(\mathbf{A})}$
mapping intents to extents.
Derivation forms a order adjunction
$\mathsf{deriv}_{\mathbf{A}} 
= \langle {{\mathbf{A}}^{\!\Rightarrow}}^{\propto}, {\mathbf{A}}^{\!\Leftarrow} \rangle
: {\wp}\,\mathsf{inst}(\mathbf{A}) \rightleftharpoons {{\wp}\,\mathsf{typ}(\mathbf{B})}^{\propto}$
from the complete lattice of extents to the complete (opposite) lattice of intents.
%$X \cdot {\models}^{\mathbf{A}}_{\Rightarrow} \supseteq Y$ 
%\underline{iff}
%$X \subseteq Y \cdot {\models}^{\mathbf{A}}_{\Leftarrow}$
%for all extents $X : 1 \rightarrow \mathsf{inst}(\mathbf{A})$
%and intents $Y : 1 \rightarrow \mathsf{typ}(\mathbf{A})$
Thus,
a classification $\mathbf{A}$ defines an object
$\mathsf{incl}(\mathbf{A})
=
\left( {\wp}\mathsf{inst}(\mathbf{A}), 
\mathsf{deriv}_{\mathbf{A}}, 
{{\wp}\mathsf{typ}(\mathbf{A})}^{\propto} \right)$
in the arrow subcategory $\mathsf{Adj}(\mathcal{B})_{=}^{\mathsf{2}}$.

\begin{sloppypar}
Application of polar factorization to derivation,
results in the conceptual structure 
$\mathsf{clg}(\mathbf{A}) = \div_{\mathsf{Adj}_{=}}(\mathsf{incl}(\mathbf{A}))$,
which is visualized as
\[
{\wp}\,\mathsf{inst}(\mathbf{A}) 
\stackrel{\mathsf{extent}_{\mathsf{clg}(\mathbf{A})}}{\rightleftharpoons} 
\mathsf{axis}(\mathbf{A})
\stackrel{\mathsf{intent}_{\mathsf{clg}(\mathbf{A})}}{\rightleftharpoons}
{\wp}\,{\mathsf{typ}(\mathbf{A})}^{\propto}
.\]
The axis of derivation $\mathsf{axis}(\mathbf{A}) = \diamondsuit(\mathsf{deriv}(\mathbf{A}))$ 
is called the \emph{concept lattice} of $\mathbf{A}$.
A bipole of derivation, called a \emph{formal concept},
is a pair $(X, Y)$ consisting of
a closed extent $X \in^1 \mathsf{clo}(\mathsf{deriv}(\mathbf{A}))$ and 
an open intent $Y \in^1 \mathsf{open}(\mathsf{deriv}(\mathbf{A}))$,
where $X = Y \cdot {\models}_{\mathbf{A}}^{\Leftarrow}$ 
(equivalently, $Y = X \cdot {\models}_{\mathbf{A}}^{\Rightarrow}$).
The \emph{extent} reflection of derivation
$\mathsf{extent}_{\mathsf{clg}(\mathbf{A})} 
= \langle \xi_0, \pi_0 \rangle
= \mathsf{ref}(\mathsf{deriv}(\mathbf{A}))
: {\wp}\,\mathsf{inst}(\mathbf{A}) \rightleftharpoons \mathsf{axis}(\mathbf{A})$
consists of 
the source embedding
$\xi_0 : {\wp}\,\mathsf{inst}(\mathbf{A}) \rightarrow \mathsf{axis}(\mathbf{A})$
and the projection 
$\pi_0 : \mathsf{axis}(\mathbf{A}) \rightarrow {\wp}\,\mathsf{inst}(\mathbf{A})$.
The \emph{intent} coreflection of derivation
$\mathsf{intent}_{\mathsf{clg}(\mathbf{A})} 
= \langle \pi_1^{\propto}, \xi_1 \rangle
= \mathsf{ref}^{\propto}(\mathsf{deriv}(\mathbf{A}))
: \mathsf{axis}(\mathbf{A}) \rightleftharpoons {\wp}\,\mathsf{typ}(\mathbf{A})^{\propto}$
consists of the projection 
$\pi_1^{\propto} : \mathsf{axis}(\mathbf{A}) \rightarrow {\wp}\,\mathsf{typ}(\mathbf{A})^{\propto}$,
and the target embedding
$\xi_1 : {\wp}\,\mathsf{typ}(\mathbf{A})^{\propto} \rightarrow \mathsf{axis}(\mathbf{A})$.
The preorder $\mathsf{axis}(\mathbf{A})$ is a complete lattice,
with join and meet defined by
\begin{center}
$\begin{array}{r@{\hspace{5pt}=\hspace{5pt}}ll}
\vee_{\mathbf{A}} 
& \exists {\pi}_0 \cdot \cup_{\mathsf{inst}(\mathbf{A})} \cdot\, {\xi}_0
& : {\wp}\,\mathsf{axis}(\mathbf{A}) \rightarrow {\wp}\,{\wp}\,\mathsf{inst}(\mathbf{A})  
\rightarrow {\wp}\,\mathsf{inst}(\mathbf{A}) \rightarrow \mathsf{axis}(\mathbf{A})
\\
& \exists{\pi}_1 \cdot \cap_{\mathsf{typ}(\mathbf{A})}^{\propto} \cdot\, {\xi}_1
& : {\wp}\,\mathsf{axis}(\mathbf{A}) \rightarrow {\wp}\,{\wp}\,\mathsf{typ}(\mathbf{A})  
\rightarrow {\wp}\,\mathsf{typ}(\mathbf{A})^{\propto} \rightarrow \mathsf{axis}(\mathbf{A})
\\
\wedge_{\mathbf{A}} 
& {(\exists{\pi}_0)}^{\propto} \cdot \cap_{\mathsf{inst}(\mathbf{A})} \cdot\, {\xi}_0
& : {\wp}\,\mathsf{axis}(\mathbf{A})
\rightarrow {\left({\wp}\,{\wp}\,\mathsf{inst}(\mathbf{A})\right)}^{\propto}  
\rightarrow {\wp}\,\mathsf{inst}(\mathbf{A}) \rightarrow \mathsf{axis}(\mathbf{A})
\\
& {(\exists{\pi}_1)}^{\propto} \cdot \cup_{\mathsf{typ}(\mathbf{A})}^{\propto} \cdot\, {\xi}_1
& : {({\wp}\,\mathsf{axis}(\mathbf{A}))}^{\propto}
\rightarrow {\left({\wp}\,{\wp}\,\mathsf{typ}(\mathbf{A})\right)}^{\propto}  
\rightarrow {\wp}\,\mathsf{typ}(\mathbf{A})^{\propto} \rightarrow \mathsf{axis}(\mathbf{A})
\end{array}$ 
\end{center}
In summary,
this conceptual structure is the polar factorization of the classification structure $\mathbf{A}$ in its adjunction version.
We can recover the original classification structure by composition:
$\mathsf{deriv}_{\mathbf{A}} 
= \mathsf{extent}_{\mathsf{clg}(\mathbf{A})} 
\circ \mathsf{intent}_{\mathsf{clg}(\mathbf{A})}$. 
\end{sloppypar}

%%%%%%%%%%%%%%%%%%%%%%%%%%%%%%%%%%%%%%%%%%%%%%%%%%%%%%%%%%%%%%%%%%%%%%
\subsection{Infomorphisms}\label{subsec:infomorphisms}
%%%%%%%%%%%%%%%%%%%%%%%%%%%%%%%%%%%%%%%%%%%%%%%%%%%%%%%%%%%%%%%%%%%%%%

A morphism of classification structures
$\mathbf{f} 
= \langle \mathsf{inst}(\mathbf{f}), \mathsf{typ}(\mathbf{f}) \rangle 
= \langle \check{\mathbf{f}}, \hat{\mathbf{f}} \rangle 
: \mathbf{A} \rightleftharpoons \mathbf{B}$
called an \emph{infomorphism}, 
consists of 
an \emph{instance} $\mathcal{E}$-morphism 
$\mathsf{inst}(\mathbf{f}) = \check{\mathbf{f}} 
: \mathsf{inst}(\mathbf{A}) \leftarrow \mathsf{inst}(\mathbf{B})$ 
and 
a \emph{type} $\mathcal{B}$-morphism 
$\mathsf{typ}(\mathbf{f}) = \hat{\mathbf{f}} 
: \mathsf{typ}(\mathbf{A}) \rightarrow \mathsf{typ}(\mathbf{B})$.
Infomorphisms can alternately be defined in three isomorphic versions:
a relation version, a morphism version or an adjunction version.
Each version expresses the invariance of classification under change of notation. 
The relation version of infomorphism \cite{barwise:seligman:97}
satisfies the fundamental condition
\begin{description}
\item [fundamental:]
$\mathbf{A}(\mathsf{inst}(\mathbf{f}),1_{\mathsf{typ}(\mathbf{A})})
= \mathbf{B}(1_{\mathsf{inst}(\mathbf{B})},\mathsf{typ}(\mathbf{f}))$; 
or equivalently,
\item [external:]
$\mathbf{A}(x{\cdot}\mathsf{inst}(\mathbf{f}), y)
= \mathbf{B}(x, y{\cdot}\mathsf{typ}(\mathbf{f}))$
for every instance element
$x \in^C \mathsf{inst}(\mathbf{B})$ 
and type element $y \in^C \mathsf{typ}(\mathbf{A})$.
\end{description}
The morphism version of infomorphism
satisfies the two naturality conditions
$\mathsf{ext}_{\mathbf{A}} \cdot {\mathsf{inst}(\mathbf{f})}^{{-}1} 
= \mathsf{typ}(\mathbf{f}) \cdot \mathsf{ext}_{\mathbf{B}}$
and
$\mathsf{int}_{\mathbf{B}} \cdot {\mathsf{typ}(\mathbf{f})}^{{-}1} 
= \mathsf{inst}(\mathbf{f}) \cdot \mathsf{int}_{\mathbf{A}}$.
The fundamental condition for infomorphisms
can be extended (existentionally) in two ways to extents and intents.
\begin{itemize}
\item First,
fix source type $y \in \mathsf{typ}(\mathbf{A})$ and let 
instance $x$ universally vary over some target extent $X \subseteq \mathsf{inst}(\mathbf{B})$.
Then,
the fundamental condition translates to
$y \in \mathbf{A}^{\Rightarrow}(\exists\mathsf{inst}(\mathbf{f})(X))
\;\;\mbox{iff}\;\; y \in {\mathsf{typ}(\mathbf{f})}^{-1}(\mathbf{B}^{\Rightarrow}(X))$
for any source type $y \in \mathsf{typ}(\mathbf{A})$
and target extent $X \subseteq \mathsf{inst}(\mathbf{B})$.
Pointlessly,
since
$\exists\mathsf{inst}(\mathbf{f}) \cdot 
(\mathbf{A}^{\Rightarrow})^{\propto}
= \exists\mathsf{inst}(\mathbf{f}) \cdot \exists\mathsf{int}_{\mathbf{A}} 
\cdot {\cap}_{\mathsf{typ}(\mathbf{A})}^{\propto}
= \exists\mathsf{int}_{\mathbf{B}} \cdot \exists{\mathsf{typ}(\mathbf{f})}^{{-}1} 
\cdot {\cap}_{\mathsf{typ}(\mathbf{A})}^{\propto}
= \exists\mathsf{int}_{\mathbf{B}} 
\cdot {\cap}_{\mathsf{typ}(\mathbf{B})}^{\propto}
\cdot {\mathsf{typ}(\mathbf{f})}^{-1}
= (\mathbf{B}^{\Rightarrow})^{\propto} \cdot {\mathsf{typ}(\mathbf{f})}^{-1}$,
\begin{center}
$\begin{array}{l@{\hspace{12pt}}r@{\hspace{5pt}=\hspace{5pt}}l@{\hspace{5pt}:\hspace{5pt}}l}
\mathrm{\underline{morphism}}
& \exists\mathsf{inst}(\mathbf{f}) \cdot (\mathbf{A}^{\Rightarrow})^{\propto}
& \mathbf{B}^{\Rightarrow} \cdot {\mathsf{typ}(\mathbf{f})}^{-1}
& {\wp}\mathsf{inst}(\mathbf{B}) \rightarrow {\wp}\mathsf{typ}(\mathbf{A})^{\propto} \\
\mathrm{\underline{relation}}
& \mathsf{inst}(\mathbf{f})^{\triangleright} \circ {\models}_{\mathbf{A}}
& {\models}_{\mathbf{B}} \circ \mathsf{typ}(\mathbf{f})^{\triangleleft}
& \mathsf{inst}(\mathbf{B}) \rightarrow \mathsf{typ}(\mathbf{A}).
\end{array}$
\end{center}
\item Second,
fix target instance $x \in \mathsf{inst}(\mathbf{B})$ and let 
type $y$ universally vary over some source intent $Y \subseteq \mathsf{typ}(\mathbf{A})$.
Then,
the fundamental condition translates to
$x \in {\mathsf{inst}(\mathbf{f})}^{-1}(\mathbf{A}^{\Leftarrow}(Y))
\;\;\mbox{iff}\;\; x \in \mathbf{B}^{\Leftarrow}(\exists\mathsf{typ}(\mathbf{f})(Y))$
for any target instance $x \in \mathsf{inst}(\mathbf{B})$
and source intent $Y \subseteq \mathsf{typ}(\mathbf{A})$.
Pointlessly,
since
$\mathbf{A}^{\Leftarrow} \cdot {\mathsf{inst}(\mathbf{f})}^{-1}
= (\exists\mathsf{ext}_{\mathbf{A}})^{\propto} \cdot 
{\cap}_{\mathsf{inst}(\mathbf{A})} 
\cdot {\mathsf{inst}(\mathbf{f})}^{-1}
= (\exists\mathsf{ext}_{\mathbf{A}})^{\propto} \cdot 
(\exists{\mathsf{inst}(\mathbf{f})}^{{-}1})^{\propto} 
\cdot {\cap}_{\mathsf{inst}(\mathbf{B})}
= (\exists\mathsf{typ}(\mathbf{f}))^{\propto} 
\cdot (\exists\mathsf{ext}_{\mathbf{B}})^{\propto} 
\cdot {\cap}_{\mathsf{inst}(\mathbf{B})}
= (\exists\mathsf{typ}(\mathbf{f}))^{\propto} \cdot \mathbf{B}^{\Leftarrow}$,
\begin{center}
$\begin{array}{l@{\hspace{12pt}}r@{\hspace{5pt}=\hspace{5pt}}l@{\hspace{5pt}:\hspace{5pt}}l}
\mathrm{\underline{morphism}}
& \mathbf{A}^{\Leftarrow} \cdot {\mathsf{inst}(\mathbf{f})}^{-1}
& (\exists\mathsf{typ}(\mathbf{f}))^{\propto} \cdot \mathbf{B}^{\Leftarrow}
& {\wp}\mathsf{typ}(\mathbf{A})^{\propto} \rightarrow {\wp}\mathsf{inst}(\mathbf{B}) \\
\mathrm{\underline{relation}}
& {\models}_{\mathbf{A}}^{\propto} \circ \mathsf{inst}(\mathbf{f})^{\triangleleft}
& \mathsf{typ}(\mathbf{f})^{\triangleright} \circ {\models}_{\mathbf{B}}^{\propto}
& \mathsf{typ}(\mathbf{A}) \rightarrow \mathsf{inst}(\mathbf{B}).
\end{array}$
\end{center}
\end{itemize}
Hence,
the adjunction version of infomorphism (Figure~\ref{extent-intent})
satisfies the naturality condition
$\mathsf{dir}(\mathsf{inst}(\mathbf{f})) \circ \mathsf{deriv}_{\mathbf{A}} 
= \mathsf{deriv}_{\mathbf{B}} \circ \mathsf{inv}(\mathsf{typ}(\mathbf{f}))$.
Thus,
an infomorphism $\mathbf{f} : \mathbf{A} \rightleftharpoons \mathbf{B}$
defines a morphism
$\mathsf{incl}(\mathbf{f})
=
\left( \mathsf{dir}(\mathsf{inst}(\mathbf{f})), 
\mathsf{inv}(\mathsf{typ}(\mathbf{f})) \right)
: \mathsf{incl}(\mathbf{B}) \rightarrow \mathsf{incl}(\mathbf{A})$
in the category $\mathsf{Adj}(\mathcal{B})_{=}^{\mathsf{2}}$.

\begin{figure}
\begin{center}
\begin{tabular}{c}
\begin{tabular}{c@{\hspace{100pt}}c}
\setlength{\unitlength}{0.6pt}
\begin{picture}(120,120)(0,-20)
%\put(0,0){\framebox(120,60){}}
%\put(0,60){\framebox(120,60){}}
\put(-60,90){\makebox(120,60) {${\wp}\mathsf{typ}(\mathbf{A})^{\mathrm{op}}$}}
\put(60,90){\makebox(120,60)  {${\wp}\mathsf{typ}(\mathbf{B})^{\mathrm{op}}$}}
\put(-60,30){\makebox(120,60) {$\mathsf{conc}(\mathbf{A})$}}
\put(60,30){\makebox(120,60)  {$\mathsf{conc}(\mathbf{B})$}}
\put(-60,-30){\makebox(120,60){${\wp}\mathsf{inst}(\mathbf{A})$}}
\put(60,-30){\makebox(120,60) {${\wp}\mathsf{inst}(\mathbf{B})$}}
\put(-60,60){\makebox(60,60){\footnotesize{$\mathsf{intent}_{\mathbf{A}}$}}}
\put(125,60){\makebox(60,60){\footnotesize{$\mathsf{intent}_{\mathbf{B}}$}}}
\put(-60,0){\makebox(60,60){\footnotesize{$\mathsf{extent}_{\mathbf{A}}$}}}
\put(125,0){\makebox(60,60){\footnotesize{$\mathsf{extent}_{\mathbf{B}}$}}}
\put(10,105){\makebox(120,60){\footnotesize{$\mathsf{inv}(\mathsf{typ}(\mathbf{f}))$}}}
\put(0,45){\makebox(120,60){\footnotesize{$\mathsf{adj}(\mathbf{f})$}}}
\put(4,-15){\makebox(120,60){\footnotesize{$\mathsf{dir}(\mathsf{inst}(\mathbf{f}))$}}}
\thicklines
\put(0,75){\vector(0,1){30}}
\put(0,15){\vector(0,1){30}}
\put(120,75){\vector(0,1){30}}
\put(120,15){\vector(0,1){30}}
\put(80,120){\vector(-1,0){40}}
\put(80,60){\vector(-1,0){40}}
\put(80,0){\vector(-1,0){40}}
\end{picture}
&
\setlength{\unitlength}{0.6pt}
\begin{picture}(160,160)(0,0)
%\put(0,0){\framebox(160,80){}}
%\put(0,80){\framebox(160,80){}}
\put(-10,100){\vector(0,1){40}}
\put(10,140){\vector(0,-1){40}}
\put(150,100){\vector(0,1){40}}
\put(170,140){\vector(0,-1){40}}
\put(-10,20){\vector(0,1){40}}
\put(10,60){\vector(0,-1){40}}
\put(150,20){\vector(0,1){40}}
\put(170,60){\vector(0,-1){40}}
\put(50,170){\vector(1,0){60}}
\put(110,150){\vector(-1,0){60}}
\put(50,90){\vector(1,0){60}}
\put(110,70){\vector(-1,0){60}}
\put(50,10){\vector(1,0){60}}
\put(110,-10){\vector(-1,0){60}}
\put(-50,100){\makebox(40,40){\footnotesize{$\mathsf{int}_{\mathbf{A}}$}}}
\put(10,100){\makebox(40,40){\footnotesize{$\mathsf{tau}_{\mathbf{A}}$}}}
\put(110,100){\makebox(40,40){\footnotesize{$\mathsf{int}_{\mathbf{B}}$}}}
\put(170,100){\makebox(40,40){\footnotesize{$\mathsf{tau}_{\mathbf{B}}$}}}
\put(-50,20){\makebox(40,40){\footnotesize{$\mathsf{iota}_{\mathbf{A}}$}}}
\put(10,20){\makebox(40,40){\footnotesize{$\mathsf{ext}_{\mathbf{A}}$}}}
\put(110,20){\makebox(40,40){\footnotesize{$\mathsf{iota}_{\mathbf{B}}$}}}
\put(170,20){\makebox(40,40){\footnotesize{$\mathsf{ext}_{\mathbf{B}}$}}}
\put(48,160){\makebox(60,40){\footnotesize{${\exists}\mathsf{typ}(\mathbf{f})$}}}
\put(58,120){\makebox(60,40){\footnotesize{${\mathsf{typ}(\mathbf{f})}^{-1}$}}}
\put(50,80){\makebox(60,40){\footnotesize{$\mathsf{right}(\mathbf{f})$}}}
\put(50,40){\makebox(60,40){\footnotesize{$\mathsf{left}(\mathbf{f})$}}}
\put(58,0){\makebox(60,40){\footnotesize{${\mathsf{inst}(\mathbf{f})}^{-1}$}}}
\put(50,-41){\makebox(60,40){\footnotesize{${\exists}\mathsf{inst}(\mathbf{f})$}}}
\put(-80,120){\makebox(160,80) {$\langle {\wp}\mathsf{typ}(\mathbf{A}), \supseteq \rangle$}}
\put(80,120){\makebox(160,80)  {$\langle {\wp}\mathsf{typ}(\mathbf{B}), \supseteq \rangle$}}
\put(-80,40){\makebox(160,80) {$\langle \mathsf{conc}(\mathbf{A}), \leq_{\mathbf{A}} \rangle$}}
\put(80,40){\makebox(160,80)  {$\langle \mathsf{conc}(\mathbf{B}), \leq_{\mathbf{B}} \rangle$}}
\put(-80,-40){\makebox(160,80){$\langle {\wp}\mathsf{inst}(\mathbf{A}), \subseteq \rangle$}}
\put(80,-40){\makebox(160,80) {$\langle {\wp}\mathsf{inst}(\mathbf{B}), \subseteq \rangle$}}
\end{picture}
\\ \\ \\
(iconic) & (detailed)
\end{tabular}
\\ \\
$\begin{array}{r@{\hspace{5pt}=\hspace{5pt}}c@{\hspace{5pt}:\hspace{5pt}}l}
\mathsf{dir}(\mathsf{inst}(\mathbf{f})) 
& \langle \exists\mathsf{inst}(\mathbf{f}), {\mathsf{inst}(\mathbf{f})}^{-1} \rangle
& {\wp}\mathsf{inst}(\mathbf{B}) \rightleftharpoons {\wp}\mathsf{inst}(\mathbf{A})
\\
\mathsf{deriv}(\mathbf{A}) 
& \langle \mathbf{A}^{\Rightarrow}, \mathbf{A}^{\Leftarrow} \rangle
& {\wp}\mathsf{inst}(\mathbf{A}) \rightleftharpoons {\wp}\mathsf{typ}(\mathbf{A})^{\propto}
\\
\mathsf{deriv}(\mathbf{B}) 
& \langle \mathbf{B}^{\Rightarrow}, \mathbf{B}^{\Leftarrow} \rangle
& {\wp}\mathsf{inst}(\mathbf{B}) \rightleftharpoons {\wp}\mathsf{typ}(\mathbf{B})^{\propto}
\\
\mathsf{inv}(\mathsf{typ}(\mathbf{f})) 
& \langle {\mathsf{typ}(\mathbf{f})}^{-1}, \exists\mathsf{typ}(\mathbf{f}) \rangle
& {\wp}\mathsf{typ}(\mathbf{B})^{\propto} \rightleftharpoons {\wp}\mathsf{typ}(\mathbf{A})^{\propto}
\end{array}$
\end{tabular}
\end{center}
\caption{The extent and intent natural transformations}
\label{extent-intent}
\end{figure}

Defining composition and identity coordinatewise,
classifications and infomorphisms form the category $\mathsf{Clsn}(\mathcal{B})$.
There is an inclusion functor
$\mathsf{incl}_{\mathcal{B}} 
: \mathsf{Clsn}(\mathcal{B})^{\mathrm{op}} \rightarrow \mathsf{Adj}(\mathcal{B})_{=}^{\mathsf{2}}$.
Classifications and infomorphisms resolve into components:
there is an instance functor
$\mathsf{inst}_{\mathcal{B}} 
: \mathsf{Clsn}(\mathcal{B})^{\mathrm{op}} \rightarrow \mathcal{B}$
with 
$\mathsf{inst}_{\mathcal{B}} \circ \mathsf{dir}_{\mathcal{B}} 
= \mathsf{incl}_{\mathcal{B}} \circ \partial_0 
: \mathsf{Clsn}(\mathcal{B})^{\mathrm{op}} \rightarrow \mathsf{Adj}(\mathcal{B})_{=}$,
and there is a type functor
$\mathsf{typ}_{\mathcal{B}} 
: \mathsf{Clsn}(\mathcal{B}) \rightarrow \mathcal{B}$
with 
$\mathsf{typ}_{\mathcal{B}}^{\mathrm{op}} \circ \mathsf{inv}_{\mathcal{B}} 
= \mathsf{incl}_{\mathcal{B}} \circ \partial_1 
: \mathsf{Clsn}(\mathcal{B})^{\mathrm{op}} \rightarrow \mathsf{Adj}(\mathcal{B})_{=}$.
The morphism version of infomorphism means
there is an extent natural transformation
$\mathsf{ext} : \mathsf{typ}_{\mathcal{B}} \Rightarrow \mathsf{inst}_{\mathcal{B}}^{\mathrm{op}} \circ {(-)}_{\mathcal{B}}^{-1} 
: \mathsf{Clsn}(\mathcal{B}) \rightarrow \mathcal{B}$
and an intent natural transformation
$\mathsf{int} : \mathsf{inst}_{\mathcal{B}} \Rightarrow \mathsf{typ}_{\mathcal{B}}^{\mathrm{op}} \circ {(-)}_{\mathcal{B}}^{-1} : \mathsf{Clsn}(\mathcal{B})^{\mathrm{op}} \rightarrow \mathcal{B}$.
The adjunction version of infomorphism means
there is a derivation natural transformation
$\mathsf{deriv}_{\mathcal{B}}
: \mathsf{inst}_{\mathcal{B}} \circ \mathsf{dir}_{\mathcal{B}} \Rightarrow \mathsf{typ}_{\mathcal{B}}^{\mathrm{op}} \circ \mathsf{inv}_{\mathcal{B}} 
: \mathsf{Clsn}(\mathcal{B})^{\mathrm{op}} \!\rightarrow \mathsf{Adj}(\mathcal{B})$
with 
$\mathsf{deriv}_{\mathcal{B}} 
= \mathsf{incl}_{\mathcal{B}} \,\alpha_{\mathsf{Adj}(\mathcal{B})_{=}}$.

\begin{sloppypar}
Application of polar factorization to derivation morphism,
results in the morphism of conceptual structures 
$\mathsf{clg}(\mathbf{f}) = \div_{\mathsf{Adj}_{=}}(\mathsf{incl}(\mathbf{f}))$.
The axis of derivation morphism
$\mathsf{axis}_{\mathbf{f}}
= \diamondsuit_{\mathsf{deriv}(\mathbf{f})}
: \diamondsuit(\mathsf{clg}(\mathbf{B}))
\rightleftharpoons \diamondsuit(\mathsf{clg}(\mathbf{A}))$,
which is called the \emph{concept adjunction} of $\mathbf{f}$,
is defined by polar diagonalization of the commutative square
$\mathsf{ref}_{\mathbf{B}} 
\circ
\left( \mathsf{ref}_{\mathbf{B}}^\propto \circ \mathsf{inv}(\mathsf{typ}(\mathbf{f})) \right)
= 
\left( \mathsf{dir}(\mathsf{inst}(\mathbf{f})) \circ \mathsf{ref}_{\mathbf{A}} \right)
\circ 
\mathsf{ref}_{\mathbf{A}}^\propto$.
The concept adjunction is defined as follows.
\end{sloppypar}
\begin{center}
{\footnotesize 
$\begin{array}{r@{\hspace{5pt}}c@{\hspace{5pt}}l@{\hspace{-5pt}}l}
\mathsf{left}(\mathsf{axis}_{\mathbf{f}}) 
& \doteq &
\pi_1^{\mathbf{B}\,\propto} \cdot \mathsf{typ}(\mathbf{f})^{-1} \cdot \xi_1^{\mathbf{A}}
&
\\
& = & 
\pi_0^{\mathbf{B}} \cdot \exists\mathsf{inst}(\mathbf{f}) \cdot \xi_0^{\mathbf{A}}
& : \mathsf{axis}(\mathbf{B}) \rightarrow \mathsf{axis}(\mathbf{A})
\\
\mathsf{right}(\mathsf{axis}_{\mathbf{f}})
& \doteq &
\pi_0^{\mathbf{A}} \cdot \mathsf{inst}(\mathbf{f})^{-1} \cdot \xi_0^{\mathbf{B}}
&
\\ 
& = &
\pi_1^{\mathbf{A}\,\propto} \cdot \exists\mathsf{typ}(\mathbf{f}) \cdot \xi_1^{\mathbf{B}}
& : \mathsf{axis}(\mathbf{A}) \rightarrow \mathsf{axis}(\mathbf{B})
\end{array}$}
\end{center}
Hence,
to compute either adjoint, 
first project to either extent or intent order,
next use the corresponding component adjoint,
and finally embed from the corresponding order.

%%%%%%%%%%%%%%%%%%%%%%%%%%%%%%%%%%%%%%%%%%%%%%%%%%%%%%%%%%%%%%%%%%%%%%
\subsubsection{Orders.}
%%%%%%%%%%%%%%%%%%%%%%%%%%%%%%%%%%%%%%%%%%%%%%%%%%%%%%%%%%%%%%%%%%%%%%

Any preorder 
$\mathbf{A} = \langle A, {\leq}_{\mathbf{A}} \rangle$ 
is a classification $\mathsf{incl}(\mathbf{A})$,
whose objects of instances and types are the underlying object $A$,
and whose classification relation is the order relation 
${\leq}_{\mathbf{A}} : A \rightharpoondown A$.
The intent morphism of $\mathsf{incl}(\mathbf{A})$ is the underlying up segment morphism
$\mathsf{int}_{\mathsf{incl}(\mathbf{A})} = {\uparrow}_{\mathbf{A}} 
: A \rightarrow {\wp}A$,
and dually
the extent morphism of $\mathsf{incl}(\mathbf{A})$ is the underlying down segment morphism
$\mathsf{ext}_{\mathsf{incl}(\mathbf{A})} = {\downarrow}_{\mathbf{A}} 
: A \rightarrow {\wp}A$.
The derivation adjunction of $\mathsf{incl}(\mathbf{A})$ is the bound adjunction
$\mathsf{dir}_{\mathsf{incl}(\mathbf{A})} 
= \mathsf{bnd}_{\mathbf{A}} 
: {\wp}\mathbf{A} \rightarrow {\wp}\mathbf{A}^{\propto}$.
Any order adjunction 
$\mathbf{f} = \langle \check{\mathbf{f}}, \hat{\mathbf{f}} \rangle 
: \mathbf{A} \rightleftharpoons \mathbf{B}$ 
is an infomorphism (in reverse direction)
$\mathsf{incl}(\mathbf{f}) 
: \mathsf{incl}(\mathbf{B}) \rightleftharpoons \mathsf{incl}(\mathbf{A})$,
whose instance morphism is the left adjoint  
$\check{\mathbf{f}} : \mathbf{A} \rightarrow \mathbf{B}$,
whose type morphism is the right adjoint  
$\hat{\mathbf{f}} : \mathbf{B} \rightarrow \mathbf{A}$,
and whose fundamental condition is that of order adjunctions.
There is an inclusion functor
$\mathsf{incl}_{\mathcal{B}} 
: \mathsf{Adj}(\mathcal{B})^{\mathrm{op}} \rightarrow \mathsf{Clsn}(\mathcal{B})$.

%%%%%%%%%%%%%%%%%%%%%%%%%%%%%%%%%%%%%%%%%%%%%%%%%%%%%%%%%%%%%%%%%%%%%%
\subsubsection{Complete Lattices.}
%%%%%%%%%%%%%%%%%%%%%%%%%%%%%%%%%%%%%%%%%%%%%%%%%%%%%%%%%%%%%%%%%%%%%%
A \emph{complete lattice} 
$\mathbf{L} 
= \langle L, {\leq}_{\mathbf{L}}, {\vee}_{\mathbf{L}}, {\wedge}_{\mathbf{L}} \rangle$ 
in (internal to) $\mathcal{B}$
is a partial order
$\mathsf{ord}(\mathbf{L}) 
= \langle L, {\leq}_{\mathbf{L}} \rangle$ 
that is isomorphic to the axis of derivation 
$\mathbf{L} \cong \mathsf{axis}(\mathbf{L})$ 
in the polar factorization of its bound adjunction
$\mathsf{bnd}_{\mathbf{L}}
= \langle {\Uparrow}_{\mathbf{L}}^{\propto}, {\Downarrow}_{\mathbf{L}} \rangle 
: {\wp}\mathbf{L} \rightarrow {\wp}\mathbf{L}^{\propto}$,
with ${\downarrow}_{\mathbf{L}} \cong \pi_0^{\mathbf{L}}$
and ${\uparrow}_{\mathbf{L}} \cong \pi_1^{\mathbf{L}}$,
\[
\mathsf{bnd}_{\mathbf{L}}
=
{\wp}\,\mathbf{L} 
\stackrel{\mathsf{join}_{\mathbf{L}}}{\rightleftharpoons} 
\mathbf{L}
\stackrel{\mathsf{meet}_{\mathbf{L}}}{\rightleftharpoons}
{\wp}\,\mathbf{L}^{\propto}
.\]
The \emph{join} reflection
$\mathsf{join}_{\mathbf{L}}
= \langle {\vee}_{\mathbf{L}}, {\downarrow}_{\mathbf{L}} \rangle 
: {\wp}\,\mathbf{L} \rightleftharpoons \mathbf{L}$
has the \emph{join} monotonic function 
$\vee_L : {\wp}\mathbf{L} \rightarrow \mathbf{L}$
as left adjoint and 
the down segment monotonic function
$\downarrow_L : \mathbf{L} \rightarrow {\wp}\mathbf{L}$
as right adjoint.
The \emph{meet} coreflection
$\mathsf{meet}_{\mathbf{L}}
= \langle {\uparrow}_{\mathbf{L}}^{\propto}, {\wedge}_{\mathbf{L}} \rangle 
: \mathbf{L} \rightleftharpoons {\wp}\,\mathbf{L}^{\propto}$
has the (opposite) up segment monotonic function 
$\uparrow_L^{\propto} : \mathbf{L} \rightarrow {\wp}\,\mathbf{L}^{\propto}$
as left adjoint and the \emph{meet} monotonic function 
$\wedge_L : {\wp}\,\mathbf{L}^{\propto} \rightarrow \mathbf{L}$
as right adjoint.
The fundamental condition for the join reflection
$\mathbf{L}(X{\cdot}{\vee}_{\mathbf{L}}, x)
= {\wp}\mathbf{L}(X, x{\cdot}{\downarrow}_{\mathbf{L}})$
states that ``$X{\cdot}{\vee}_{\mathbf{L}}$ is the least upper bound of $X$''
for every subobject $X \in^1 {\wp}\,L$;
also,
$x \cdot {\downarrow}_{\mathbf{L}} \cdot {\vee}_{\mathbf{L}} = x$
for every element $x \in^1 L$.
The fundamental condition for the meet coreflection
${{\wp}\mathbf{L}}^{\!\propto}(y{\cdot}{\uparrow}_{\mathbf{L}}, Y)
= {\mathbf{L}}(y, Y{\cdot}{\wedge}_{\mathbf{L}})$
states that ``$Y{\cdot}{\wedge}_{\mathbf{L}}$ is the greatest lower bound of $Y$''
for every subobject $Y \in^1 {\wp}\,L$;
also,
$y \cdot {\uparrow}_{\mathbf{L}} \cdot {\wedge}_{\mathbf{L}} = y$
for every element $y \in^1 L$.
The composition
$\mathsf{bnd}_{\mathbf{L}}
= \mathsf{join}_{\mathbf{L}} \circ \mathsf{meet}_{\mathbf{L}}$
means that
${\Uparrow}_{\mathbf{L}}^{\propto}
= {\vee}_{\mathbf{L}} \cdot {\uparrow}_{\mathbf{L}}^{\propto}$
and
${\Downarrow}_{\mathbf{L}}
= {\wedge}_{\mathbf{L}} \cdot {\downarrow}_{\mathbf{L}}$.
Hence,
${\Uparrow}_{\mathbf{L}}^{\propto} \cdot {\wedge}_{\mathbf{L}} = {\vee}_{\mathbf{L}}$
and ${\Downarrow}_{\mathbf{L}} \cdot {\vee}_{\mathbf{L}} = {\wedge}_{\mathbf{L}}$.

Complete lattices are related through order adjunctions.
A \emph{complete adjoint}
$\mathbf{g} 
= \langle \check{\mathbf{g}}, \hat{\mathbf{g}} \rangle
: \mathbf{A} \rightleftharpoons \mathbf{B}$
(internal to) a topos $\mathcal{B}$
is an order adjunction between complete lattices $\mathbf{A}$ and $\mathbf{B}$.
The left adjoint is join-preserving and the right adjoint is meet-preserving.
They determine each other.
Let $\mathsf{CAdj}_{\mathcal{B}}$
denote the full subcategory of complete lattices and complete adjoints 
with inclusion functor
$\mathsf{incl}_{\mathcal{B}} 
: \mathsf{CAdj}(\mathcal{B}) \hookrightarrow \mathsf{Adj}(\mathcal{B})$.

Any order bimodule
$\mathbf{r} : \mathbf{A} \rightharpoondown \mathbf{B}$
between complete lattices,
has 
(1) a \emph{01-embedding} monotonic function
$\mathbf{r}^{\wedge} 
= \mathbf{r}^{01} \cdot {\wedge}_{\mathbf{B}}
: \mathbf{A} \rightarrow {\wp}\mathbf{B}^{\propto} \rightarrow \mathbf{B}$
that is the composite of target meet with the 01-fiber, and
(2) a \emph{10-embedding} monotonic function
$\mathbf{r}^{\vee}
= \mathbf{r}^{10} \cdot {\vee}_{\mathbf{A}}
: \mathbf{B} \rightarrow {\wp}\mathbf{A} \rightarrow \mathbf{A}$
that is the composite of source join with the 10-fiber.
Any monotonic morphism
$\mathbf{f} : \mathbf{A} \rightarrow \mathbf{B}$
between complete lattices
is both the 01-embedding of its forward bimodule
$(\mathbf{f}^{\triangleright})^{\vee}
= (\mathbf{f}^{\triangleright})^{01} \cdot {\vee}_{\mathbf{B}}
= \mathbf{f} \cdot {\uparrow}_{\mathbf{B}} \cdot {\vee}_{\mathbf{B}}
=  \mathbf{f}$
and the 10-embedding of its reverse bimodule
$(\mathbf{f}^{\triangleleft})^{\wedge}
= (\mathbf{f}^{\triangleleft})^{10} \cdot {\wedge}_{\mathbf{B}}
= \mathbf{f} \cdot {\downarrow}_{\mathbf{B}} \cdot {\wedge}_{\mathbf{B}}
=  \mathbf{f}$.
Hence,
the forward and reverse maps are injective and the embedding maps are surjective.
However,
there may be order bimodules that are not the embedding of any monotonic morphism.
%Any order bimodule
%$\mathbf{r} : \mathbf{A} \rightharpoondown \mathbf{B}$
%between complete lattices
%is both the forward bimodule of it 01-embedding
%$\mathbf{r}^{\triangleright}_{\vee}
%= \mathbf{r}^{\triangleright}_{01} \cdot {\vee}_{\mathbf{B}}
%=  \mathbf{r}$
%and the reverse bimodule of its 10-embedding
%$\mathbf{r}^{\triangleleft}_{\wedge}
%= \mathbf{r}^{\triangleleft}_{10} \cdot {\wedge}_{\mathbf{B}}
%=  \mathbf{r}$.

%\begin{figure}
\begin{center}
\setlength{\unitlength}{0.6pt}
\begin{picture}(200,100)(0,0)
%\put(0,0){\framebox(200,100){}}
\put(-50,42){\makebox(100,50){monotonic}}
\put(-50,26){\makebox(100,50){morphism}}
\put(-50,9){\makebox(100,50){$\mathbf{f} : \mathbf{A} \rightarrow \mathbf{B}$}}
\put(150,35){\makebox(100,50){bimodule}}
\put(150,18){\makebox(100,50){$\mathbf{r} : \mathbf{A} \rightharpoondown \mathbf{B}$}}
\put(50,70){\makebox(100,50){forward}}
\put(50,50){\makebox(100,50){${(\mbox{-})}^{\triangleright}$}}
\put(50,0){\makebox(100,50){${(\mbox{-})}_{\vee}$}}
\put(50,-20){\makebox(100,50){01-embedding}}
%\thicklines
\put(50,60){\vector(1,0){100}}
\put(150,40){\vector(-1,0){100}}
\end{picture}
\end{center}
%\end{figure}

Any adjunction 
$\mathbf{g} : \mathbf{A} \rightleftharpoons \mathbf{B}$ 
between complete lattices satisfies the naturality diagrams
$\mathsf{join}_{\mathbf{A}} \circ \mathbf{g}
= \mathsf{dir}(\mathsf{left}(\mathbf{g})) \circ \mathsf{join}_{\mathbf{B}}$
and
$\mathsf{meet}_{\mathbf{A}} \circ \mathsf{inv}(\mathsf{right}(\mathbf{g}))
= \mathbf{g} \circ \mathsf{meet}_{\mathbf{B}}$. 
The first asserts join-continuity of the left adjoint
${\vee}_{\mathbf{A}} \cdot \check{\mathbf{g}} 
= \exists\check{\mathbf{g}} \cdot {\vee}_{\mathbf{B}}$
and the fundamental condition for adjoints 
$\hat{\mathbf{g}} \cdot {\downarrow}_{\mathbf{A}} 
= {\downarrow}_{\mathbf{B}} \cdot {\check{\mathbf{g}}}^{-1}$,
and the second asserts the fundamental condition for adjoints 
${\uparrow}_{\mathbf{A}} \cdot {\hat{\mathbf{g}}}^{-1} 
= \check{\mathbf{g}} \cdot {\uparrow}_{\mathbf{B}}$
and meet-continuity of the right adjoint
$\exists\hat{\mathbf{g}} \cdot {\wedge}_{\mathbf{A}} 
= {\wedge}_{\mathbf{B}} \cdot \hat{\mathbf{g}}$. 
The join reflection and meet coreflection are two special cases,
which assert the join-continuity of join
${\cup}_{\mathbf{L}} \cdot {\vee}_{\mathbf{L}} 
= \exists{\vee}_{\mathbf{L}} \cdot {\vee}_{\mathbf{L}}$
and meet-continuity of meet
${\cup}_{\mathbf{L}}^{\propto} \cdot {\wedge}_{\mathbf{L}} 
= {\exists\wedge}_{\mathbf{L}}^{\propto} \cdot {\wedge}_{\mathbf{L}}$. 
Hence,
Join and meet are natural transformations
\begin{center}
$\begin{array}{l}
\mathsf{join}
: \mathsf{left} \circ \mathsf{dir} \Rightarrow \mathsf{incl}
: \mathsf{CAdj} \rightarrow \mathsf{Set} \rightarrow \mathsf{Adj}
\\
\mathsf{meet}
: \mathsf{incl} \Rightarrow \mathsf{right}^{\mathrm{op}} \circ \mathsf{inv}
: \mathsf{CAdj} \rightarrow \mathsf{Set}^{\mathrm{op}} \rightarrow \mathsf{Adj}
\end{array}$
\end{center}

For any adjunction $\mathbf{g} : \mathbf{A} \rightleftharpoons \mathbf{B}$ between complete lattices, 
the right adjoint morphism is expressed in terms of the left adjoint morphism 
as the composition with source join
$\hat{\mathbf{g}}
= (\check{\mathbf{g}}^\triangleright)^{10} \cdot \vee_{\mathbf{A}}
= {\downarrow}_{\mathbf{B}} \cdot \check{\mathbf{g}}^{-1} \cdot \vee_{\mathbf{A}}
: \mathsf{elem}(\mathbf{B}) \rightarrow \mathsf{elem}(\mathbf{A})$
of the 10-fiber of the forward bimodule 
$\check{\mathbf{g}}^\triangleright
: \mathsf{ord}(\mathbf{A}) \rightharpoondown \mathsf{ord}(\mathbf{B})$
induced by the left adjoint monotonic morphism, and 
the left adjoint morphism is expressed in terms of the right adjoint morphism 
as the composition with target meet
$\check{\mathbf{g}}
= (\hat{\mathbf{g}}^\triangleleft)^{01} \cdot {\wedge}_{\mathbf{B}}
= {\uparrow}_{\mathbf{A}} \cdot \hat{\mathbf{g}}^{-1} \cdot {\wedge}_{\mathbf{B}}
: \mathsf{elem}(\mathbf{A}) \rightarrow \mathsf{elem}(\mathbf{B})$
of the 01-fiber of the reverse bimodule 
$\hat{\mathbf{g}}^\triangleleft
: \mathsf{ord}(\mathbf{A}) \rightharpoondown \mathsf{ord}(\mathbf{B})$
induced by the right adjoint monotonic morphism.

\begin{lemma}\label{induce:lattice}
The following properties hold.
\begin{itemize}
\item Let $\mathbf{g} : \mathbf{A} \rightleftharpoons \mathbf{B}$ be a reflection.
If the source $\mathbf{A}$ is a poset,
then the target $\mathbf{B}$ is also a poset.
If the source $\mathbf{A}$ is a complete lattice,
then the target $\mathbf{B}$ is a complete lattice
with the definitions
$\bigvee_{\mathbf{B}}Y 
= \check{\mathbf{g}}\left( \bigvee_{\mathbf{A}} \hat{\mathbf{g}}[Y] \right)$
and
$\bigwedge_{\mathbf{B}}Y 
= \check{\mathbf{g}}\left( \bigwedge_{\mathbf{A}} \hat{\mathbf{g}}[Y] \right)$
for any target subobject $Y$.
Also,
the following identities hold:
$\hat{\mathbf{g}} \left( \bigvee_{\mathbf{B}}Y \right)
= {\left( \bigvee_{\mathbf{A}} \hat{\mathbf{g}}[Y] \right)}^{\bullet}$
and
$\hat{\mathbf{g}}\left( \bigwedge_{\mathbf{B}}Y \right)
= \bigwedge_{\mathbf{A}} \hat{\mathbf{g}}[Y]$.
\item Let $\mathbf{g} : \mathbf{A} \rightleftharpoons \mathbf{B}$ be a coreflection.
If the target $\mathbf{B}$ is a poset,
then the source $\mathbf{A}$ is also a poset.
If the target $\mathbf{B}$ is a complete lattice,
then the source $\mathbf{A}$ is a complete lattice
with the definitions
$\bigwedge_{\mathbf{A}}X 
= \hat{\mathbf{g}}\left( \bigwedge_{\mathbf{B}} \check{\mathbf{g}}[X] \right)$
and
$\bigvee_{\mathbf{A}}X 
= \hat{\mathbf{g}}\left( \bigvee_{\mathbf{B}} \check{\mathbf{g}}[X] \right)$
for any source subobject $X$.
Also,
the following identities hold:
$\check{\mathbf{g}} ( \bigwedge_{\mathbf{A}}X )
= {\left( \bigwedge_{\mathbf{B}} \check{\mathbf{g}}[X] \right)}^{\circ}$
and
$\check{\mathbf{g}}(\bigvee_{\mathbf{A}}X)
= \bigvee_{\mathbf{B}} \check{\mathbf{g}}[X]$.
\end{itemize}
\end{lemma}

%%%%%%%%%%%%%%%%%%%%%%%%%%%%%%%%%%%%%%%%%%%%%%%%%%%%%%%%%%%%%%%%%%%%%%
\subsection{Multiplication and Exponent}\label{subsec:exponent}
%%%%%%%%%%%%%%%%%%%%%%%%%%%%%%%%%%%%%%%%%%%%%%%%%%%%%%%%%%%%%%%%%%%%%%

Given any two classifications
$\mathbf{A}$ 
%= \langle \mathsf{inst}(\mathbf{A}), \mathsf{typ}(\mathbf{A}), \models_{\mathbf{A}} \rangle$ 
and
$\mathbf{B}$ 
%= \langle \mathsf{inst}(\mathbf{B}), \mathsf{typ}(\mathbf{B}), \models_{\mathbf{B}} \rangle$,
the \emph{exponent} classification ${\mathbf{B}}^{\mathbf{A}}$ is defined as follows.
\begin{itemize}
\item
The instance $\mathcal{B}$-object is the pullback\footnote{Composition of ${\mathsf{int}_{\mathbf{A}}}^{\!\mathsf{inst}(\mathbf{B})}
: {\mathsf{inst}(\mathbf{A})}^{\mathsf{inst}(\mathbf{B})}
\rightarrow {{\wp}\mathsf{typ}(\mathbf{A})}^{\mathsf{inst}(\mathbf{B})}$
with the isomorphism
${{\wp}\mathsf{typ}(\mathbf{A})}^{\mathsf{inst}(\mathbf{B})}
\cong {\wp}{\left(\mathsf{inst}(\mathbf{B}){\times}\mathsf{typ}(\mathbf{A})\right)}$
gives the $\mathcal{B}$-morphism on the bottom of Figure~\ref{exponent},
and composition of 
${\mathsf{ext}_{\mathbf{B}}}^{\!\mathsf{typ}(\mathbf{A})}
: {\mathsf{typ}(\mathbf{B})}^{\mathsf{typ}(\mathbf{A})}
\rightarrow {{\wp}\mathsf{inst}(\mathbf{B})}^{\mathsf{typ}(\mathbf{A})}$.
with the isomorphism
${{\wp}\mathsf{inst}(\mathbf{B})}^{\mathsf{typ}(\mathbf{A})}
\cong {\wp}{\left(\mathsf{typ}(\mathbf{A}){\times}\mathsf{inst}(\mathbf{B})\right)}$
gives the $\mathcal{B}$-morphism on the right of Figure~\ref{exponent}.}
$\mathsf{inst}({\mathbf{B}}^{\mathbf{A}})
= \mathcal{B}\left({\mathbf{A}},{\mathbf{B}}\right)$
in Figure~\ref{exponent}.
\item The type $\mathcal{B}$-object is the binary product
$\mathsf{typ}({\mathbf{B}}^{\mathbf{A}})
= \mathsf{inst}(\mathbf{B}){\times}\mathsf{typ}(\mathbf{A})$.
\item The character of the classification relation
$\models_{{\mathbf{B}}^{\mathbf{A}}} : \mathsf{inst}({\mathbf{B}}^{\mathbf{A}}) \rightharpoondown \mathsf{typ}({\mathbf{B}}^{\mathbf{A}})$
is defined in terms of the equalizing monomorphism
$\mathsf{inst}({\mathbf{B}}^{\mathbf{A}}) \hookrightarrow
{\mathsf{inst}(\mathbf{A})}^{\mathsf{inst}(\mathbf{B})}
\!\!{\times}
{\mathsf{typ}(\mathbf{B})}^{\mathsf{typ}(\mathbf{A})}$,
the evaluation morphisms
${\mathsf{inst}(\mathbf{B})}{\times}{{\mathsf{inst}(\mathbf{A})}^{\mathsf{inst}(\mathbf{B})}} \!\!{\rightarrow} {\mathsf{inst}(\mathbf{A})}$
and
${\mathsf{typ}(\mathbf{A})}{\times}{{\mathsf{typ}(\mathbf{B})}^{\mathsf{typ}(\mathbf{A})}} \!\!{\rightarrow} {\mathsf{typ}(\mathbf{B})}$,
and the common image character
$\mathsf{inst}(\mathbf{A}){\times}\mathsf{typ}(\mathbf{B}) \rightarrow \Omega$,
as the composite
\begin{center}
$\begin{array}{rcl}
\mathsf{inst}({\mathbf{B}}^{\mathbf{A}})
{\times} \mathsf{inst}(\mathbf{B}){\times}\mathsf{typ}(\mathbf{A})
& \hookrightarrow &
{{\mathsf{inst}(\mathbf{A})}^{\mathsf{inst}(\mathbf{B})}}
\!{\times}
{\mathsf{typ}(\mathbf{B})}^{\mathsf{typ}(\mathbf{A})}
\!{\times}
\mathsf{inst}(\mathbf{B}){\times}\mathsf{typ}(\mathbf{A})
\\
& \cong &
{\mathsf{inst}(\mathbf{B})}{\times}{{\mathsf{inst}(\mathbf{A})}^{\mathsf{inst}(\mathbf{B})}}
\!{\times}
{\mathsf{typ}(\mathbf{A})}{\times}{{\mathsf{typ}(\mathbf{B})}^{\mathsf{typ}(\mathbf{A})}}
\\
& \rightarrow & \mathsf{inst}(\mathbf{A}){\times}\mathsf{typ}(\mathbf{B})
\\
& \rightarrow & \Omega
\end{array}$
\end{center}
\end{itemize}
Given two classifications $\mathbf{A}$ and $\mathbf{B}$, 
the \emph{multiplication} classification
${\mathbf{A}} \otimes {\mathbf{B}}$
is the involution of the exponent ${\mathbf{B}}^{{\mathbf{A}}^{\!\propto}}$.

\begin{figure}
\begin{center}
\setlength{\unitlength}{0.6pt}
\begin{picture}(160,80)(0,0)
\put(-30,65){\makebox(60,30){$\mathcal{B}\left({\mathbf{A}},{\mathbf{B}}\right)$}}
\put(130,65){\makebox(60,30){${\mathsf{typ}(\mathbf{B})}^{\mathsf{typ}(\mathbf{A})}$}}
\put(-30,-15){\makebox(60,30){${\mathsf{inst}(\mathbf{A})}^{\mathsf{inst}(\mathbf{B})}$}}
\put(130,-15){\makebox(60,30){${\wp}{\left(\mathsf{inst}(\mathbf{B}){\times}\mathsf{typ}(\mathbf{A})\right)}$}}
\put(208,65){\makebox(30,30){\footnotesize{$\ni\; \hat{\mathbf{f}}$}}}
\put(-70,65){\makebox(30,30){\footnotesize{$\mathbf{f} \;\in$}}}
\put(-78,-15){\makebox(30,30){\footnotesize{$\check{\mathbf{f}} \;\in$}}}
\put(38,80){\vector(1,0){70}}
\put(47,0){\vector(1,0){42}}
\put(0,65){\vector(0,-1){50}}
\put(160,65){\vector(0,-1){50}}
\end{picture}
\end{center}
\caption{Exponent in $\mathsf{clsn}(\mathcal{B})$}
\label{exponent}
\end{figure}

%%%%%%%%%%%%%%%%%%%%%%%%%%%%%%%%%%%%%%%%%%%%%%%%%%%%%%%%%%%%%%%%%%%%%%
\subsection{The Concept Lattice Functor}\label{subsec:concept:lattice:functor}
%%%%%%%%%%%%%%%%%%%%%%%%%%%%%%%%%%%%%%%%%%%%%%%%%%%%%%%%%%%%%%%%%%%%%%

The category of classifications and infomorphisms is a subcategory of the arrow category of adjunctions
$\mathsf{incl}_{\mathcal{B}} 
: \mathsf{Clsn}(\mathcal{B}) \rightarrow \mathsf{Adj}(\mathcal{B})_{=}^{\mathsf{2}}$,
and the category of concept lattices and concept morphisms is a subcategory of the factorization category of adjunctions
$\mathsf{incl}_{\mathcal{B}} 
: \mathsf{Clg}(\mathcal{B}) \rightarrow \mathsf{Ref}(\mathcal{B}) \odot \mathsf{Ref}(\mathcal{B})^\propto$.
The concept lattice functor 
$\mathsf{clg}_{\mathcal{B}} 
: \mathsf{Clsn}(\mathcal{B}) \rightarrow \mathsf{Clg}(\mathcal{B}) 
= \mathsf{Clg}(\mathcal{B})_\iota \odot \mathsf{Clg}(\mathcal{B})_\tau$
is the restriction of the polar factorization functor
$\div_{\mathsf{Adj}_{=}} 
: \mathsf{Adj}(\mathcal{B})_{=}^{\mathsf{2}} \rightarrow
\mathsf{Ref}(\mathcal{B}) \odot \mathsf{Ref}(\mathcal{B})^\propto$
to $\mathsf{Clsn}(\mathcal{B})$ at the source and $\mathsf{Clg}(\mathcal{B})$ at the target.
This can be verified by definition of $\mathsf{Clg}(\mathcal{B})$.

%\include{clsn}

%%%%%%%%%%%%%%%%%%%%%%%%%%%%%%%%%%%%%%%%%%%%%%%%%%%%%%%%%%%%%%%%%%%%%%%%%%%%%%%%
%%%%%%%%%%%%%%%%%%%%%%%%%%%%%%%%%%%%%%%%%%%%%%%%%%%%%%%%%%%%%%%%%%%%%%%%%%%%%%%%

%%%%%%%%%%%%%%%%%%%%%%%%%%%%%%%%%%%%%%%%%%%%%%%%%%%%%%%%%%%%%%%%%%%%%%%%%%%%%%%%
%%%%%%%%%%%%%%%%%%%%%%%%%%%%%%%%%%%%%%%%%%%%%%%%%%%%%%%%%%%%%%%%%%%%%%%%%%%%%%%%

%%%%%%%%%%%%%%%%%%%%%%%%%%%%%%%%%%%%%%%%%%%%%%%%%%%%%%%%%%%%%%%%%%%%%%%%%%%%%%%%
%%%%%%%%%%%%%%%%%%%%%%%%%%%%%%%%%%%%%%%%%%%%%%%%%%%%%%%%%%%%%%%%%%%%%%%%%%%%%%%%
\section{Conceptual Structures}\label{sec:conceptual:structures}
%%%%%%%%%%%%%%%%%%%%%%%%%%%%%%%%%%%%%%%%%%%%%%%%%%%%%%%%%%%%%%%%%%%%%%%%%%%%%%%%
%%%%%%%%%%%%%%%%%%%%%%%%%%%%%%%%%%%%%%%%%%%%%%%%%%%%%%%%%%%%%%%%%%%%%%%%%%%%%%%%

%%%%%%%%%%%%%%%%%%%%%%%%%%%%%%%%%%%%%%%%%%%%%%%%%%%%%%%%%%%%%%%%%%%%%%
\subsection{Concept Lattices}\label{subsec:concept:lattices}
%%%%%%%%%%%%%%%%%%%%%%%%%%%%%%%%%%%%%%%%%%%%%%%%%%%%%%%%%%%%%%%%%%%%%%

A \emph{conceptual structure} factors as, and is composed of, two aspects: 
an extensional or denotative aspect and an intensional or connotative aspect.
It consists of a hierarchy of concepts,
a generalization-specialization hierarchy,
that extensionally links instances to concepts and intensionally links concepts to types.
Both aspects of conceptual structure can be represented in three equivalent versions:
a relation version, a morphism version and an adjunction version.
More specifically,
a conceptual structure $\mathbf{L}$ in (internal to) a topos $\mathcal{B}$
has three components,
a $\mathcal{B}$-object of \emph{instances} $\mathsf{inst}(\mathbf{L})$, 
a $\mathcal{B}$-object of \emph{types} $\mathsf{typ}(\mathbf{L})$
and a complete $\mathcal{B}$-lattice of \emph{concepts}
$\mathbf{L} = \langle L, \leq_{\mathbf{L}}, {\vee}_{\mathbf{L}}, {\wedge}_{\mathbf{L}} \rangle$
that represents the conceptual hierarchy.

\begin{figure}
\begin{center}
\begin{tabular}{l@{\hspace{10pt}}l}
\underline{\sffamily extensional aspect} & \underline{\sffamily intensional aspect}
\\ \\
{\footnotesize \begin{tabular}{cr@{\hspace{5pt}:\hspace{5pt}}l}
instance-of & $\iota_{\mathbf{L}}$ & $\mathsf{inst}(\mathbf{L}) \rightharpoondown \mathbf{L}$ 
\\
instance embedding & $\mathbf{i}_{\mathbf{L}}$ & $\mathsf{inst}(\mathbf{L}) \rightarrow \mathbf{L}$
\\
extent & $\mathsf{ext}_{\mathbf{L}}$ & $\mathbf{L} \rightarrow {\wp}\mathsf{inst}(\mathbf{L})$
\\
iota & $\mathsf{iota}_{\mathbf{L}}$ & ${\wp}\mathsf{inst}(\mathbf{L}) \rightarrow \mathbf{L}$
\end{tabular}}
&
{\footnotesize \begin{tabular}{cr@{\hspace{5pt}:\hspace{5pt}}l}
of-type & $\tau_{\mathbf{L}}$ & $\mathbf{L} \rightharpoondown \mathsf{typ}(\mathbf{L})$
\\
type embedding & $\mathbf{t}_{\mathbf{L}}$ & $\mathsf{typ}(\mathbf{L}) \rightarrow \mathbf{L}$
\\
intent & $\mathsf{int}_{\mathbf{L}}^{\propto}$ & $\mathbf{L} \rightarrow {\wp}\,\mathsf{typ}(\mathbf{L})^{\propto}$
\\
tau & $\mathsf{tau}_{\mathbf{L}}$ & ${\wp}\,\mathsf{typ}(\mathbf{L})^{\propto} \!\!\rightarrow \mathbf{L}$
\end{tabular}}
\end{tabular}
\end{center}
\caption{Equivalent Components}
\end{figure}
The relation version of conceptual structure has an extensional aspect,
represented by the \emph{instance-of} bimodule 
$\iota_{\mathbf{L}} : \mathsf{inst}(\mathbf{L}) \rightharpoondown \mathbf{L}$,
that registers which instances belong to which concepts,
and has an intensional aspect,
represented by the \emph{of-type} bimodule 
$\tau_{\mathbf{L}} : \mathbf{L} \rightharpoondown \mathsf{typ}(\mathbf{L})$,
that describes the concepts by recording the types of each.
Being bimodules,
the instance-of relation is closed on the right with respect to concept order
and the of-type relation is closed on the left with respect to concept order.
The morphism version of conceptual structure has an extensional aspect,
represented by the \emph{instance embedding} morphism
$\mathbf{i}_{\mathbf{L}} : \mathsf{inst}(\mathbf{L}) \rightarrow \mathbf{L}$,
and has an intensional aspect represented by the \emph{type embedding} morphism
$\mathbf{t}_{\mathbf{L}} : \mathsf{typ}(\mathbf{L}) \rightarrow \mathbf{L}$.
We assume 
that instance-of is the forward bimodule of the instance embedding morphism
$\iota_{\mathbf{L}} = \mathbf{i}_{\mathbf{L}}^{\triangleright}$,
and that of-type is the reverse bimodule of the type embedding morphism
$\tau_{\mathbf{L}} = \mathbf{t}_{\mathbf{L}}^{\triangleleft}$.
It follows that 
instance embedding is the 01-embedding morphism of the instance-of bimodule
$\mathbf{i}_{\mathbf{L}} = \iota_{\mathbf{L}}^{\wedge}$,
and that type embedding is the 10-embedding morphism of the of-type bimodule
$\mathbf{t}_{\mathbf{L}} = \tau_{\mathbf{L}}^{\vee}$.
Although special kinds of relations,
the instance-of relation is equivalent to the instance embedding morphism,
and the of-type relation is equivalent to the type embedding morphism.

The adjunction version of conceptual structure 
consists of an extensional reflection and an intensional coreflection 
that are composable (as adjunctions).
We further assume that
the source preorder of the extent
and the target preorder of the intent
are free.
More formally,
the adjunction version of conceptual structure 
consists of the $\mathcal{B}$-preorder
%%%%%%%%%%%%%%%%%%%%%%%%%%%%%%%%%%%%%%%%%%%%%%%%%%%%%%%%%%%%%%%%%%%%%%%%%%%%%%%%
%%%%%%%%%%%%%%%%%%%%%%%%%%%%%%%%%%%%%%%%%%%%%%%%%%%%%%%%%%%%%%%%%%%%%%%%%%%%%%%%
\footnote{\label{12}For the relation and morphism versions,
we need the following additional assumptions (restrictions):
the concept order is a complete lattice
(it satisfies antisymmetry, and the meets and joins of all subsets exist),
the subobject of embedded instances is join-dense, 
and the subobject of embedded types is meet-dense.
We need these assumptions in order to move from either the relation or morphism versions to the adjunction version.
For the adjunction version 
we need no additional assumptions;
that is,
we initially assume only that the concept hierarchy is a preorder.
All the restrictions come from the assumptions about reflections and coreflections.
By Lem.~\ref{induce:lattice}, 
antisymmetry and existence of meets and joins follow from (co)reflection properties and the fact that instance (type) power is a complete lattice.
The facts that embedded instances are join-dense and embedded types are meet-dense is equivalent to the equality constraints of the extent reflection and intent coreflection.}
%%%%%%%%%%%%%%%%%%%%%%%%%%%%%%%%%%%%%%%%%%%%%%%%%%%%%%%%%%%%%%%%%%%%%%%%%%%%%%%%
%%%%%%%%%%%%%%%%%%%%%%%%%%%%%%%%%%%%%%%%%%%%%%%%%%%%%%%%%%%%%%%%%%%%%%%%%%%%%%%%
of concepts
$\mathbf{L} = \langle L, \leq_{\mathbf{L}} \rangle$
extensionally linked to the complete lattice of instances via the \emph{extent} reflection,
and intensionally linked to the complete lattice of types via the \emph{intent} coreflection
\[
{\wp}\mathsf{inst}(\mathbf{L}) 
\stackrel{\mathsf{extent}_{\mathbf{L}}}{\rightleftharpoons} 
\mathbf{L}
\stackrel{\mathsf{intent}_{\mathbf{L}}}{\rightleftharpoons}
{\wp}\mathsf{typ}(\mathbf{L})^{\propto}.
\]
The extent reflection 
$\mathsf{extent}_{\mathbf{L}} 
= \langle \mathsf{iota}_{\mathbf{L}}, \mathsf{ext}_{\mathbf{L}} \rangle$
unpacks into the \emph{extent} monotonic morphism 
$\mathsf{ext}_{\mathbf{L}} 
: \mathbf{L} \rightarrow {\wp}\mathsf{inst}(\mathbf{L})$
and the instance concept generator (\emph{iota}) monotonic morphism
$\mathsf{iota}_{\mathbf{L}} 
: {\wp}\mathsf{inst}(\mathbf{L}) \rightarrow \mathbf{L}$.
The extent morphism,
which lists the instances of each concept,
is the 10-fiber of the instance-of bimodule
$\mathsf{ext}_{\mathbf{L}} = \iota_{\mathbf{L}}^{10}
: \mathbf{L} \rightarrow {\wp}\mathsf{inst}(\mathbf{L})$.
Conversely, the instance-of bimodule is expressed in terms of the extent morphism
either with the expression
$\iota_{\mathbf{L}}
= {\in}_{\mathsf{inst}(\mathbf{L})}(1_{\mathsf{inst}(\mathbf{L})},\mathsf{ext}_{\mathbf{L}})$
or via the infomorphism
$(1_{\mathsf{inst}(\mathbf{L})},\mathsf{ext}_{\mathbf{L}}) 
: (\mathsf{inst}(\mathbf{L}),\mathbf{L},\iota_{\mathbf{L}}) \rightleftharpoons (\mathsf{inst}(\mathbf{L}),{\wp}\mathsf{inst}(\mathbf{L}),{\in}_{\mathsf{inst}(\mathbf{L})})$.
The iota monotonic morphism,
which computes the most specific concept that contains all instances of an extent,
is expressed in terms of 
the instance-of bimodule and the instance embedding monotonic morphism
as
$\mathsf{iota}_{\mathbf{L}}
= \iota_{\mathbf{L}}^{\Rightarrow} \cdot {\wedge}_{\mathbf{L}}
= \exists\iota_{\mathbf{L}}^{01} \cdot {\cap}_{L} \cdot {\wedge}_{\mathbf{L}}
= \exists\iota_{\mathbf{L}}^{01} \cdot {\cap}_{L} \cdot {\wedge}_{\mathbf{L}}
= \exists\mathbf{i}_{\mathbf{L}} \cdot \exists{\uparrow}_{\mathbf{L}} \cdot {\cap}_{L} \cdot {\wedge}_{\mathbf{L}}
= \exists\mathbf{i}_{\mathbf{L}} \cdot {\Uparrow}_{\mathbf{L}} \cdot {\wedge}_{\mathbf{L}}
= \exists\mathbf{i}_{\mathbf{L}} \cdot {\vee}_{\mathbf{L}}
: {\wp}\,\mathsf{inst}(\mathbf{L}) \rightarrow \mathbf{L}$. 
The instance embedding morphism is the restriction of the iota morphism to single instances
$\iota_{\mathbf{L}}
= \{\mbox{-}\}_{\mathsf{inst}(\mathbf{L})} \cdot \mathsf{iota}_{\mathbf{L}}
: \mathsf{inst}(\mathbf{L}) \rightarrow \mathbf{L}$. 
Since the iota and extent morphisms are adjoint morphisms between complete lattices,
they determine each other.
Hence, 
the instance-of bimodule and the extent reflection are equivalent.

The intent coreflection 
$\mathsf{intent}_{\mathbf{L}} 
= \langle \mathsf{int}_{\mathbf{L}}, \mathsf{tau}_{\mathbf{L}} \rangle$
unpacks into
the \emph{intent} monotonic morphism 
$\mathsf{int}_{\mathbf{L}}^{\propto} 
: \mathbf{L} \rightarrow {\wp}\,\mathsf{typ}(\mathbf{L})^{\propto}$
and the type concept generator (\emph{tau}) monotonic morphism
$\mathsf{tau}_{\mathbf{L}} 
: {\wp}\,\mathsf{typ}(\mathbf{L})^{\propto} \rightarrow \mathbf{L}$.
The intent morphism,
which collects the types possessed by each concept,
is the 01-fiber of the of-type bimodule
$\mathsf{int}_{\mathbf{L}} = \tau_{\mathbf{L}}^{01}
: \mathbf{L}^{\propto} \rightarrow {\wp}\,\mathsf{typ}(\mathbf{L})$.
Conversely, the of-type bimodule is expressed in terms of the intent morphism
either with the expression
$\tau_{\mathbf{L}}
= {\in}^{\propto}_{\mathsf{typ}(\mathbf{L})}(\mathsf{int}_{\mathbf{L}},1_{\mathsf{typ}(\mathbf{L})})$
or via the infomorphism
$(\mathsf{int}_{\mathbf{L}},1_{\mathsf{typ}(\mathbf{L})}) : ({\wp}\mathsf{typ}(\mathbf{L}),\mathsf{typ}(\mathbf{L}),{\in}^{\propto}_{\mathsf{typ}(\mathbf{L})}) \rightleftharpoons (\mathbf{L},\mathsf{typ}(\mathbf{L}),\tau_{\mathbf{L}})$.
The tau morphism,
which computes the most generic concept that has all types of an intent, 
is expressed in terms of 
the of-type bimodule and the type embedding monotonic morphism 
as
$\mathsf{tau}_{\mathbf{L}}
= \tau_{\mathbf{L}}^{\Leftarrow} \cdot {\vee}_{\mathbf{L}}
= \exists\tau_{\mathbf{L}}^{10} \cdot {\cap}_{L} \cdot {\vee}_{\mathbf{L}}
= \exists\mathbf{t}_{\mathbf{L}} \cdot \exists{\downarrow}_{\mathbf{L}} \cdot {\cap}_{L} \cdot {\vee}_{\mathbf{L}}
= \exists\mathbf{t}_{\mathbf{L}} \cdot {\Downarrow}_{\mathbf{L}} \cdot {\vee}_{\mathbf{L}}
= \exists\mathbf{t}_{\mathbf{L}} \cdot {\wedge}_{\mathbf{L}}
: {\wp}\,\mathsf{typ}(\mathbf{L})^{\propto} \rightarrow \mathbf{L}$. 
The type embedding morphism is the restriction of the tau morphism to single types
$\tau_{\mathbf{L}}
= \{\mbox{-}\}_{\mathsf{typ}(\mathbf{L})} \cdot \mathsf{tau}_{\mathbf{L}}
: \mathsf{typ}(\mathbf{L}) \rightarrow \mathbf{L}$. 
Since the tau and intent morphisms are adjoint morphisms between complete lattices,
they determine each other.
Hence, 
the of-type bimodule and the intent coreflection are equivalent.

To construct the associated classification structure
$\mathsf{clsn}(\mathbf{L})$,
compose either the relation, morphism or adjunction version.
The classification $\mathsf{clsn}(\mathbf{L})$
has $\mathbf{L}$-instances as its instances
and $\mathbf{L}$-types as its types.
The relational composition
$\models_{\mathsf{clsn}(\mathbf{L})} 
= \iota_{\mathbf{L}} \circ \tau_{\mathbf{L}}$
gives its classification relation.
The morphism compositions
$\mathsf{ext}_{\mathsf{clsn}(\mathbf{L})} 
= \tau_{\mathbf{L}} \cdot \mathsf{ext}_{\mathbf{L}}$
and
$\mathsf{int}_{\mathsf{clsn}(\mathbf{L})} 
= \iota_{\mathbf{L}} \cdot \mathsf{int}_{\mathbf{L}}$
give the instance and type embedding morphisms.
The adjunction composition
\[\mathsf{deriv}_{\mathsf{clsn}(\mathbf{L})}
= \mathsf{extent}_{\mathbf{L}} \circ \mathsf{intent}_{\mathbf{L}}.\]
gives the derivation Galois connection,
with the morphism composition 
$\mathsf{iota}_{\mathbf{L}} \cdot \mathsf{int}_{\mathbf{L}}$ 
giving the forward derivation morphism,
and the morphism composition 
$\mathsf{tau}_{\mathbf{L}} \cdot \mathsf{ext}_{\mathbf{L}}$ 
giving the reverse derivation morphism.
Moreover,
given any classification $\mathbf{A}$,
the classification of the concept lattice of $\mathbf{A}$ is itself:
$\mathsf{clsn}(\mathsf{clg}(\mathbf{A})) = \mathbf{A}$.

%%%%%%%%%%%%%%%%%%%%%%%%%%%%%%%%%%%%%%%%%%%%%%%%%%%%%%%%%%%%%%%%%%%%%%
\subsection{Concept Morphisms}\label{subsec:concept:morphisms}
%%%%%%%%%%%%%%%%%%%%%%%%%%%%%%%%%%%%%%%%%%%%%%%%%%%%%%%%%%%%%%%%%%%%%%

A \emph{morphism conceptual structures} factors as, and is composed of, two aspects: 
an extensional or denotative aspect and an intensional or connotative aspect.
Both aspects of conceptual structure morphism can be represented in three equivalent versions:
a relation version, a morphism version and an adjunction version.
More specifically,
a conceptual structure morphism 
is called a \emph{concept morphism} and is symbolized as
$\mathbf{h} : \mathbf{L}_1 \rightleftharpoons \mathbf{L}_2$
with source conceptual structure $\mathbf{L}_1$ and target conceptual structure $\mathbf{L}_2$.

%%%%%%%%%%%%%%%%%%%%%%%%%%%%%%%%%%%%%%%%%%%%%%%%%%%%%%%%%%%%%%%%%%%%%%
\subsubsection{Extension.}
%%%%%%%%%%%%%%%%%%%%%%%%%%%%%%%%%%%%%%%%%%%%%%%%%%%%%%%%%%%%%%%%%%%%%%

The extensional aspect of a concept morphism
$\mathbf{h} : \mathbf{L}_1 \rightleftharpoons \mathbf{L}_2$
consists of 
a \emph{conceptual connection} adjunction 
$\mathsf{adj}(\mathbf{h})
= \langle \mathsf{left}(\mathbf{h}), \mathsf{right}(\mathbf{h}) \rangle
: \mathbf{L}_2 \rightleftharpoons \mathbf{L}_1$
(in the reverse direction) between conceptual hierarchies,
and an \emph{instance} morphism
$\mathsf{inst}(\mathbf{h}) 
: \mathsf{inst}(\mathbf{L}_2) \rightarrow \mathsf{inst}(\mathbf{L}_1)$ 
(in the reverse direction) between instance ${\mathcal B}$-objects.

\begin{sloppypar}
In the relation version of the extensional aspect of conceptual morphisms,
the conceptual connection and instance morphism 
are required to preserve instance-of relationship 
by satisfying the \emph{extensional condition}
$\iota_{\mathbf{L}_1}(\mathsf{inst}(\mathbf{h}),1_{\mathbf{L}_1})
= \iota_{\mathbf{L}_2}(1_{\mathsf{inst}(\mathbf{L}_2)},\mathsf{right}(\mathbf{h}))$.
This condition states that
$\langle \mathsf{inst}(\mathbf{h}), \mathsf{right}(\mathbf{h}) \rangle 
: \langle \mathsf{inst}(\mathbf{L}_1), \mathbf{L}_1, \iota_{\mathbf{L}_1} \rangle \rightleftharpoons \langle \mathsf{inst}(\mathbf{L}_2), \mathbf{L}_2, \iota_{\mathbf{L}_2} \rangle$
is an infomorphism.
In the morphism version of the extensional aspect of conceptual morphisms,
the conceptual connection and instance morphism 
are required to preserve instance embedding:
$\iota_{\mathbf{L}_2} \cdot \mathsf{left}(\mathbf{h})
= \mathsf{inst}(\mathbf{h}) \cdot \iota_{\mathbf{L}_1}$;
in turn,
this implies the extensional condition.
In the adjunction version of the extensional aspect of conceptual morphisms,
the conceptual connection and instance morphism 
are required to preserve extent:
$\mathsf{extent}_{\mathbf{L}_2} \circ \mathsf{adj}(\mathbf{h})
= \mathsf{dir}(\mathsf{inst}(\mathbf{h})) \circ \mathsf{extent}_{\mathbf{L}_1}$;
in turn,
this implies the extensional condition.
This extent constraint unpacks into
the extent morphism identity
$\mathsf{right}(\mathbf{h}) \cdot \mathsf{ext}_{\mathbf{L}_2}
= \mathsf{ext}_{\mathbf{L}_1} \cdot {\mathsf{inst}(\mathbf{h})}^{-1}$
and
the iota morphism identity
$\mathsf{iota}_{\mathbf{L}_2} \cdot \mathsf{left}(\mathbf{h})
= \exists\,\mathsf{inst}(\mathbf{h}) \cdot \mathsf{iota}_{\mathbf{L}_1}$.
These are equivalent,
and both imply the extensional condition.
\end{sloppypar}

%%%%%%%%%%%%%%%%%%%%%%%%%%%%%%%%%%%%%%%%%%%%%%%%%%%%%%%%%%%%%%%%%%%%%%
\subsubsection{Intension.}
%%%%%%%%%%%%%%%%%%%%%%%%%%%%%%%%%%%%%%%%%%%%%%%%%%%%%%%%%%%%%%%%%%%%%%

The intensional aspect of a conceptual morphism
$\mathbf{h} : \mathbf{L}_1 \rightleftharpoons \mathbf{L}_2$
consists of 
a conceptual connection adjunction 
$\mathsf{adj}(\mathbf{h}) 
= \langle \mathsf{left}(\mathbf{h}), \mathsf{right}(\mathbf{h}) \rangle
: \mathbf{L}_2 \rightleftharpoons \mathbf{L}_1$
(in the reverse direction) between conceptual hierarchies,
and a \emph{type} morphism
$\mathsf{typ}(\mathbf{h}) 
: \mathsf{typ}(\mathbf{L}_1) \rightarrow \mathsf{typ}(\mathbf{L}_2)$ 
(in the forward direction) between type ${\mathcal B}$-objects.

\begin{sloppypar}
In the relation version of the intensional aspect of conceptual morphisms,
the conceptual connection and type morphism
are required to preserve of-type relationship 
by satisfying the \emph{intensional condition}
$\tau_{\mathbf{L}_1}(\mathsf{left}(\mathbf{h}),1_{\mathsf{typ}(\mathbf{L}_1)})
= \tau_{\mathbf{L}_2}(1_{\mathbf{L}_2},\mathsf{type}(\mathbf{h}))$.
This condition states that
$\langle \mathsf{left}(\mathbf{h}), \mathsf{typ}(\mathbf{h}) \rangle 
: \langle \mathbf{L}_1, \mathsf{typ}(\mathbf{L}_1), \tau_{\mathbf{L}_1} \rangle \rightleftharpoons \langle \mathbf{L}_2, \mathsf{typ}(\mathbf{L}_2), \tau_{\mathbf{L}_2} \rangle$
is an infomorphism.
In the morphism version of the intensional aspect of conceptual morphisms,
the conceptual connection and type morphism
are required to preserve type embedding
$\tau_{\mathbf{L}_1} \cdot \mathsf{right}(\mathbf{h})
= \mathsf{typ}(\mathbf{h}) \cdot \tau_{\mathbf{L}_2}$;
in turn,
this implies the intensional condition.
In the adjunction version of the intensional aspect of conceptual morphisms,
the conceptual connection and type morphism 
are required to preserve intent:
$\mathsf{adj}(\mathbf{h}) \circ \mathsf{intent}_{\mathbf{L}_1}
= \mathsf{intent}_{\mathbf{L}_2} \circ \mathsf{inv}(\mathsf{typ}(\mathbf{h}))$;
in turn,
this implies the intensional condition.
This intent constraint unpacks into
the intent morphism identity
$\mathsf{left}(\mathbf{h}) \cdot \mathsf{int}_{\mathbf{L}_1}
= \mathsf{int}_{\mathbf{L}_2} \cdot {\mathsf{typ}(\mathbf{h})}^{-1}$
and the tau morphism identity
$\mathsf{tau}_{\mathbf{L}_1} \cdot \mathsf{right}(\mathbf{h})
= \exists\,\mathsf{typ}(\mathbf{h}) \cdot \mathsf{tau}_{\mathbf{L}_2}$.
These are equivalent,
and both imply the extensional condition.
\end{sloppypar}

%%%%%%%%%%%%%%%%%%%%%%%%%%%%%%%%%%%%%%%%%%%%%%%%%%%%%%%%%%%%%%%%%%%%%%
\subsubsection{Abstraction.}
%%%%%%%%%%%%%%%%%%%%%%%%%%%%%%%%%%%%%%%%%%%%%%%%%%%%%%%%%%%%%%%%%%%%%%

The extensional aspect (top-upper part of Fig.~\ref{conceptual-structure}) is abstracted as 
a category $\mathsf{Clg}(\mathcal{B})_\iota$ of extensional conceptual structures,
with a contravariant instance component functor
$\mathsf{inst} : \mathsf{Clg}(\mathcal{B})_\iota^{\mathrm{op}} \rightarrow \mathcal{B}$,
a contravariant adjoint component functor
$\mathsf{adj} : \mathsf{Clg}(\mathcal{B})_\iota^{\mathrm{op}} \rightarrow \mathsf{Adj}(\mathcal{B})$,
and a natural transformation
$\mathsf{extent} 
: \mathsf{inst} \circ \mathsf{dir} \Rightarrow \mathsf{adj}
: \mathsf{Clg}(\mathcal{B})_\iota^{\mathrm{op}} \rightarrow \mathsf{Adj}(\mathcal{B})$.
The category $\mathsf{Clg}(\mathcal{B})_\iota^{\mathrm{op}}$ is a subcategory of 
$\mathsf{Ref}(\mathcal{B})$ the category of reflections,
with inclusion functor
$\mathsf{incl} : \mathsf{Clg}(\mathcal{B})_\iota^{\mathrm{op}} \rightarrow \mathsf{Ref}(\mathcal{B})$.
The components are related as
$\mathsf{inst} \circ \mathsf{dir} = \mathsf{incl} \circ \partial_0^{\mathrm{h}}$,
$\mathsf{adj} = \mathsf{incl} \circ \partial_1^{\mathrm{h}}$
and $\mathsf{extent} = \mathsf{incl} \circ \mathsf{ref}$.

The intensional aspect (top-lower part of Fig.~\ref{conceptual-structure}) is abstracted as 
a category $\mathsf{Clg}(\mathcal{B})_\tau$ of intensional conceptual structures,
with a contravariant adjoint component functor
$\mathsf{adj} : \mathsf{Clg}(\mathcal{B})_\tau^{\mathrm{op}} \rightarrow \mathsf{Adj}(\mathcal{B})$,
a covariant type component functor
$\mathsf{typ} : \mathsf{Clg}(\mathcal{B})_\tau \rightarrow \mathcal{B}$,
and a natural transformation
$\mathsf{intent} 
: \mathsf{adj} \Rightarrow {\mathsf{typ}}^{\mathrm{op}} \circ \mathsf{inv}
: \mathsf{Clg}(\mathcal{B})_\tau^{\mathrm{op}} \rightarrow \mathsf{Adj}(\mathcal{B})$.
The category $\mathsf{Clg}_\tau^{\mathrm{op}}$ is a subcategory of 
$\mathsf{Ref}(\mathcal{B})^{\propto}$ the category of coreflections,
with inclusion functor
$\mathsf{incl} : \mathsf{Clg}(\mathcal{B})_\tau^{\mathrm{op}} \rightarrow \mathsf{Ref}(\mathcal{B})^{\propto}$.
The components are related as
$\mathsf{adj} = \mathsf{incl} \circ \partial_0^{\mathrm{h}}$,
${\mathsf{typ}}^{\mathrm{op}} \circ \mathsf{inv} = \mathsf{incl} \circ \partial_1^{\mathrm{h}}$
and $\mathsf{intent} = \mathsf{incl} \circ \mathsf{ref}^\propto$.

Conceptual structures combine their extensional and intensional aspects 
by matching in the center.
The category of conceptual structures $\mathsf{Clg}(\mathcal{B})$ is the pullback
(bottom part of Fig.~\ref{conceptual-structure}) 
of $\mathsf{Clg}(\mathcal{B})_\iota$ and $\mathsf{Clg}(\mathcal{B})_\tau$ along their adjunction projections.
Hence,
there are two projection functors
$\pi_0 : \mathsf{Clg}(\mathcal{B}) \rightarrow \mathsf{Clg}(\mathcal{B})_\iota$
and $\pi_1 : \mathsf{Clg}(\mathcal{B}) \rightarrow \mathsf{Clg}(\mathcal{B})_\tau$
satisfying the identity
$\pi_0 \circ \mathsf{adj} = \pi_1 \circ \mathsf{adj}$.
The category of concept lattices and concept morphisms is a subcategory
(middle part of Fig.~\ref{conceptual-structure}) 
of the factorization category of adjunctions
$\mathsf{incl} 
: \mathsf{Clg}(\mathcal{B}) \rightarrow \mathsf{Ref}(\mathcal{B}) \odot \mathsf{Ref}(\mathcal{B})^{\propto}$
satisfying the identities
$\mathsf{incl} \circ \pi_0 = \pi_0 \circ \mathsf{incl}$
and $\mathsf{incl} \circ \pi_1 = \pi_1 \circ \mathsf{incl}$,
and the category of classifications and infomorphisms is a subcategory of the arrow category of adjunctions
$\mathsf{incl} : \mathsf{Clsn}(\mathcal{B}) \rightarrow \mathsf{Adj}(\mathcal{B})_{=}^{\mathsf{2}}$.
The classification functor 
$\mathsf{clsn} 
: \mathsf{Clg} = \mathsf{Clg}(\mathcal{B})_\iota \odot \mathsf{Clg}(\mathcal{B})_\tau \rightarrow \mathsf{Clsn}(\mathcal{B})$
is the restriction of the composition functor
$\circ_{\mathsf{Adj}_{=}} 
: \mathsf{Ref}(\mathcal{B}) \odot \mathsf{Ref}(\mathcal{B})^\propto
\rightarrow \mathsf{Adj}(\mathcal{B})_{=}^{\mathsf{2}} $
to $\mathsf{Clg}(\mathcal{B})$ at the source and $\mathsf{Clsn}(\mathcal{B})$ at the target.
This can be verified by definition of $\mathsf{Clsn}(\mathcal{B})$.
Also,
$\mathsf{clg} \circ \mathsf{clsn} = \mathsf{id}_{\mathsf{Clsn}(\mathcal{B})}$
and $\mathsf{clsn} \circ \mathsf{clg} \cong \mathsf{id}_{\mathsf{Clg}(\mathcal{B})}$.

\begin{theorem} [Restricted Equivalence] \label{restricted-equivalence}
The category $\mathsf{Clsn}(\mathcal{B})$
is equivalent 
(middle part of Fig.~\ref{conceptual-structure}) 
to the category $\mathsf{Clg}(\mathcal{B})$
\[\mathsf{Clsn}(\mathcal{B}) \equiv \mathsf{Clg}(\mathcal{B})\footnote{This is a category-theoretic rendering of the fundamental theorem of Formal Concept Analysis. 
See \cite{ganter:wille:99} and \cite{kent:02}.}.\]
\end{theorem}
This equivalence,
mediated by 
the restricted factorization of the concept lattice functor
and the restricted composition of the classification functor,
is a restriction of the special equivalence for adjunctions
(Thm.~\ref{special-equivalence}).

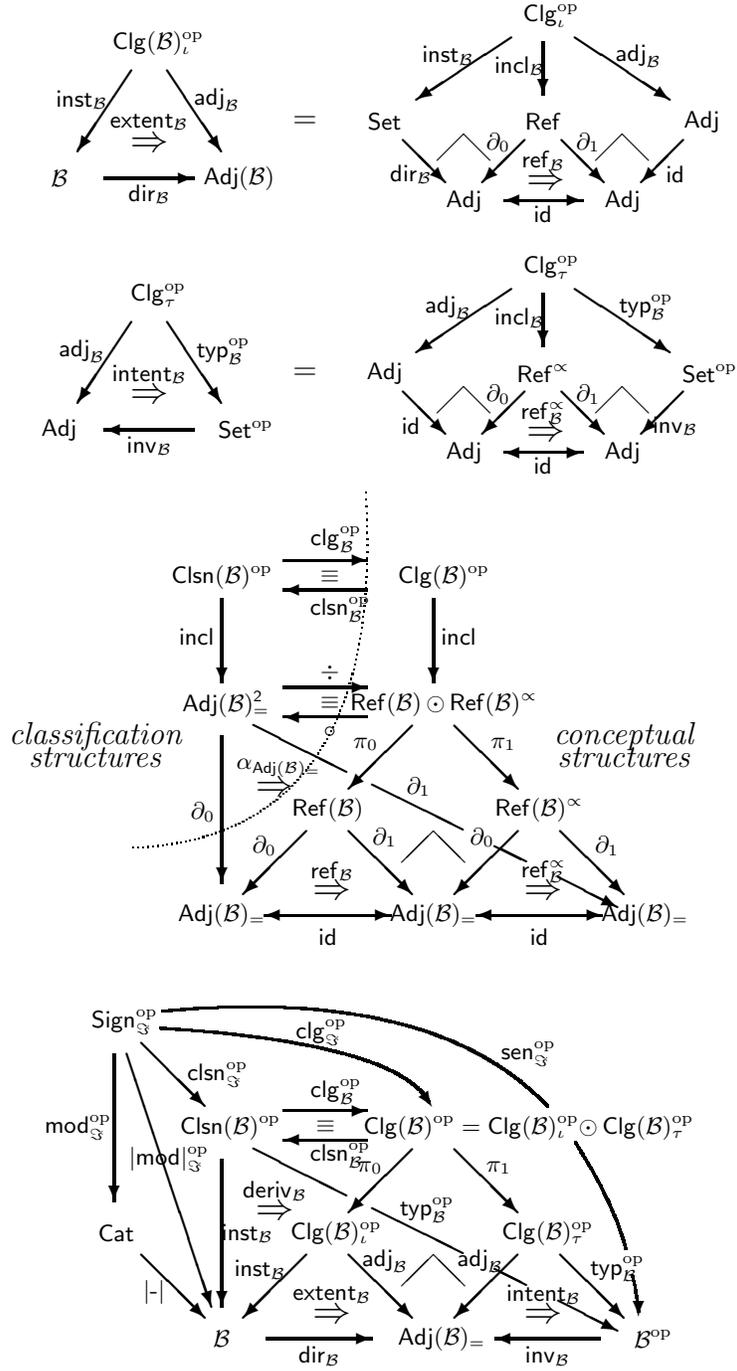
\begin{figure}
\begin{center}
\begin{tabular}{c}
\begin{tabular}{c@{\hspace{25pt}}c@{\hspace{30pt}}c}
\setlength{\unitlength}{0.85pt}
%\linethickness{0.5pt}
\begin{picture}(80,70)(0,-10)
\put(14,45){\makebox(60,30){$\mathsf{Clg}(\mathcal{B})_{\iota}^{\mathrm{op}}$}}
\put(-30,-15){\makebox(60,30){$\mathcal{B}$}}
\put(50,-15){\makebox(60,30){$\mathsf{Adj}(\mathcal{B})$}}
\put(-20,20){\makebox(60,30){\footnotesize{$\mathsf{inst}_{\mathcal{B}}$}}}
\put(40,20){\makebox(60,30){\footnotesize{$\mathsf{adj}_{\mathcal{B}}$}}}
\put(10,-22){\makebox(60,30){\footnotesize{$\mathsf{dir}_{\mathcal{B}}$}}}
\put(10,10){\makebox(60,30){\footnotesize{$\mathsf{extent}_{\mathcal{B}}$}}}
\put(10,0){\makebox(60,30){\Large{$\Rightarrow$}}}
\thicklines
\put(32,48){\vector(-2,-3){24}}
\put(48,48){\vector(2,-3){24}}
\put(20,0){\vector(1,0){40}}
\end{picture}
& 
\begin{picture}(0,0)(0,-30)
\put(0,0){\makebox(0,0){\large{$=$}}}
\end{picture}
&
\begin{picture}(120,70)(0,-30)
\put(33,25){\makebox(60,30){$\mathsf{Clg}_{\iota}^{\mathrm{op}}$}}
\put(-30,-15){\makebox(60,30){$\mathsf{Set}$}}
\put(30,-15){\makebox(60,30){$\mathsf{Ref}$}}
\put(90,-15){\makebox(60,30){$\mathsf{Adj}$}}
\put(0,-45){\makebox(60,30){$\mathsf{Adj}$}}
\put(60,-45){\makebox(60,30){$\mathsf{Adj}$}}
\put(-6,10){\makebox(60,30){\footnotesize{$\mathsf{inst}_{\mathcal{B}}$}}}
\put(21,6){\makebox(60,30){\footnotesize{$\mathsf{incl}_{\mathcal{B}}$}}}
\put(66,10){\makebox(60,30){\footnotesize{$\mathsf{adj}_{\mathcal{B}}$}}}
\put(-20,-35){\makebox(60,30){\footnotesize{$\mathsf{dir}_{\mathcal{B}}$}}}
\put(80,-35){\makebox(60,30){\footnotesize{$\mathsf{id}$}}}
\put(13,-23){\makebox(60,30){\footnotesize{$\partial_0$}}}
\put(47,-23){\makebox(60,30){\footnotesize{$\partial_1$}}}
\put(30,-30){\makebox(60,30){\footnotesize{$\mathsf{ref}_{\mathcal{B}}$}}}
\put(31,-38){\makebox(60,30){\Large{$\Rightarrow$}}}
\put(20,-15){\begin{picture}(20,10)(0,0)
\put(10,10){\line(-1,-1){10}}
\put(10,10){\line(1,-1){10}}
\end{picture}}
\put(80,-15){\begin{picture}(20,10)(0,0)
\put(10,10){\line(-1,-1){10}}
\put(10,10){\line(1,-1){10}}
\end{picture}}
\thicklines
\put(48,32){\vector(-3,-2){36}}
\put(72,32){\vector(3,-2){36}}
\put(60,30){\vector(0,-1){20}}
\put(7,-7){\vector(1,-1){16}}
\put(53,-7){\vector(-1,-1){16}}
\put(67,-7){\vector(1,-1){16}}
\put(113,-7){\vector(-1,-1){16}}
\put(45,-30){\vector(1,0){30}}
\put(45,-30){\vector(-1,0){0}}
\put(30,-50){\makebox(60,30){\footnotesize{$\mathsf{id}$}}}
\end{picture}
\\ \\ \\
\setlength{\unitlength}{0.85pt}
%\linethickness{0.5pt}
\begin{picture}(80,70)(0,-10)
\put(14,45){\makebox(60,30){$\mathsf{Clg}_{\tau}^{\mathrm{op}}$}}
\put(-30,-15){\makebox(60,30){$\mathsf{Adj}$}}
\put(53,-15){\makebox(60,30){$\mathsf{Set}^{\mathrm{op}}$}}
\put(-20,20){\makebox(60,30){\footnotesize{$\mathsf{adj}_{\mathcal{B}}$}}}
\put(43,20){\makebox(60,30){\footnotesize{$\mathsf{typ}^{\mathrm{op}}_{\mathcal{B}}$}}}
\put(10,-22){\makebox(60,30){\footnotesize{$\mathsf{inv}_{\mathcal{B}}$}}}
\put(10,10){\makebox(60,30){\footnotesize{$\mathsf{intent}_{\mathcal{B}}$}}}
\put(10,0){\makebox(60,30){\Large{$\Rightarrow$}}}
\thicklines
\put(32,48){\vector(-2,-3){24}}
\put(48,48){\vector(2,-3){24}}
\put(60,0){\vector(-1,0){40}}
\end{picture}
& 
\begin{picture}(0,0)(0,-30)
\put(0,0){\makebox(0,0){\large{$=$}}}
\end{picture}
&
\begin{picture}(120,70)(0,-30)
\put(33,25){\makebox(60,30){$\mathsf{Clg}_{\tau}^{\mathrm{op}}$}}
\put(-30,-15){\makebox(60,30){$\mathsf{Adj}$}}
\put(30,-15){\makebox(60,30){$\mathsf{Ref}^\propto$}}
\put(93,-15){\makebox(60,30){$\mathsf{Set}^{\mathrm{op}}$}}
\put(0,-45){\makebox(60,30){$\mathsf{Adj}$}}
\put(60,-45){\makebox(60,30){$\mathsf{Adj}$}}
\put(-6,10){\makebox(60,30){\footnotesize{$\mathsf{adj}_{\mathcal{B}}$}}}
\put(21,6){\makebox(60,30){\footnotesize{$\mathsf{incl}_{\mathcal{B}}$}}}
\put(69,10){\makebox(60,30){\footnotesize{$\mathsf{typ}^{\mathrm{op}}_{\mathcal{B}}$}}}
\put(-20,-35){\makebox(60,30){\footnotesize{$\mathsf{id}$}}}
\put(80,-35){\makebox(60,30){\footnotesize{$\mathsf{inv}_{\mathcal{B}}$}}}
\put(13,-23){\makebox(60,30){\footnotesize{$\partial_0$}}}
\put(47,-23){\makebox(60,30){\footnotesize{$\partial_1$}}}
\put(30,-30){\makebox(60,30){\footnotesize{$\mathsf{ref}^\propto_{\mathcal{B}}$}}}
\put(31,-38){\makebox(60,30){\Large{$\Rightarrow$}}}
\put(80,-15){\begin{picture}(20,10)(0,0)
\put(10,10){\line(-1,-1){10}}
\put(10,10){\line(1,-1){10}}
\end{picture}}
\put(20,-15){\begin{picture}(20,10)(0,0)
\put(10,10){\line(-1,-1){10}}
\put(10,10){\line(1,-1){10}}
\end{picture}}
\thicklines
\put(48,32){\vector(-3,-2){36}}
\put(72,32){\vector(3,-2){36}}
\put(60,30){\vector(0,-1){20}}
\put(7,-7){\vector(1,-1){16}}
\put(53,-7){\vector(-1,-1){16}}
\put(67,-7){\vector(1,-1){16}}
\put(113,-7){\vector(-1,-1){16}}
\put(45,-30){\vector(1,0){30}}
\put(45,-30){\vector(-1,0){0}}
\put(30,-50){\makebox(60,30){\footnotesize{$\mathsf{id}$}}}
\end{picture}
\end{tabular}
\\ \\ \\ \\ \\
\setlength{\unitlength}{0.8pt}
\begin{picture}(200,160)(-75,0)
\qbezier[100](18,200)(28,37)(-92,32)
\put(-160,60){\makebox(100,50){ {\itshape\large classification} }}
\put(-160,50){\makebox(100,50){ {\itshape\large structures} }}
\put(90,60){\makebox(100,50){ {\itshape\large conceptual} }}
\put(90,50){\makebox(100,50){ {\itshape\large structures} }}
\put(0,0){\begin{picture}(100,100)(0,0)
\put(5,75){\makebox(100,50){$\mathsf{Ref}(\mathcal{B}) \,{\odot}\, \mathsf{Ref}(\mathcal{B})^\propto$}}
\put(-50,25){\makebox(100,50){$\mathsf{Ref}(\mathcal{B})$}}
\put(50,25){\makebox(100,50){$\mathsf{Ref}(\mathcal{B})^{\propto}$}}
\put(-100,-25){\makebox(100,50){$\mathsf{Adj}(\mathcal{B})_{=}$}}
\put(0,-25){\makebox(100,50){$\mathsf{Adj}(\mathcal{B})_{=}$}}
\put(100,-25){\makebox(100,50){$\mathsf{Adj}(\mathcal{B})_{=}$}}
\put(-32,56){\makebox(100,50){\footnotesize{$\pi_0$}}}
\put(33,56){\makebox(100,50){\footnotesize{$\pi_1$}}}
\put(-80,8){\makebox(100,50){\footnotesize{$\partial_0$}}}
\put(-23,12){\makebox(100,50){\footnotesize{$\partial_1$}}}
\put(23,12){\makebox(100,50){\footnotesize{$\partial_0$}}}
\put(82,8){\makebox(100,50){\footnotesize{$\partial_1$}}}
\put(-50,-35){\makebox(100,50){\footnotesize{$\mathsf{id}$}}}
\put(50,-35){\makebox(100,50){\footnotesize{$\mathsf{id}$}}}
\put(-48,-4){\makebox(100,50){\footnotesize{$\mathsf{ref}_{\mathcal{B}}$}}}
\put(-48,-14){\makebox(100,50){\Large{$\Rightarrow$}}}
\put(52,-4){\makebox(100,50){\footnotesize{$\mathsf{ref}^{\propto}_{\mathcal{B}}$}}}
\put(52,-14){\makebox(100,50){\Large{$\Rightarrow$}}}
\put(50,25){\begin{picture}(30,15)(0,-15)
\put(0,0){\line(-1,-1){15}}
\put(0,0){\line(1,-1){15}}
\end{picture}}
\thicklines
\put(40,90){\vector(-1,-1){30}}
\put(60,90){\vector(1,-1){30}}
\put(-10,40){\vector(-1,-1){30}}
\put(10,40){\vector(1,-1){30}}
\put(-30,0){\vector(1,0){60}}
\put(-30,0){\vector(-1,0){0}}
\put(90,40){\vector(-1,-1){30}}
\put(110,40){\vector(1,-1){30}}
\put(70,0){\vector(1,0){60}}
\put(70,0){\vector(-1,0){0}}
\end{picture}}
\thicklines
\put(5,135){\makebox(100,50){$\mathsf{Clg}(\mathcal{B})^{\mathrm{op}}$}}
\put(-100,135){\makebox(100,50){$\mathsf{Clsn}(\mathcal{B})^{\mathrm{op}}$}}
\thicklines
\put(-112,107){\makebox(100,50){\footnotesize{$\mathsf{incl}$}}}
\put(12,107){\makebox(100,50){\footnotesize{$\mathsf{incl}$}}}
\put(50,150){\vector(0,-1){40}}
\put(-50,150){\vector(0,-1){40}}
\put(-21,168){\vector(1,0){40}}
\put(-46,153){\makebox(100,50){\footnotesize{$\mathsf{clg}_{\mathcal{B}}^{\mathrm{op}}$}}}
\put(-49,135.5){\makebox(100,50){\footnotesize{$\equiv$}}}
\put(-44,120){\makebox(100,50){\footnotesize{$\mathsf{clsn}_{\mathcal{B}}^{\mathrm{op}}$}}}
\put(19,154){\vector(-1,0){40}}
\put(-21,108){\vector(1,0){40}}
\put(-49,90){\makebox(100,50){\footnotesize{$\div$}}}
\put(-49,75.5){\makebox(100,50){\footnotesize{$\equiv$}}}
\put(-49,62){\makebox(100,50){\footnotesize{$\circ$}}}
\put(19,94){\vector(-1,0){40}}
\put(-35,90){\line(2,-1){44}}
\put(21,62){\line(2,-1){43}}
\put(79,33){\line(2,-1){15}}
\put(109,18){\vector(2,-1){28}}
\put(-50,85){\vector(0,-1){70}}
\put(-98,75){\makebox(100,50){$\mathsf{Adj}(\mathcal{B})_{=}^{\mathsf{2}}$}}
\put(-109,24){\makebox(100,50){\footnotesize{$\partial_0$}}}
\put(-7,36){\makebox(100,50){\footnotesize{$\partial_1$}}}
\put(-73,45){\makebox(100,50){\footnotesize{$\alpha_{\mathsf{Adj}(\mathcal{B})_{=}}$}}}
\put(-75,35){\makebox(100,50){\Large{$\Rightarrow$}}}
\end{picture}
\\ \\ \\ \\ \\ \\ \\ \\
\setlength{\unitlength}{0.8pt}
\begin{picture}(100,100)(-25,0)
%\qbezier[100](18,140)(28,37)(-92,32)
%\put(-160,60){\makebox(100,50){ {\itshape\large classification} }}
%\put(-160,50){\makebox(100,50){ {\itshape\large structures} }}
%\put(90,60){\makebox(100,50){ {\itshape\large conceptual} }}
%\put(90,50){\makebox(100,50){ {\itshape\large structures} }}
\put(0,0){\begin{picture}(100,100)(0,0)
\put(45,75){\makebox(100,50){$\mathsf{Clg}(\mathcal{B})^{\mathrm{op}} = \mathsf{Clg}(\mathcal{B})_{\iota}^{\mathrm{op}} {\odot}\, \mathsf{Clg}(\mathcal{B})_{\tau}^{\mathrm{op}}$}}
\put(-46,25){\makebox(100,50){$\mathsf{Clg}(\mathcal{B})_{\iota}^{\mathrm{op}}$}}
\put(54,25){\makebox(100,50){$\mathsf{Clg}(\mathcal{B})_{\tau}^{\mathrm{op}}$}}
\put(-100,-25){\makebox(100,50){$\mathcal{B}$}}
\put(4,-25){\makebox(100,50){$\mathsf{Adj}(\mathcal{B})_{=}$}}
\put(104,-25){\makebox(100,50){$\mathcal{B}^{\mathrm{op}}$}}
\put(-30,56){\makebox(100,50){\footnotesize{$\pi_0$}}}
\put(31,56){\makebox(100,50){\footnotesize{$\pi_1$}}}
\put(-82,8){\makebox(100,50){\footnotesize{$\mathsf{inst}_{\mathcal{B}}$}}}
\put(-23,12){\makebox(100,50){\footnotesize{$\mathsf{adj}_{\mathcal{B}}$}}}
\put(22,12){\makebox(100,50){\footnotesize{$\mathsf{adj}_{\mathcal{B}}$}}}
\put(87,8){\makebox(100,50){\footnotesize{$\mathsf{typ}^{\mathrm{op}}_{\mathcal{B}}$}}}
\put(-54,-33){\makebox(100,50){\footnotesize{$\mathsf{dir}_{\mathcal{B}}$}}}
\put(54,-33){\makebox(100,50){\footnotesize{$\mathsf{inv}_{\mathcal{B}}$}}}
\put(-48,-4){\makebox(100,50){\footnotesize{$\mathsf{extent}_{\mathcal{B}}$}}}
\put(-48,-14){\makebox(100,50){\Large{$\Rightarrow$}}}
\put(52,-4){\makebox(100,50){\footnotesize{$\mathsf{intent}_{\mathcal{B}}$}}}
\put(52,-14){\makebox(100,50){\Large{$\Rightarrow$}}}
\put(50,25){\begin{picture}(30,15)(0,-15)
\put(0,0){\line(-1,-1){15}}
\put(0,0){\line(1,-1){15}}
\end{picture}}
\thicklines
\put(40,90){\vector(-1,-1){30}}
\put(60,90){\vector(1,-1){30}}
\put(-10,40){\vector(-1,-1){30}}
\put(10,40){\vector(1,-1){30}}
\put(90,40){\vector(-1,-1){30}}
\put(110,40){\vector(1,-1){30}}
\put(-29,0){\vector(1,0){50}}
\put(129,0){\vector(-1,0){50}}
\end{picture}}
\thicklines
\put(-88,140){\vector(1,-1){30}}
\put(-95,135){\vector(1,-3){40}}
\put(-88,40){\vector(1,-1){30}}
\put(-101,135){\vector(0,-1){70}}
\qbezier(-79,147)(32,140)(48,112)
\put(49.5,109.50){\vector(2,-3){0}}
\qbezier(-78,155)(50,167)(104,110)
\qbezier(117,93)(135,65)(141,42.5)
\qbezier(144,30)(146,23)(147,13)
\put(147,10.5){\vector(0,-1){0}}
\put(-21,108){\vector(1,0){40}}
\put(-46,92){\makebox(100,50){\footnotesize{$\mathsf{clg}_{\mathcal{B}}^{\mathrm{op}}$}}}
\put(-51,75.5){\makebox(100,50){\footnotesize{$\equiv$}}}
\put(-44,62){\makebox(100,50){\footnotesize{$\mathsf{clsn}_{\mathcal{B}}^{\mathrm{op}}$}}}
\put(19,94){\vector(-1,0){40}}
\put(-35,90){\line(2,-1){44}}
\put(21,62){\line(2,-1){43}}
\put(79,33){\line(2,-1){15}}
\put(109,18){\vector(2,-1){28}}
\put(-50,85){\vector(0,-1){70}}
\put(-96,75){\makebox(100,50){$\mathsf{Clsn}(\mathcal{B})^{\mathrm{op}}$}}
\put(-146,125){\makebox(100,50){$\mathsf{Sign}_{\Im}^{\mathrm{op}}$}}
\put(-150,25){\makebox(100,50){$\mathsf{Cat}$}}
\put(-168,75){\makebox(100,50){\footnotesize{$\mathsf{mod}_{\Im}^{\mathrm{op}}$}}}
\put(45,110){\makebox(100,50){\footnotesize{$\mathsf{sen}_{\Im}^{\mathrm{op}}$}}}
\put(-53,120){\makebox(100,50){\footnotesize{$\mathsf{clg}_{\Im}^{\mathrm{op}}$}}}
\put(-102,100){\makebox(100,50){\footnotesize{$\mathsf{clsn}_{\Im}^{\mathrm{op}}$}}}
\put(-125,60){\makebox(100,50){\footnotesize{${|\mathsf{mod}|}_{\Im}^{\mathrm{op}}$}}}
\put(-132,-3){\makebox(100,50){\footnotesize{$|\mbox{-}|$}}}
\put(-88,25){\makebox(100,50){\footnotesize{$\mathsf{inst}_{\mathcal{B}}$}}}
\put(-3,36){\makebox(100,50){\footnotesize{$\mathsf{typ}_{\mathcal{B}}^{\mathrm{op}}$}}}
\put(-75,45){\makebox(100,50){\footnotesize{$\mathsf{deriv}_{\mathcal{B}}$}}}
\put(-75,35){\makebox(100,50){\Large{$\Rightarrow$}}}
\end{picture}
\\ \\
\end{tabular}
\end{center}
\caption{Equivalence between Classification and Conceptual Structures}
\label{conceptual-structure}
\end{figure}

%\include{clg}

%%%%%%%%%%%%%%%%%%%%%%%%%%%%%%%%%%%%%%%%%%%%%%%%%%%%%%%%%%%%%%%%%%%%%%%%%%%%%%%%
%%%%%%%%%%%%%%%%%%%%%%%%%%%%%%%%%%%%%%%%%%%%%%%%%%%%%%%%%%%%%%%%%%%%%%%%%%%%%%%%

%%%%%%%%%%%%%%%%%%%%%%%%%%%%%%%%%%%%%%%%%%%%%%%%%%%%%%%%%%%%%%%%%%%%%%%%%%%%%%%%
%%%%%%%%%%%%%%%%%%%%%%%%%%%%%%%%%%%%%%%%%%%%%%%%%%%%%%%%%%%%%%%%%%%%%%%%%%%%%%%%

%%%%%%%%%%%%%%%%%%%%%%%%%%%%%%%%%%%%%%%%%%%%%%%%%%%%%%%%%%%%%%%%%%%%%%
%%%%%%%%%%%%%%%%%%%%%%%%%%%%%%%%%%%%%%%%%%%%%%%%%%%%%%%%%%%%%%%%%%%%%%
\section{Institutions}\label{sec:institutions}
%%%%%%%%%%%%%%%%%%%%%%%%%%%%%%%%%%%%%%%%%%%%%%%%%%%%%%%%%%%%%%%%%%%%%%
%%%%%%%%%%%%%%%%%%%%%%%%%%%%%%%%%%%%%%%%%%%%%%%%%%%%%%%%%%%%%%%%%%%%%%

The theory of institutions is abstract model theory.
It abstracts and generalizes Tarski's ``semantic definition of truth''.
The central construct in the theory of institutions is 
the relation of satisfaction between models and sentences.
Examples of institutions include:
first order logic with first order structures as models,
many-sorted equational logic with abstract algebras as models,
Horn clause logic, and variants of higher order and modal logic.
Institutions are usually defined in classification structures style (relation version).
However,
based upon the results in this paper,
institutions can also be defined in 
classification structures style (function or adjunction version)
and
conceptual structures style (relation, function or adjunction version).
Here we describe the relation and adjunction versions of both styles.
The adjunction versions are simplest\footnotemark[12].

%%%%%%%%%%%%%%%%%%%%%%%%%%%%%%%%%%%%%%%%%%%%%%%%%%%%%%%%%%%%%%%%%%%%%%
%\subsection{Classification Structures Style}\label{subsec:classification:structures:style}
%%%%%%%%%%%%%%%%%%%%%%%%%%%%%%%%%%%%%%%%%%%%%%%%%%%%%%%%%%%%%%%%%%%%%%

{\footnotesize \begin{quotation}
\begin{sloppypar}
\noindent 
{\bfseries Class-Rel Style:}
An \emph{institution}\footnote{A more refined definition of institution would indicate whether collections of models and sentences are set-theoretically ``small'' or ``large''. In the presentation given here we ignore size considerations, assuming that truth factors for both sizes.} 
$\Im = \langle \mathsf{Sign}_{\Im}, \mathsf{mod}_{\Im}, \mathsf{sen}_{\Im}, \models_{\Im} \rangle$
\cite{goguen:burstall:92}
in (internal to) a topos $\mathcal{B}$ 
has
an abstract category $\mathsf{Sign}_{\Im}$ of signatures $\Sigma$,
a model fiber (reduct) functor $\mathsf{mod}_{\Im} : \mathsf{Sign}_{\Im} \rightarrow \mathsf{Cat}(\mathcal{B})^{\mathsf{op}}$
indexing abstract models $\mathsf{mod}_{\Im}(\Sigma)$ by signatures $\Sigma$, 
a sentence fiber functor $\mathsf{sen}_{\Im} : \mathsf{Sign}_{\Im} \rightarrow \mathcal{B}$
indexing abstract sentences $\mathsf{sen}_{\Im}(\Sigma)$ by signatures $\Sigma$,
and
a function $\models_{\Im} : {|\mathsf{Sign}_{\Im}|}_{\mathcal{B}} \rightarrow \mathsf{Rel}(\mathcal{B})$
indexing abstract satisfaction relations
$\models_{\Im, \Sigma} : |\mathsf{mod}|_{\Im}(\Sigma) \rightharpoondown \mathsf{sen}_{\Im}(\Sigma)$ by signatures $\Sigma$.
An institution must satisfy the \emph{satisfaction condition},
${\models}_{\Im, \Sigma_{1}}\left( |\mathsf{mod}|_{\Im}(\sigma), 1_{\mathsf{sen}_{\Im}(\Sigma_{1})} \right)
= {\models}_{\Im, \Sigma_{2}}\left( 1_{|\mathsf{mod}|_{\Im}(\Sigma_{2})}, \mathsf{sen}_{\Im}(\sigma) \right)\footnote{Satisfaction does not use morphisms in $\mathsf{mod}_{\Im}(\Sigma)$,
and hence is expressed in terms of the underlying model functor
$|\mathsf{mod}|_{\Im}
= \mathsf{mod}_{\Im} \circ {\scriptstyle |{-}|}_{\mathcal{B}}^{\mathsf{op}} : \mathsf{Sign}_{\Im} \rightarrow \mathcal{B}^{\mathsf{op}}$.}$
for any signature morphism $\sigma : \Sigma_{1} \rightarrow \Sigma_{2}$,
which expresses the invariance of truth under change of notation. 
\end{sloppypar}
\end{quotation}}

\noindent The components of an institution $\Im$,
can be packed together as a classification functor
$\mathsf{clsn}_{\Im} : \mathsf{Sign}_{\Im} \rightarrow \mathsf{Clsn}(\mathcal{B})$,
where for every signature $\Sigma$ 
the satisfaction relation forms the classification
$\mathsf{clsn}_{\Im}(\Sigma) = \langle |\mathsf{mod}|_{\Im}(\Sigma), \mathsf{sen}_{\Im}(\Sigma), \models_{\Im, \Sigma} \rangle$
and for every signature morphism $\sigma : \Sigma_{1} \rightarrow \Sigma_{2}$ 
the satisfaction condition states the fundamental condition for the infomorphism
$\mathsf{clsn}_{\Im}(\sigma) = \langle |\mathsf{mod}|_{\Im}(\sigma), \mathsf{sen}_{\Im}(\sigma) \rangle : \mathsf{clsn}_{\Im}(\Sigma_{1})\rightleftharpoons\mathsf{clsn}_{\Im}(\Sigma_{2})$.

{\footnotesize \begin{quotation}
\begin{sloppypar}
\noindent 
{\bfseries Class-Adj Style:}
An institution 
$\Im = \langle \mathsf{Sign}_{\Im}, \mathsf{mod}_{\Im}, \mathsf{sen}_{\Im}, \mathsf{deriv}_{\Im} \rangle$
in (internal to) a topos $\mathcal{B}$
has components 
$\mathsf{Sign}_{\Im}$, $\mathsf{mod}_{\Im}$ and $\mathsf{sen}_{\Im}$ 
as above,
plus a function 
$\mathsf{deriv}_{\Im} : |\mathsf{Sign}_{\Im}| \rightarrow \mathsf{Adj}(\mathcal{B})$
indexing abstract derivation Galois connections
$\mathsf{deriv}_{\Im, \Sigma} 
: {\wp}\,|\mathsf{mod}|_{\Im}(\Sigma) \rightleftharpoons {\wp}\,\mathsf{sen}_{\Im}(\Sigma)^{\mathrm{op}}$ by signatures $\Sigma$.
An institution must satisfy the \emph{derivation condition},
$\mathsf{dir}(|\mathsf{mod}|_{\Im}(\sigma)) \circ \mathsf{deriv}_{\Im, \Sigma_1}
= \mathsf{deriv}_{\Im, \Sigma_2} \circ \mathsf{inv}(\mathsf{sen}_{\Im}(\sigma))$,
for any signature morphism $\sigma : \Sigma_{1} \rightarrow \Sigma_{2}$.
\end{sloppypar}
\end{quotation}}

\noindent By Thm.~\ref{restricted-equivalence},
the category of classification structures  
is equivalent to 
the category of conceptual structures $\mathsf{Clsn}(\mathcal{B}) \cong \mathsf{Clg}(\mathcal{B})$.
In what follows,
we choose mediation by the open polar factorization functor
$\mathsf{clg}_{\mathcal{B}}^{\circ} 
: \mathsf{Clsn}(\mathcal{B}) \rightarrow \mathsf{Clg}(\mathcal{B})$.
An alternate expression for an institution
is a concept lattice functor 
$\mathsf{clg}_{\Im}^{\circ} 
= \mathsf{clg}_{\Im} \circ \mathsf{clg}^{\circ}
: \mathsf{Sign}_{\Im} \rightarrow \mathsf{Clg}(\mathcal{B})$.

{\footnotesize \begin{quotation}
\noindent 
{\bfseries Conc-Adj Style:} 
An institution 
$\Im = \langle \mathsf{Sign}_{\Im}, \mathsf{mod}_{\Im}, \mathsf{sen}_{\Im}, \mathsf{th}_{\Im}, \mathsf{extent}_{\Im}, \\ \mathsf{intent}_{\Im} \rangle$
in (internal to) a topos $\mathcal{B}$
has components 
$\mathsf{Sign}_{\Im}$, $\mathsf{mod}_{\Im}$ and $\mathsf{sen}_{\Im}$ 
as above,
plus 
a theory fiber functor $\mathsf{th}_{\Im} : \mathsf{Sign}_{\Im} \rightarrow \mathsf{Adj}(\mathcal{B})^{\mathsf{op}}$
indexing abstract concept lattices of theories 
$\mathsf{th}_{\Im}(\Sigma)$ by signatures $\Sigma$
and
indexing concept lattice of theories adjunctions
$\mathsf{th}_{\Im, \sigma}
: \mathsf{th}_{\Im}(\Sigma_{2}) \rightleftharpoons \mathsf{th}_{\Im}(\Sigma_{1})$
by signature morphisms $\sigma : \Sigma_1 \rightarrow \Sigma_2$,
a function 
$\mathsf{extent}_{\Im} : |\mathsf{Sign}_{\Im}| \rightarrow \mathsf{Ref}(\mathcal{B})$
indexing abstract extent reflections
$\mathsf{extent}_{\Im}(\Sigma) 
: {\wp}\,|\mathsf{mod}|_{\Im}(\Sigma) \rightleftharpoons \mathsf{th}_{\Im}(\Sigma)$ 
by signatures $\Sigma$,
and
a function 
$\mathsf{intent}_{\Im} : |\mathsf{Sign}_{\Im}| \rightarrow \mathsf{Ref}(\mathcal{B})^{\propto}$
indexing abstract intent coreflections
$\mathsf{intent}_{\Im}(\Sigma) 
: \mathsf{th}_{\Im}(\Sigma) \rightleftharpoons {\wp}\,\mathsf{sen}_{\Im}(\Sigma)^{\mathrm{op}}$ 
by signatures $\Sigma$.
An institution must satisfy 
the \emph{extent condition},
$\mathsf{extent}_{\Sigma_2} \circ \mathsf{th}_{\Im}(\sigma)
= \mathsf{dir}(|\mathsf{mod}|_{\Im}(\sigma)) \circ \mathsf{extent}_{\Sigma_1}$,
and
the \emph{intent condition},
$\mathsf{th}_{\Im}(\sigma) \circ \mathsf{intent}_{\Sigma_1}
= \mathsf{intent}_{\Sigma_2} \circ \mathsf{inv}(\mathsf{sen}_{\Im}(\sigma))$,
for any signature morphism $\sigma : \Sigma_1 \rightarrow \Sigma_2$.
\end{quotation}}

{\footnotesize \begin{quotation}
\begin{sloppypar}
\noindent 
{\bfseries Conc-Rel Style:}
An institution
$\Im = \langle \mathsf{Sign}_{\Im}, \mathsf{mod}_{\Im}, \mathsf{sen}_{\Im}, \mathsf{th}_{\Im}, \models_{\Im} \rangle$
in (internal to) a topos $\mathcal{B}$
has components 
$\mathsf{Sign}_{\Im}$, $\mathsf{mod}_{\Im}$, $\mathsf{sen}_{\Im}$ and $\mathsf{th}_{\Im}$ 
as above,
plus 
a function $\models_{\Im} : |\mathsf{Sign}_{\Im}| \rightarrow \mathsf{Rel}(\mathcal{B})$
indexing abstract satisfaction relations
${\models}_{\Im, \Sigma} 
: |\mathsf{mod}|_{\Im}(\Sigma) \rightharpoondown \mathsf{th}_{\Im}(\Sigma)$ 
by signatures $\Sigma$.
These satisfaction relations are a special case of the open instance-of relations.
The dual open of-type relations are trivial --- being reverse membership relations for theories; their constraining conditions are also trivial, being the definition of the inverse image function on theories.
An institution must satisfy 
the \emph{satisfaction condition},
${\models}_{\Im, \Sigma_1}
\left( |\mathsf{mod}|_{\Im}(\sigma), 1_{\mathsf{th}_{\Im}(\Sigma_{1})} \right)
= {\models}_{\Im, \Sigma_2} 
\left( 1_{|\mathsf{mod}|_{\Im}(\Sigma_{2})}, 
{(\mbox{-})}^{\bullet} \cdot \exists\mathsf{sen}_{\Im}(\sigma) \right)$
for any signature morphism $\sigma : \Sigma_{1} \rightarrow \Sigma_{2}$.
\end{sloppypar}
\end{quotation}}

\noindent 
The Grothendieck construction, 
appied to the existential quantification theory fiber functor
$\exists_{\Im}
= \mathsf{th}_{\Im} \circ \mathsf{right}
: \mathsf{Sign}_{\Im} \rightarrow \mathsf{Adj}(\mathcal{B})^{\mathsf{op}} \rightarrow \mathsf{Ord}(\mathcal{B})$,
amalgamates the lattice fibers $\mathsf{th}_{\Im}(\Sigma)$,
producing a flattened category of theories
$\mathsf{Theory}_{\Im}^{\exists} = {\mathsf{Gr}(\exists_{\Im})}^{\mathsf{op}}$
that represents the cLOT construction.
If $\mathsf{Sign}_{\Im}$ is cocomplete,
then $\mathsf{Theory}_{\Im}^{\exists}$ is cocomplete,
and the semantic integration of ontologies is represented by the colimit construction in $\mathsf{Theory}_{\Im}^{\exists}$.

%\include{ins}

%%%%%%%%%%%%%%%%%%%%%%%%%%%%%%%%%%%%%%%%%%%%%%%%%%%%%%%%%%%%%%%%%%%%%%%%%%%%%%%%
%%%%%%%%%%%%%%%%%%%%%%%%%%%%%%%%%%%%%%%%%%%%%%%%%%%%%%%%%%%%%%%%%%%%%%%%%%%%%%%%

%%%%%%%%%%%%%%%%%%%%%%%%%%%%%%%%%%%%%%%%%%%%%%%%%%%%%%%%%%%%%%%%%%%%%%
%%%%%%%%%%%%%%%%%%%%%%%%%%%%%%%%%%%%%%%%%%%%%%%%%%%%%%%%%%%%%%%%%%%%%%
\section{Summary}\label{sec:summary}
%%%%%%%%%%%%%%%%%%%%%%%%%%%%%%%%%%%%%%%%%%%%%%%%%%%%%%%%%%%%%%%%%%%%%%
%%%%%%%%%%%%%%%%%%%%%%%%%%%%%%%%%%%%%%%%%%%%%%%%%%%%%%%%%%%%%%%%%%%%%%

In this paper, 
we proved a general equivalence theorem for categories having a factorization system with choice.
We applied this to the polar factorization of adjunctions between posets,
producing a special equivalence.
We demonstrated 
how classification and conceptual structures form a restricted equivalence
mediated by polar factorization and composition
--- truth factors in terms of, and is the composition of, extension and intension.
Finally,
we applied this restricted equivalence 
to define cLOT and get alternate definitions for institutions.

\begin{figure}
\begin{center}
{\scriptsize \begin{tabular}{c@{\hspace{20pt}}l}

\begin{tabular}[t]{c}
{\sffamily object} \\ $A$
\end{tabular}
&
\begin{tabular}[t]{l@{\hspace{10pt}}l}
$\gamma_{A,B} : B \rightarrow {(A{\times}B)}^A$
& \emph{constant augmentation} morphism for object $B$ \\
$\varepsilon_{A,C} : A{\times}{C^A} \rightarrow C$
& \emph{evaluation} morphism for object $C$ \\
$\chi_{m} : A \rightarrow \Omega$ & \emph{character} morphism for subobject $m$ \\
$\iota_{m} : \Box_{m} \hookrightarrow A$ 
& \emph{representative monic} morphism for subobject $m$ \\
$\top_{A} : A \hookrightarrow \Omega$  & \emph{top} character morphism \\
$\perp_{A} : A \hookrightarrow \Omega$ & \emph{bottom} character morphism \\
$\Delta_A : A \hookrightarrow A{\times}A$     & \emph{delta} morphism \\
${\{\mbox{-}\}}_A : A \hookrightarrow {\wp}A$ & \emph{singleton} morphism \\
$1_A : A \rightharpoondown A$          & \emph{identity} relation \\
${\in}_A : A \rightharpoondown {\wp}A$ & \emph{membership} relation
\end{tabular}

\\ & \\

\begin{tabular}[t]{c}
{\sffamily morphism} \\ 
$f : A \rightarrow B$
\end{tabular}
&
\begin{tabular}[t]{l@{\hspace{10pt}}l}
${\exists} f : {\wp}A \rightarrow {\wp}B$, 
& \emph{existential direct image} monotonic morphism \\
$f^{-1} : {\wp}B \rightarrow {\wp}A$, 
& \emph{inverse image} monotonic morphism \\
${\forall} f : {\wp}A \rightarrow {\wp}B$, 
& \emph{universal direct image} monotonic function
\end{tabular}

\\ & \\

\begin{tabular}[t]{c}
{\sffamily relation} \\ 
$r : A \rightharpoondown B$
\end{tabular}
&
\begin{tabular}[t]{l@{\hspace{10pt}}l}
$r^{01} : A \rightarrow {\wp}B$, 
& \emph{01 fiber} function \\
$r^{10} : B \rightarrow {\wp}A$, 
& \emph{10 fiber} function \\
$r^{\Rightarrow} : {\wp}A^{\propto} \rightarrow {\wp}B$, 
& \emph{forward derivation} monotonic function \\
$r^{\Leftarrow} : {\wp}B^{\propto} \rightarrow {\wp}A$, 
& \emph{reverse derivation} monotonic function \\
$r{\setminus}s : B \rightharpoondown C$,
& \emph{left residuation} relation of $s : A \rightharpoondown C$ \\
$r{/}s : C \rightharpoondown A$,
& \emph{right residuation} relation of $s : C \rightharpoondown B$ 
\end{tabular}

\\ & \\

\begin{tabular}[t]{c}
{\sffamily preorder} \\ 
$\mathbf{A} = \langle A, {\leq}_{\mathbf{A}} \rangle$
\end{tabular}
&
\begin{tabular}[t]{l@{\hspace{10pt}}l}
${\leq}_{\mathbf{A}} : A \rightarrow A$
& \emph{order} relation \\
${\uparrow}_{\mathbf{A}} : \mathbf{A}^{\propto} \rightarrow {\wp}\mathbf{A}$
& \emph{up segment} contravariant monotonic function \\
${\downarrow}_{\mathbf{A}} : \mathbf{A} \rightarrow {\wp}\mathbf{A}$
& \emph{down segment} covariant monotonic function \\
${\Uparrow}_{\mathbf{A}} : {\wp}\mathbf{A}^{\propto} \rightarrow {\wp}\mathbf{A}$
& \emph{upper bound} monotonic function \\
${\Downarrow}_{\mathbf{A}} : {\wp}\mathbf{A}^{\propto} \rightarrow {\wp}\mathbf{A}$
& \emph{lower bound} monotonic function
\end{tabular}

\\ & \\

\begin{tabular}[t]{c}
{\sffamily monotonic morphism} \\ 
${\mathbf f} : {\mathbf A} \rightarrow {\mathbf B}$
\end{tabular}
&
\begin{tabular}[t]{l@{\hspace{10pt}}l}
${\mathbf f}^{\triangleright} : {\mathbf A} \rightharpoondown {\mathbf B}$, 
& \emph{forward} bimodule \\
${\mathbf f}^{\triangleleft} : {\mathbf B} \rightharpoondown {\mathbf A}$, 
& \emph{reverse} bimodule
\end{tabular}

\\ & \\

\begin{tabular}[t]{c}
{\sffamily bimodule} \\ 
$\mathbf{r} : \mathbf{A} \rightharpoondown \mathbf{B}$
\end{tabular}
&
\begin{tabular}[t]{l@{\hspace{10pt}}l}
$\mathbf{r}^{01} : \mathbf{A} \rightarrow {\wp}\mathbf{B}$, 
& \emph{01 fiber} monotonic morphism \\
$\mathbf{r}^{10} : \mathbf{B} \rightarrow {\wp}\mathbf{A}$, 
& \emph{10 fiber} monotonic morphism \\
\multicolumn{2}{l}{between complete lattices} \\
$\mathbf{r}^{\wedge} : \mathbf{A} \rightarrow \mathbf{B}$
& \emph{01-embedding} monotonic function \\
$\mathbf{r}^{\vee} : \mathbf{B} \rightarrow \mathbf{A}$
& \emph{10-embedding} monotonic function
\end{tabular}

\\ & \\

\begin{tabular}[t]{c}
{\sffamily complete lattice} \\
$\mathbf{L} = \langle L, {\leq}_{\mathbf{L}}, {\vee}_{\mathbf{L}}, {\wedge}_{\mathbf{L}} \rangle$
\end{tabular}
&
\begin{tabular}[t]{l@{\hspace{10pt}}l}
${\vee}_{\mathbf{L}} : {\wp}\mathbf{L} \rightarrow \mathbf{L}$,
& \emph{join} monotonic function \\
${\wedge}_{\mathbf{L}} : {\wp}\mathbf{L}^{\propto} \rightarrow \mathbf{L}$,
& \emph{meet} monotonic function
\end{tabular}

\\ & \\ \hline & \\

{\sffamily Identities:}
&
\begin{tabular}[t]{c@{\hspace{20pt}}c@{\hspace{20pt}}c}
$\begin{array}[t]{r@{\hspace{5pt}=\hspace{5pt}}l}
{\Uparrow}_{\mathbf{A}}   & {\leq}_{\mathbf{A}}^{\Rightarrow} \\
{\Downarrow}_{\mathbf{A}} & {\leq}_{\mathbf{A}}^{\Leftarrow}
\end{array}$ 
&
$\begin{array}[t]{r@{\hspace{5pt}=\hspace{5pt}}l}
\mathbf{r}^{\wedge} & \mathbf{r}^{01} \cdot {\wedge}_{\mathbf{B}} \\ 
\mathbf{r}^{\vee}   & \mathbf{r}^{10} \cdot {\vee}_{\mathbf{A}}
\end{array}$ 
&
$\begin{array}[t]{r@{\hspace{5pt}=\hspace{5pt}}l}
(\mathbf{f}^{\triangleright})^{\vee}  & \mathbf{f} \\
(\mathbf{f}^{\triangleleft})^{\wedge} & \mathbf{f}
\end{array}$ 

\\ && \\

$\begin{array}[t]{r@{\hspace{5pt}=\hspace{5pt}}l}
(\mathbf{f}^\triangleright)^{01} & \mathbf{f} \cdot {\uparrow}_{\mathbf{B}} \\
(\mathbf{f}^\triangleright)^{10} & {\downarrow}_{\mathbf{B}} \cdot \mathbf{f}^{-1} \\
(\mathbf{f}^\triangleleft)^{01}  & {\uparrow}_{\mathbf{A}} \cdot \mathbf{f}^{-1} \\
(\mathbf{f}^\triangleleft)^{10}  & \mathbf{f} \cdot {\downarrow}_{\mathbf{A}}
\end{array}$ 

& 

$\begin{array}[t]{r@{\hspace{5pt}=\hspace{5pt}}l}
{\uparrow}_{\mathbf{B}} \cdot {\vee}_{\mathbf{B}}     & 1_{\mathbf{L}} \\
{\downarrow}_{\mathbf{B}} \cdot {\wedge}_{\mathbf{B}} & 1_{\mathbf{L}}
\end{array}$ 

& 

$\begin{array}[t]{r@{\hspace{5pt}=\hspace{5pt}}l}
r^{\Rightarrow} & \exists r^{01} \cdot {\cap}_{B} \\
r^{\Leftarrow} & \exists r^{10} \cdot {\cap}_{A}
\end{array}$ 

\end{tabular}
\end{tabular}}
\end{center}
\caption{Notation}
\label{notation}
\end{figure}

%\bibliographystyle{plain}
%\bibliography{kent}

\end{document}